\documentclass[12pt,oneside,english]{amsart}
\usepackage[T1]{fontenc}
\usepackage[latin9]{inputenc}
\usepackage{geometry}
\usepackage{amstext}
\usepackage{amsthm}
\usepackage{amssymb}
\usepackage{stmaryrd}
\usepackage{graphicx}
\usepackage{mathtools}
\usepackage{wasysym}
\usepackage{mathrsfs}
\usepackage{bbm}
\usepackage{esint}
\usepackage{color}
\usepackage{verbatim}
\usepackage{subfigure}
\usepackage{url}
\usepackage[percent]{overpic}

\makeatletter

\newcommand{\sertit}{Slit-strip Ising boundary conformal field theory}
\newcommand{\shorttit}{Slit-strip Ising BCFT}
\newcommand{\partnum}{1}
\newcommand{\parttit}{Discrete and continuous function spaces}

\numberwithin{equation}{section}
\numberwithin{figure}{section}
\theoremstyle{plain}
\newtheorem{thm}{\protect\theoremname}
\newtheorem*{thm*}{\protect\theoremname}
\newtheorem{prop}[thm]{\protect\propositionname}
\newtheorem*{prop*}{\protect\propositionname}
\newtheorem{lem}[thm]{\protect\lemmaname}
\newtheorem{cor}[thm]{\protect\corollaryname}
\newtheorem{rmk}[thm]{\protect\remarkname}

\numberwithin{thm}{section}

\definecolor{kallecol}{rgb}{.99,.1,.5}
\definecolor{davidcol}{rgb}{.5,.1,.99}
\definecolor{sketchcol}{rgb}{.4,.4,.8}
\definecolor{outlinecol}{rgb}{.8,.4,.3}

\newcommand{\sketch}[1]{{\color{sketchcol}{#1}}}

\newcommand{\term}[1]{{\bf #1}} 

\newcommand{\lftsym}{\mathrm{L}}
\newcommand{\rgtsym}{\mathrm{R}}
\newcommand{\topsym}{\mathrm{T}}

\newcommand{\wildsym}{\star}
\newcommand{\wildsymbis}{\star'}

\newcommand{\ii}{\mathbbm{i}}
\newcommand{\re}{\Re\mathfrak{e}}
\newcommand{\im}{\Im\mathfrak{m}}

\newcommand{\eps}{\varepsilon}
\newcommand{\bdry}{\partial}
\newcommand{\domain}{\Omega}
\newcommand{\tccw}{\tau}

\newcommand{\innprod}[2]{\langle #1 , #2 \rangle }

\newcommand{\projtoline}[1]{\mathrm{pr}_{#1}}

\newcommand{\half}{\frac{1}{2}}
\newcommand{\mhalf}{\frac{-1}{2}}
\newcommand{\phalf}{\frac{+1}{2}}

\newcommand{\dist}{\mathrm{dist}}

\newcommand{\const}{\mathrm{const.}}

\newcommand{\aaa}{-\frac{1}{2}}
\newcommand{\bbb}{\frac{1}{2}}

\newcommand{\OO}{\mathcal{O}}

\newcommand{\Est}{\mathrm{E}}
\newcommand{\Wst}{\mathrm{W}}
\newcommand{\Nth}{\mathrm{N}}
\newcommand{\Sth}{\mathrm{S}}

\newcommand{\eighthroot}{\lambda}
\newcommand{\eighthrootbar}{\lambda^{-1}}
\newcommand{\eighthrootthree}{\lambda^3}
\newcommand{\eighthrootfive}{\lambda^{-3}}

\newcommand{\dmn}{\mathrm{dim}}
\newcommand{\spn}{\mathrm{span}}

\newcommand{\C}{\mathbb{C}} 
\newcommand{\R}{\mathbb{R}} 
\newcommand{\Z}{\mathbb{Z}} 
\newcommand{\N}{\mathbb{N}} 
\newcommand{\UHP}{\mathbb{H}} 
\newcommand{\bC}{\C} 
\newcommand{\bR}{\R} 
\newcommand{\bZ}{\Z} 
\newcommand{\bN}{\N} 
\newcommand{\bH}{\UHP} 

\newcommand{\ud}{\mathrm{d}} 

\newcommand{\width}{\ell}
\newcommand{\widthL}{\ell^\lftsym}
\newcommand{\widthR}{\ell^\rgtsym}

\newcommand{\cstrip}{{\mathbb{S}}}

\newcommand{\dstrip}{{\cstrip^{(\width)}}}

\newcommand{\slit}{\mathrm{slit}}
\newcommand{\cslitstrip}{\cstrip_{\slit}}
\newcommand{\dslitstrip}{{\cstrip_{\slit}^{(\width)}}}
\newcommand{\dslitstripW}[1]{{\cstrip_{\slit}^{(#1)}}}

\newcommand{\dslitstripT}{{\cstrip_{\slit}^{\topsym;(\width)}}}
\newcommand{\dslitstripL}{{\cstrip_{\slit}^{\lftsym;(\width)}}}
\newcommand{\dslitstripR}{{\cstrip_{\slit}^{\rgtsym;(\width)}}}

\newcommand{\dE}{\mathrm{E}}
\newcommand{\dV}{\mathrm{V}}

\newcommand{\poshalfint}{\mathcal{K}}
\newcommand{\dposhalfint}{\mathcal{K}^{(\width)}}
\newcommand{\dposhalfintW}[1]{\mathcal{K}^{(#1)}}
\newcommand{\dposhalfintL}{\dposhalfintW{\widthL}}
\newcommand{\dposhalfintR}{\dposhalfintW{\widthR}}

\newcommand{\ccffun}{e}
\newcommand{\ccffunL}{\ccffun^{\lftsym}}
\newcommand{\ccffunR}{\ccffun^{\rgtsym}}
\newcommand{\ccffunT}{\ccffun^{\topsym}}
\newcommand{\ccffunX}{\ccffun^{{\wildsym}}}
\newcommand{\ccffunXbis}{\ccffun^{{\wildsymbis}}}

\newcommand{\ccfFun}{E}
\newcommand{\ccfFunL}{\ccfFun^{\lftsym}}

\newcommand{\set}[1]{\left\{ #1 \right\}}

\newcommand{\lft}{a}
\newcommand{\rgt}{b}
\newcommand{\mdpt}{0}
\newcommand{\lftd}{\lft'}
\newcommand{\rgtd}{\rgt'}
\newcommand{\crosssec}{I}

\newcommand{\crosssecdual}{\crosssec^*}
\newcommand{\crosssecLdual}{\crosssec_{\lftsym}^*}
\newcommand{\crosssecRdual}{\crosssec_{\rgtsym}^*}
\newcommand{\dinterval}[2]{\llbracket #1 , #2 \rrbracket}
\newcommand{\dintervaldual}[2]{\llbracket #1 , #2 \rrbracket^*}
\newcommand{\ccrosssec}{\mathcal{I}}

\newcommand{\geneigval}{\boldsymbol{\mu}}

\newcommand{\vaceigvalLR}%
	{\geneigval^{(\widthL,\widthR)}_{\emptyset,\emptyset}}

\newcommand{\parts}{\boldsymbol{\alpha}}

\newcommand{\pDeg}[1]{\boldsymbol{d}(#1)}
\newcommand{\pLen}[1]{\boldsymbol{s}(#1)}

\newcommand{\poleToMixSym}{B}
\newcommand{\mixToPoleSym}{A}
\newcommand{\poleToMixT}[2]{\poleToMixSym^\topsym_{#1 , #2}}
\newcommand{\poleToMixL}[2]{\poleToMixSym^\lftsym_{#1 , #2}}
\newcommand{\poleToMixR}[2]{\poleToMixSym^\rgtsym_{#1 , #2}}
\newcommand{\mixToPoleT}[2]{\mixToPoleSym^\topsym_{#1 , #2}}
\newcommand{\mixToPoleL}[2]{\mixToPoleSym^\lftsym_{#1 , #2}}
\newcommand{\mixToPoleR}[2]{\mixToPoleSym^\rgtsym_{#1 , #2}}

\newcommand{\functionsp}{\mathscr{F}}

\newcommand{\cfunctionsp}{\mathscr{L}^2}

\newcommand{\cfspTzero}{\cfunctionsp_{\topsym;\mathrm{zero}}}
\newcommand{\cfspTpole}{\cfunctionsp_{\topsym;\mathrm{pole}}}
\newcommand{\cfspLzero}{\cfunctionsp_{\lftsym;\mathrm{zero}}}
\newcommand{\cfspLpole}{\cfunctionsp_{\lftsym;\mathrm{pole}}}
\newcommand{\cfspRzero}{\cfunctionsp_{\rgtsym;\mathrm{zero}}}
\newcommand{\cfspRpole}{\cfunctionsp_{\rgtsym;\mathrm{pole}}}
\newcommand{\cprTzero}{\Pi_{\topsym;\mathrm{zero}}}
\newcommand{\cprTpole}{\Pi_{\topsym;\mathrm{pole}}}
\newcommand{\cprLzero}{\Pi_{\lftsym;\mathrm{zero}}}
\newcommand{\cprLpole}{\Pi_{\lftsym;\mathrm{pole}}}
\newcommand{\cprRzero}{\Pi_{\rgtsym;\mathrm{zero}}}
\newcommand{\cprRpole}{\Pi_{\rgtsym;\mathrm{pole}}}
\newcommand{\dprTzero}{\Pi_{\topsym;\mathrm{zero}}^{(\width)}}
\newcommand{\dprTpole}{\Pi_{\topsym;\mathrm{pole}}^{(\width)}}
\newcommand{\dprLzero}{\Pi_{\lftsym;\mathrm{zero}}^{(\width)}}
\newcommand{\dprLpole}{\Pi_{\lftsym;\mathrm{pole}}^{(\width)}}
\newcommand{\dprRzero}{\Pi_{\rgtsym;\mathrm{zero}}^{(\width)}}
\newcommand{\dprRpole}{\Pi_{\rgtsym;\mathrm{pole}}^{(\width)}}
\newcommand{\deigvalsym}{\Lambda}
\newcommand{\ceigvalsym}{\Lambda}

\newcommand{\eigval}[1]{\deigvalsym_{#1}^{(\width)}}
\newcommand{\eigvalW}[2]{\deigvalsym_{#2}^{(#1)}}

\newcommand{\eigf}[1]{\mathfrak{f}_{#1}}
\newcommand{\eigfL}[1]{\mathfrak{f}_{\lftsym;#1}}
\newcommand{\eigfR}[1]{\mathfrak{f}_{\rgtsym;#1}}
\newcommand{\eigfT}[1]{\mathfrak{f}_{\topsym;#1}}
\newcommand{\eigfX}[1]{\mathfrak{f}_{{\wildsym};#1}}
\newcommand{\eigfXbis}[1]{\mathfrak{f}_{{\wildsymbis};#1}}
\newcommand{\eigF}[1]{\mathfrak{F}_{#1}}
\newcommand{\eigFL}[1]{\mathfrak{F}_{\lftsym;#1}}
\newcommand{\eigFR}[1]{\mathfrak{F}_{\rgtsym;#1}}
\newcommand{\poleL}[1]{\mathfrak{p}_{\lftsym;#1}}
\newcommand{\poleR}[1]{\mathfrak{p}_{\rgtsym;#1}}
\newcommand{\poleT}[1]{\mathfrak{p}_{\topsym;#1}}
\newcommand{\poleX}[1]{\mathfrak{p}_{{\wildsym};#1}}
\newcommand{\poleXbis}[1]{\mathfrak{p}_{{\wildsymbis};#1}}
\newcommand{\PoleL}[1]{\mathfrak{P}_{\lftsym;#1}}
\newcommand{\PoleR}[1]{\mathfrak{P}_{\rgtsym;#1}}
\newcommand{\PoleT}[1]{\mathfrak{P}_{\topsym;#1}}
\newcommand{\cmixPoleL}[1]{\tilde{{P}}^{\lftsym}_{#1}}
\newcommand{\cmixPoleR}[1]{\tilde{{P}}^{\rgtsym}_{#1}}
\newcommand{\cmixPoleT}[1]{\tilde{{P}}^{\topsym}_{#1}}
\newcommand{\cpoleL}[1]{{p}^{\lftsym}_{#1}}
\newcommand{\cpoleR}[1]{{p}^{\rgtsym}_{#1}}
\newcommand{\cpoleT}[1]{{p}^{\topsym}_{#1}}
\newcommand{\cpoleX}[1]{{p}^{{\wildsym}}_{#1}}
\newcommand{\cpoleXbis}[1]{{p}^{{\wildsymbis}}_{#1}}
\newcommand{\cPoleL}[1]{{P}^{\lftsym}_{#1}}
\newcommand{\cPoleR}[1]{{P}^{\rgtsym}_{#1}}
\newcommand{\cPoleT}[1]{{P}^{\topsym}_{#1}}

\newcommand{\dfunctionsp}{\functionsp^{(\width)}}
\newcommand{\dfunctionspW}[1]{\functionsp^{(#1)}}
\newcommand{\dfunctionspL}{\functionsp^{(\width)}_{\lftsym}}
\newcommand{\dfunctionspR}{\functionsp^{(\width)}_{\rgtsym}}
\newcommand{\dfspTzero}{\functionsp_{\topsym;\mathrm{zero}}^{(\width)}}
\newcommand{\dfspTpole}{\functionsp_{\topsym;\mathrm{pole}}^{(\width)}}
\newcommand{\dfspLzero}{\functionsp_{\lftsym;\mathrm{zero}}^{(\width)}}
\newcommand{\dfspLpole}{\functionsp_{\lftsym;\mathrm{pole}}^{(\width)}}
\newcommand{\dfspRzero}{\functionsp_{\rgtsym;\mathrm{zero}}^{(\width)}}
\newcommand{\dfspRpole}{\functionsp_{\rgtsym;\mathrm{pole}}^{(\width)}}

\newcommand{\refl}{\mathsf{R}}
\newcommand{\propag}{\mathsf{A}}

\newcommand{\iis}[1]{\im \left(\int_\times {#1}^2 \right)}

\newcommand{\wgt}{\mathbf{w}}

\newcommand{\confmap}{\varphi}
\newcommand{\mapSS}{\confmap}

\usepackage{babel}

  \providecommand{\corollaryname}{Corollary}
  \providecommand{\definitionname}{Definition}
  \providecommand{\lemmaname}{Lemma}
  \providecommand{\propositionname}{Proposition}
  \providecommand{\remarkname}{Remark}
\providecommand{\theoremname}{Theorem}

\usepackage{babel}

  \providecommand{\corollaryname}{Corollary}
  \providecommand{\definitionname}{Definition}
  \providecommand{\lemmaname}{Lemma}
  \providecommand{\propositionname}{Proposition}
  \providecommand{\remarkname}{Remark}
\providecommand{\theoremname}{Theorem}

\usepackage{babel}
\providecommand{\corollaryname}{Corollary}
  \providecommand{\definitionname}{Definition}
  \providecommand{\lemmaname}{Lemma}
  \providecommand{\propositionname}{Proposition}
  \providecommand{\remarkname}{Remark}
\providecommand{\theoremname}{Theorem}

\makeatother

\usepackage{babel}
\providecommand{\corollaryname}{Corollary}
\providecommand{\definitionname}{Definition}
\providecommand{\lemmaname}{Lemma}
\providecommand{\propositionname}{Proposition}
\providecommand{\remarkname}{Remark}
\providecommand{\theoremname}{Theorem}

\begin{document}

\title[\shorttit{} \partnum{}]%
{{\large\scshape\bfseries \sertit{} \partnum{}: \\
\parttit{}}}

\author[Ameen \& Kyt\"ol\"a \& Park \& Radnell]%
{Taha Ameen, Kalle Kyt\"ol\"a, S.C. Park, and David Radnell}

\address{Department of Electrical and Computer Engineering, University of Illinois at Urbana-Champaign, Urbana, IL}

\email{tahaa3@illinois.edu}

\address{School of Mathematics, Korea Institute of Advanced Study, 85 Hoegi-ro, 
Dongdaemun-gu, Seoul 02455, Republic of Korea}

\email{scpark@kias.re.kr}

\address{Department of Mathematics and Systems Analysis, Aalto University,
P.O. Box 11100, FI-00076 Aalto, Finland}

\email{kalle.kytola@aalto.fi}

\email{david.radnell@aalto.fi}

\begin{abstract}
	This is the first in a series of 
	articles about
	recovering the full algebraic structure of
	a boundary conformal field theory (CFT) from
	the scaling limit of the critical Ising model in slit-strip 
	geometry. Here, we introduce spaces of holomorphic functions in
	continuum domains as well as corresponding spaces of discrete 
	holomorphic functions in lattice domains.
	We find distinguished sets of functions
	characterized by their singular behavior in the three 
	infinite directions in the slit-strip domains, and note in particular that natural subsets of these functions span analogues
	of Hardy spaces.
	We prove convergence results of the distinguished 
	discrete holomorphic functions to the continuum ones.
	In the subsequent articles,
	the discrete holomorphic functions will be used for the
	calculation of the Ising model fusion coefficients
	(as well as for the 
	diagonalization of the Ising transfer matrix),
	and the convergence of the functions is used to prove the 
	convergence of the fusion coefficients.
	It will also be shown that
	the vertex operator algebra of the boundary conformal field theory
	can be recovered from 
	the limit of the fusion coefficients via geometric 
	transformations involving the distinguished continuum
	functions.

\end{abstract}

\maketitle

\section{Introduction}%
\label{sec: intro}

\subsection*{Conformal invariance results about the Ising model scaling 
limit}

We have witnessed a breakthrough in 
the mathematically precise understanding of the
conformal invariance properties of the critical planar
Ising model, following the discrete complex analysis ideas
pioneered by Smirnov~\cite{Smirnov-towards_conformal_invariance}.
The conformal invariance properties
arise in the scaling limit of the lattice model, upon
zooming out so that the lattice mesh tends to zero.

One facet of the progress has been advances in the random
geometry description of the scaling limit.
It has been proven that interfaces arising with 
Dobrushin boundary conditions
in both the Ising model
and its 
random cluster model 
counterpart
tend to conformally invariant random curves known as 
Schramm-Loewner evolutions~(SLE)
~\cite{Smirnov-conformal_invariance_in_RCM_1,
CS-universality_in_Ising,
CDHKS-convergence_of_ising_interfaces}.
Generalizations of interface convergence results 
for boundary conditions other than Dobrushin type
have been obtained
in~\cite{HK-Ising_interfaces_and_free_bc,
Izyurov-Smirnovs_observable_for_free_boundary_conditions,
Izyurov-Ising_interfaces_in_multiply_connected_domains,
KS-boundary_touching_loops,
KS-configurations_of_FK_Ising_interfaces,
BDH-crossing_probabilities_with_free_bc,
BPW-uniqueness_of_global_multiple_SLEs,
PW-crossing_probabilities_of_multiple_Ising_interfaces,
Karrila-multiple_SLE_scaling_limits}.
The full collection of all interfaces in the Ising model
and its random cluster model counterpart tend to processes
known as conformal 
loop ensembles~(CLE)~\cite{KS-conformal_invariance_of_RCM_2,
BH-scaling_limit_of_Ising_interfaces}.

Instead of the random geometry of interfaces, the physics tradition
as well as the constructive quantum field theory tradition
place focus on correlation functions. 
The existence of scaling limits of renormalized Ising model correlation 
functions, and the conformal covariance of these scaling limits,
have been shown for 
energy~\cite{HS-energy_density, Hongler-thesis}
and spin~\cite{CHI-conformal_invariance_of_spin_correlations}.
Recently a similar conclusion has been obtained for
mixed correlation functions of all primary fields including
the spin and energy~\cite{CHI-primary_field_correlations}.
It has even been shown that the set of all possible lattice
local fields of the Ising model carries a representation of the Virasoro 
algebra~\cite{HKV-CFT_at_the_lattice_level},
a hallmark of conformal field theories (CFT), and that 
with generic renormalization local correlation functions 
of
such fields have conformally covariant limits
~\cite{GHP-Ising_local_spin_correlations}.
Building on the correlation function results, it has furthermore been 
proven that the collection of 
Ising spins viewed as a random field
converges to a conformally covariant scaling 
limit~\cite{CGN-planar_Ising_magnetization_field,
CGN-planar_Ising_magnetization_field_2}.

The 100 year history of the Ising model contains a wealth of 
ingenious mathematical ideas that have enabled 
rigorous results, including 
transfer matrix methods~\cite{KW-statistics_of_the_2D_ferromagnet,
Onsager-crystal_statistics,
Yang-spontaneous_magnetization}
and their fermionic formulations~\cite{Kaufman-crystal_statistics_2,
SML-Ising_model_as_a_problem_of_fermions} and
Toeplitz determinants~\cite{MPW-correlations_and_spontaneous_magnetization,
Wu-theory_of_Toeplitz_determinants-1,
FH-Toeplitz_determinants},
Kac-Ward matrices~\cite{KW-combinatorial_solution},
dimer representations~\cite{Kasteleyn-dimer_statistics,
Fisher-dimer_solution_of_planar_Ising_model},
discrete complex analysis~\cite{KC-determination_of_an_operator_algebra,
Mercat-discrete_Riemann_surfaces},
commuting families of transfer 
matrices~\cite{SM-a_new_representation_of_the_solution_of_Ising},
Yang-Baxter equations~\cite{BE-Nth_solution},
non-linear differential equations (particularly Painlev\'e type)
and difference equations
~\cite{WMTB-spin_spin,
Perk-quadratic_identities_for_Ising_model_correlations,
JM-studies_on_holonomic_quantum_fields_17},
and bosonization~\cite{Dubedat-exact_bosonization};
for more on the various mathematical developments
see, e.g.,~\cite{MW-two_dimensional_Ising_model,
Baxter-exactly_solved_models,
Palmer-planar_Ising_correlations} and
\cite{DIK-Toeplitz,CCK-revisiting}.
The recent breakthrough mathematical progress on the conformal
invariance of both random geometry and correlation functions of
the Ising model, however,
has been enabled mainly by novel notions of discrete complex 
analysis that apply particularly well to the Ising model:
\emph{s-holomorphicity} 
and specific \emph{Riemann boundary value problems}
\cite{Smirnov-conformal_invariance_in_RCM_1,
CS-discrete_complex_analysis_on_isoradial_graphs,
CS-universality_in_Ising}.

\subsection*{The conformal field theory picture}

The prediction of conformal invariance
was made in theoretical physics research in the 1980's,
in a research field titled conformal field theory (CFT).
Physicists predicted that, very generally, models of 
two-dimensional statistical physics
at their critical points of continuous phase transitions
should, in the scaling limit, be described by
field theories with conformal 
symmetry~\cite{BPZ-infinite_conformal_symmetry_of_critical_fluctuations}.
Such conformal field theories turn out to be algebraically 
very stringently 
constrained~\cite{BPZ-infinite_conformal_symmetry_in_QFT}~--- 
in mathematical terms their 
chiral symmetry algebras are \emph{vertex operator
algebras}~(VOA)~\cite{FLM-VOAs_and_Monster,
Kac-vertex_algebras_for_beginners,
LL-introduction_to_VOAs, Huang-CFT_and_VOA}.
This prediction and the associated algebraic structure
leads to absolutely remarkable, specific, 
exact 
predictions about the statistical physics models~--- including
values of critical exponents, formulas for scaling limit
correlation functions,
modular invariance of renormalized scaling limit
partition functions on tori, etc.,
see, e.g.~\cite{DMS-CFT,Mussardo-statistical_field_theory}.

The square lattice Ising model is an archetype of such statistical
physics models, and known results about it lend very strong support to
the predicted general picture. But although there is thus virtually no doubt 
that the conformal field theory picture for the scaling limit of the 
Ising model is correct in an exact sense without approximations,
there is still no mathematical result establishing a complete conformal 
field theory as the scaling limit of the critical Ising model, and
one even struggles to find a precisely stated mathematical
conjecture about it in the literature!

The general goal of this series of articles is to remedy this situation
by showing that the full algebraic structure of the conformal
field theory generally conjectured to describe the scaling limit of the
Ising model is indeed recovered in the scaling limit.
More precisely, the combination of results proven
in this series establishes that the fusion coefficients
of the Ising model with locally monochromatic boundary conditions 
in slit-strip geometry
(defined in~\cite{part-2} as renormalized limits of boundary
correlation functions in lattice slit-strips)
converge in the scaling limit to the structure
constants of the vertex operator algebra of the fermionic Ising
boundary conformal field theory.

\subsection*{Slit-strip geometry and boundary conformal field theory}

In this series of articles we consider the Ising model
in lattice approximations of
the strip and slit-strip geometries illustrated 
in Figures~\ref{fig: cont strip and slit-strip}.
Likewise, we extensively use spaces of
holomorphic functions in the strip and the slit-strip, as well
as discrete holomorphic functions in their lattice approximations.
Let us briefly explain the role that these geometries play
in boundary conformal field theory.

\begin{figure}[tb]
\centering
\subfigure[Cylinder surface.] 
{
    \includegraphics[scale=2.5]{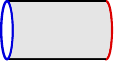}
	\label{sfig: cylinder}
}
\hspace{2.5cm}
\subfigure[Strip surface (truncated).] 
{
    \includegraphics[scale=2.5]{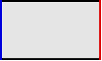}
	\label{sfig: strip surface}
} \\
\subfigure[Pair-of-pants surface.] 
{
    \includegraphics[scale=2.5]{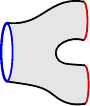}
	\label{sfig: pair-of-pants}
}
\hspace{2.5cm}
\subfigure[Slit-strip surface (truncated).] 
{
    \includegraphics[scale=2.5]{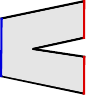}
	\label{sfig: slit-strip surface}
}
\caption{
The (truncated) strip and (truncated) slit-strip surfaces
have the same role in boundary CFT as 
the cylinder and pair-of-pants surfaces have in bulk CFT.
}
\label{fig: the generating surfaces}
\end{figure}
The rough idea is that the strip and the slit-strip are for
boundary conformal field theory what the cylinder (or annulus)
and the pair-of-pants Riemann surfaces are for 
bulk conformal field theory.
Let us begin with
    the more familiar setup of
bulk conformal field theory.
\footnote{%
\emph{Bulk conformal field theory} is in fact commonly referred to
plainly as \emph{conformal field theory} (CFT), and the term
\emph{boundary conformal field theory} (BCFT)
is then used to distinguish the case when the domains of interest
have physical boundaries on which boundary conditions can be
imposed (for the fields of the quantum field theory and for the
statistical mechanics model that is to be described by the quantum
field theory).
If one focuses only on symmetry algebras, the term
\emph{chiral conformal field theory} could be used in 
place of boundary conformal field theory as well, although this
term originally arises from
decomposing the symmetry algebra of a bulk conformal field theory
into two parts: holomorphic and antiholomorphic chiralities.
Our main focus will be the Ising model in domains with boundary, but
we use the term conformal field theory generally to variously refer to
any of the above~--- we then use the epithets boundary, bulk, and
chiral, where attention needs to be drawn to the particularities
of the case in question.
}
The role of geometry is most transparent in Segal's
axiomatization of conformal field 
theories~\cite{Segal-definition_of_CFT,
Segal-definition_of_CFT-new},
in which a CFT is defined as a (projective) functor~---
subject 
to certain axioms~--- from the category whose
morphisms are bordered Riemann surfaces 
with parametrized boundary components
to the category whose morphisms
are trace-class operators between tensor products of a given
Hilbert space. Segal's approach is clearly 
motivated by the
transfer matrix formalism in statistical mechanics:
the operators associated to cylider surfaces (of different moduli)
form a semigroup which is thought of as 
the scaling limit of the semigroup generated by the transfer matrix
itself.
\begin{figure}[tb]
\centering
\subfigure[Sewing together cylinders.] 
{
    \includegraphics[scale=2]{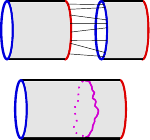}
	\label{sfig: two cylinders sewn}
}
\hspace{2.5cm}
\subfigure[Sewing together rectangles.] 
{
    \includegraphics[scale=2]{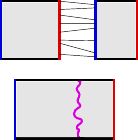}
	\label{sfig: two rectangles sewn}
}
\caption{The sewing together of cylinders in bulk CFT and the sewing 
together of truncated strips (i.e., rectangles) in boundary CFT give 
rise to semigroups
that correspond to the scaling limits of the transfer matrices.
}
\label{fig: sewing of cylinders and rectangles}
\end{figure}
Sewing together bordered Riemann surfaces
along (parametrized) boundary components is the composition in 
the category.
Cylinder surfaces alone can only be sewn with each other to 
form cylinders (sew together two cylinders) 
as in Figure~\ref{sec: intro}.\ref{sfig: two cylinders sewn}
or tori (sew together the two ends of a cylinder).
On the other hand, the pair-of-pants surfaces 
can be sewn together as 
in Figure~\ref{fig: general surface decomposition} 
to form surfaces of
arbitrary genus, and in this sense they are the building block of
all Riemann surfaces. The use of the term vertex operator
(and the symbol~$Y$ used for it)
originates from the picture of the pair-of pants surface
(a vertex diagram in string theory) and 
the operator that is associated with this surface.
\begin{figure}[tb]
\centering
\includegraphics[scale=2]{pants_surface.pdf}
\hspace{1.5cm}
\includegraphics[scale=2]{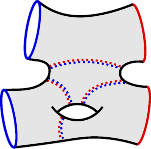}
\hspace{1.5cm}
\includegraphics[scale=2]{pants_surface.pdf}
\caption{Sewing many pair-of-pants can produce arbitrary genus surfaces.
}
\label{fig: general surface decomposition}
\end{figure}

It is a natural change in the point of view~\cite{Huang-CFT_and_VOA,
RSS-quasiconformal_Teichmuller_theory_as_foundation}
to equip the bordered Riemann surfaces with tubes 
(cylinders infinitely extended 
in one direction) attached to the boundary components so that the 
surfaces become punctured surfaces: the punctures correspond to the
infinite extremities
of the tubes (and thus there is one for each boundary component of the
original bordered Riemann surface) and they become 
equipped with a choice of local coordinates.
With this point of view, cylinders correspond to the 
Riemann sphere with two punctures, and the pair-of-pants 
corresponds to the the Riemann sphere with three punctures.

A rectangle, or equivalently a truncated strip,
is the natural counterpart in boundary
conformal field theory to a cylinder of finite modulus
in bulk conformal field theory, whereas the doubly infinite strip
is the counterpart of the cylinder with tubes attached
to each end, i.e., the twice-punctured sphere.
The transfer matrix
is, indeed, simplest to use for calculations in rectangles 
and strips.
Along with the main result of this series,
we will of course also verify the familiar statement 
that the scaling limits of powers of transfer matrices
form the semigroup generated by the energy operator~$L_0$ 
in the vertex operator algebra, 
in agreement with the interpretation in Segal's formulation.

\begin{figure}[tb]
\centering
\includegraphics[scale=2]{pants_plane.pdf}
\hspace{1.5cm}
\includegraphics[scale=2]{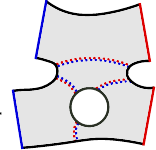}
\hspace{1.5cm}
\includegraphics[scale=2]{pants_plane.pdf}
\caption{Sewing many truncated slit-strips can produce
arbitrary multiply connected domains.
}
\label{fig: general domain decomposition}
\end{figure}
A truncated slit-strip is the natural counterpart in boundary
conformal field theory for a pair-of-pants surface
in bulk conformal field theory, whereas the infinite slit-strip
(with three infinite extremities) is the counterpart of the
pair-of-pants with tubes attached to each end, i.e., the 
thrice-punctured sphere.
By sewing together boundary 
intervals of (surfaces conformally equivalent to)
truncated slit-strips,
it is possible to build general multiply connected 
planar domains with boundaries as in 
Figure~\ref{fig: general domain decomposition}.
In this sense the 
slit-strip geometry is the building block of general
domains in boundary 
conformal field theory in exactly the same way that the 
pair-of-pants surface is in bulk conformal field theory.

Our main result that the fusion coefficients of the Ising model
with locally monochromatic boundary conditions in the
lattice slit-strip tend, in the scaling limit, to the structure
constants of the vertex operator algebra, therefore recover
the clear boundary conformal field theory analogue of the role 
that vertex operators have in bulk conformal field theory
as operators associated with the pair-of-pants geometry.

\subsection*{Overview of the series}
\nopagebreak[4]
This series of articles is divided into three parts.
\nopagebreak[4]
\begin{description}
\item[Part 1: Discrete and continuous function spaces]
This first part concerns the function spaces needed in
the analysis of the scaling limit of the Ising model fusion 
coefficients.
We will consider spaces of holomorphic solutions to a
Riemann boundary value problem in the strip and the slit-strip,
as well as their discretized analogues: spaces of s-holomorphic
solutions to a
Riemann boundary value problem in a lattice strip and 
a lattice slit-strip.
We will, in particular, study the restrictions of such functions
to a cross-section of the strip or the slit-strip~---
much like in Segal's CFT one would view cross-sections of 
surfaces as carrying physical states in the Hilbert space,
which are then acted on by the operators associated to
Riemann surfaces lying between different cross-sections.\footnote{%
The Hilbert space of functions we consider here is
not directly the analogue of the Hilbert space of states in the 
quantum field theory, however. Instead, a good analogue of the
quantum field theory state space is the alternating tensor algebra
of the subspace of functions which admit a regular extension 
in one infinite direction of the strip, i.e., of a suitable analogue
of a Hardy space.
For details, see the subsequent parts~\cite{part-2,part-3}.}
The Riemann boundary values (and s-holomorphicity) are real-linear
conditions, and the function spaces here will all be real vector 
spaces. They will have the natural Hilbert space structure coming 
from square integrability of the functions on cross-sections.
We will find and concretely describe distinguished bases of
functions, which have prescribed singularities in one of the infinite
extremities and which are regular in the other two.
We prove the convergence of the distinguished functions on 
the lattices to the distinguished continuum functions, and
the convergence of all corresponding inner products.
The distinguished discrete functions in both the lattice
strip and the lattice slit-strip will be crucially used
in the calculations with the
Ising model in the second part~\cite{part-2}.
The distinguished continuum functions in the slit-strip will 
be used as a more natural basis for calculations with the 
vertex operator algebra in the third part~\cite{part-3}.
\item[Part 2: Scaling limits of fusion coefficients]
The second part of this series~\cite{part-2}
will focus on the Ising model itself,
in the lattice strip and the lattice slit-strip and with locally 
monochromatic boundary conditions. It will be shown that there is
a way to diagonalize the transfer matrices associated with the 
strip and the slit-strip using Clifford algebra valued discrete
one-forms built from one set of distinguished discrete functions
in the present article, 
and that s-holomorphicity and Riemann boundary values
underlie the possibility to perform contour deformations
in the integrals of these one-forms. 
The contour deformations are clearly analogous
to boundary conformal field theory, and using them with the
other set of distinguished discrete functions of the present
article, we derive a
recursive characterization of the fusion coefficients of the 
Ising model. The recursion involves inner products of the
distinguished discrete functions, and by the present results 
on their convergence, we will be able to derive the convergence
of the fusion coefficients in the scaling limit.
\item[Part 3: The vertex operator algebra in the scaling limit]
In the third part of this series~\cite{part-3}
we arrive at the main statement: from the scaling limits of the 
fusion coefficients one can recover the structure constants
of the vertex operator algebra of the fermionic conformal field theory
that has been claimed to describe the Ising model, and 
conversely from the structure constants one can
recover the scaling limits of the fusion coefficients.
The recovery involves only changes of bases related to the
choice of natural local coordinates 
at the three infinite extremities of the slit-strip; which again
naturally involve the distinguished continuum functions 
of the present article.
\end{description}

Together, the series provides a fully rigorously
worked out model case of a mathematically precise statement
about the emergence of the full algebraic structure of a
boundary conformal field theory in the scaling limit of
a lattice model of statistical mechanics.
Given the broad conjectured validity of the conformal field
theory picture,
this should be viewed as the prototype of a precise conjecture to
be formulated about many other models.
Some of our steps are inevitably specific to the Ising model
(particularly the role of s-holomorphicity and
Riemann boundary values), but certain steps could
even offer technical insights into the cases of other models.

\subsection*{Organization of this article}

This first article of the series is organized as follows.

\begin{description} 
\item[Section~\ref{sec: function spaces}]
We define the continuum function spaces
and find the distinguished holomorphic solutions to the
Riemann boundary value problems in the strip and the slit-strip.
In the case of the strip, the distinguished functions are
eigenfunctions of vertical translations, and their restrictions
to horizontal cross-sections are simply certain Fourier modes
with a quarter-of-a-full-turn 
phase difference between the two boundaries.
In the slit-strip, the distinguished functions are 
constructed as linear combinations of pull-backs of certain 
monomial functions from the upper half-plane.
\item[Section~\ref{sec: discrete complex analysis}]
We define the discrete function spaces
and find the distinguished s-holomorphic solutions to the
Riemann boundary value problems in the lattice strip and the 
lattice slit-strip.
The distinguished functions in the lattice slit-strip are again
vertical translation eigenfunctions, whose restrictions to
a cross-section turn out to be combinations of two discrete
Fourier modes with opposite frequencies. 
The qualitative properties of the spectrum of the
vertical translations of s-holomorphic solutions to the
Riemann boundary value 
problem in the strip were observed
in~\cite{HKZ-discrete_holomorphicity_and_operator_formalism},
but for our purposes we need both the spectrum
and the eigenfuctions explicitly.
This explicit calculation is essentially 
equivalent to the diagonalization of the induced 
rotation of the transfer matrix on the space of
Clifford generators~\cite{Palmer-planar_Ising_correlations},
but we find the vertical translation eigenfunction problem for 
discrete holomorphic functions conceptually simpler, 
and we provide the details of the derivation
also so that the series constitutes a self-contained proof of
the main result.
The distinguished functions in the lattice slit-strip can be
constructed based on the invertibility of a
finite-dimensional linear system of subspace projection 
equations.
\item[Section~\ref{sec: convergence of functions}]
The final section addresses the convergence of the
distinguished discrete functions to the continuum ones.
Having calculated the vertical translation 
eigenfunctions in the lattice strip explicitly, their convergence
becomes a matter of straightforward inspection.
The remaining results rely on the 
``imaginary part of the integral of the square'' 
technique~\cite{Smirnov-towards_conformal_invariance,
CS-discrete_complex_analysis_on_isoradial_graphs,
CS-universality_in_Ising}.
We first of all use it to prove the invertibility of
the linear system of subspace projections from
the previous section.
Moreover, the resulting regularity theory for s-holomorphic 
functions is used to prove
the convergence of the distinguished discrete 
functions in the slit-strip.
\end{description}

\subsection*{Novelty}

We do not claim essential novelty in any of the results
concerning the strip geometry~---
this case is included mainly 
for coherent formulation of the whole:
the definitions are needed in any case, and proofs are
provided for self-containedness.
All of our calculations in the lattice strip are fully explicit
and in essence equivalent to the calculations needed to diagonalize 
the transfer matrix of the Ising model with locally monochromatic 
boundary conditions. The well known diagonalization of this
transfer matrix~\cite{AM-transfer_matrix_for_a_pure_phase,
Palmer-planar_Ising_correlations}
in particular allows one to conclude without difficulty that 
the suitable powers of the transfer matrix converge
to the exponentials of the energy operator~$L_0$ in the 
vertex operator algebra of the fermionic Ising CFT,
for example by realizing the VOA as an inductive limit of transfer 
state spaces.

Instead, the novelty of our work pertains almost exclusively
to the slit-strip geometry. 
Key objects for us are the distinguished
functions in the lattice slit-strip, whose asymptotics
in one of the extremities matches the behavior of the explicit
strip functions.
Such globally defined discrete holomorphic functions
are analogous to objects needed in Segal's CFT
for vertex operators; not merely the semigroup generated
by the energy operator~$L_0$.
The fact that such globally defined s-holomorphic 
functions exist at all is crucial to our later contour 
deformation arguments, and their convergence is at the heart of the
convergence of the Ising model fusion coefficients.
For both of these, recently developed 
specific techniques of discrete complex analysis
\cite{Smirnov-towards_conformal_invariance,
Smirnov-conformal_invariance_in_RCM_1,
CS-discrete_complex_analysis_on_isoradial_graphs,
CS-universality_in_Ising}
are indispensable.
And it is precisely thus established convergence and
recursion properties of the Ising model fusion coefficients
which allow us to recover
the vertex operator algebra
in the scaling limit.


\section{Continuum function spaces and decompositions}%
\label{sec: function spaces}

In this section, we introduce the function spaces which play a crucial
role in our analysis of the continuum limit of the Ising model fusion 
coefficients.
A key notion are certain Riemann boundary values
for holomorphic functions~\cite{Smirnov-towards_conformal_invariance}. 
The notion has found some use in
functional analysis~\cite{HP-Hardy_spaces_and_Ising_model},
but it is the analogous notion in the lattice
setup 
that has turned out particularly fruitful for the study of the Ising
model~\cite{Smirnov-discrete_complex_analysis_and_probability,
Chelkak-state_of_the_art_and_perspectives}.
The straightforward continuum problem considered in the current
section provides an instructive blueprint for what to expect of the
lattice discretizations of Section~\ref{sec: discrete complex analysis}.

For our purposes, holomorphic functions with 
Riemann boundary values will be studied in two different geometries: the 
infinite strip~$\cstrip$ and the infinite slit-strip~$\cslitstrip$
of Figure~\ref{fig: cont strip and slit-strip}.
In the spirit of Segal's geometric formulation of conformal field
theories~\cite{Segal-definition_of_CFT, Segal-definition_of_CFT-new}, we focus 
in particular on the restrictions of 
such functions to a 
crosscut of the strip or the slit-strip. In both cases, the crosscut is 
basically an interval, and the appropriate function space is a space of 
square-integrable complex valued functions on the crosscut interval.
This space of complex valued functions is made into a real Hilbert space,
because the Riemann boundary values are a real-linear condition.
An obvious difference to Segal's formulation is that we consider geometries
with boundaries, analogous to open-string string theory rather than the more 
common closed-string version for which Segal's formulation is
directly suitable.
Correspondingly the cross sections are not (disjoint unions of) circles as in 
Segal's formulation, but rather (disjoint unions of) intervals.
\begin{figure}[tb]
\centering
\subfigure[The infinite vertical strip~$\cstrip$ and its 
horizontal cross-section~$\ccrosssec$.] 
{
  \includegraphics[width=.35\textwidth]{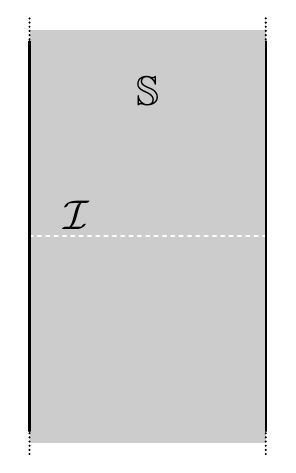}
  \label{sfig: cont strip}
}
\hspace{2.5cm}
\subfigure[The infinite slit-strip~$\cslitstrip$ and its 
horizontal cross-section~$\ccrosssec$.] 
{
  \includegraphics[width=.35\textwidth]{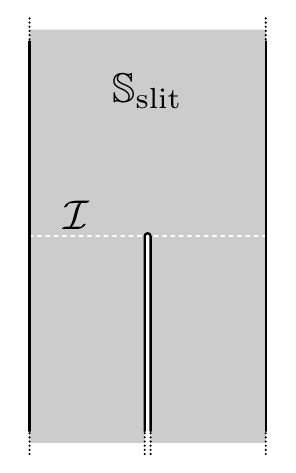}
  \label{sfig: cont slit-strip}
}
\caption{The strip and slit-strip domains.
}
\label{fig: cont strip and slit-strip}
\end{figure}

In Section~\ref{ssec: functions in strip and slit-strip} we 
define the Riemann boundary values, and introduce the appropriate function 
spaces for the strip and the slit-strip geometries.
In Section~\ref{ssec: vertical translation eigenfunctions} we introduce
the basis of the function space corresponding to vertical translation 
eigenfunctions in the strip. 
These continuum functions are just Fourier modes with a 
quarter-integer phase difference between boundaries, but their
discrete analogues will be a key to
the diagonalization of the Ising transfer matrix.
In Section~\ref{ssec: decompositions of functions in slit-strip}
we introduce analogous functions in the slit-strip, defined 
locally near each of the three extremities of the slit-strip, as well as 
globally defined functions which have prescribed singularities 
at the three extremities. 
The latter will feature naturally in expressions for VOA 
matrix elements.

\subsection{Functions in the strip and the slit-strip}
\label{ssec: functions in strip and slit-strip}

\subsubsection*{Riemann boundary values for holomorphic functions}
\label{sssec: RBV}

Let $\domain \subset \bC$ be a domain (open, connected subset).
Suppose that $x_0 \in \bdry \domain$ is a boundary point of the domain
(more precisely a prime end) such that locally near~$x_0$ the 
boundary~$\bdry \domain$ is a smooth curve, and let
$\tccw(x_0)$ be a unit complex number representing the direction of the 
counterclockwise oriented tangent to the boundary at~$x_0$.
A holomorphic function~$F:\domain\to\bC$ which continuously extends 
to~$x_0$ has \term{Riemann boundary value} at~$x_0$ if
\begin{align} \label{eq:rbv_cf}
F(x_0) \in \ii \, \tccw(x_0)^{-1/2} \, \bR .
\end{align}

\subsubsection*{The strip and the slit-strip}
\label{sssec: strip and slit-strip}
The two domains of interest to us will be
the unit width vertical
\term{strip} 
\begin{align}\label{eq: strip def}
\cstrip := \set{ z \in \bC \; \bigg| \; \aaa < \re(z) < \bbb}
\end{align}
and the \term{slit-strip} 
\begin{align}\label{eq: slit-strip def}
\cslitstrip := \cstrip \setminus \set{\ii y \; \big| \; y \leq 0 } .
\end{align}
These domains are illustrated in 
Figure~\ref{fig: cont strip and slit-strip}.

According to definition~\eqref{eq:rbv_cf}, a holomorphic 
function~$F \colon \cstrip \to \bC$ in the strip
has Riemann boundary values if for all~$y \in \bR$ we have
\begin{align}\label{eq: RBV in strip}
F \Big( \frac{-1}{2} + \ii y \Big) \in e^{-\ii \pi / 4} \, \bR
\qquad \text{ and } \qquad
F \Big( \frac{+1}{2} + \ii y \Big) \in e^{+\ii \pi / 4} \, \bR .
\end{align}
For a holomorphic function $F \colon \cslitstrip \to \bC$
to have Riemann boundary values in the slit-strip, 
in addition to the above it is required that 
for any~$y < 0$, the left and right limits on the slit 
part of the boundary satisfy
\begin{align}\label{eq: RBV on the slit part}
F \big( 0^- + \ii y \big) \in e^{+\ii \pi / 4} \, \bR
\qquad \text{ and } \qquad
F \big( 0^+ + \ii y \big) \in e^{-\ii \pi / 4} \, \bR .
\end{align}

\subsubsection*{The horizontal cross-section}
We study functions on the strip~$\cstrip$ and the 
slit-strip~$\cslitstrip$ domains through their 
restrictions to the horizontal cross-section at zero imaginary part
\begin{align}\label{eq: cross section}
\ccrosssec := \left[ \aaa , \bbb\right] ,
\end{align}
and therefore consider appropriate spaces of complex valued functions on 
this interval. For this purpose, we use the real Hilbert space
\begin{align}\label{eq: L2 space}
\cfunctionsp := L^2_{\bR} ( \ccrosssec , \bC) ,
\end{align}
of square-integrable complex valued functions on the cross section.
The square-integrability requirement can be seen as imposing the
Riemann bounday value also at the tip of the slit, in an
appropriate (conformal) sense.
The norm~$\|f\|$ of~$f \in \cfunctionsp$ is obtained from
\begin{align*}
\| f \|^2 = \int_{\aaa}^{\bbb} |f(x)|^2 \; \ud x ,
\end{align*}
as usual, but we emphasize that the inner product takes the form
\begin{align}\label{eq: inner product on L2}
\innprod{f}{g} 
= \; & \int_{\aaa}^{\bbb} \Big(
    \re \big( f(x) \big) \,  \re \big( g(x) \big)
    + \im \big( f(x) \big) \,  \im \big( g(x) \big) \Big) \; \ud x \\
\nonumber
= \; & \int_{\aaa}^{\bbb} \re \Big(
    f(x) \,  \overline{g(x)} \Big) \; \ud x ,
\end{align}
since we view~$\cfunctionsp$ as a Hilbert space over~$\bR$, not~$\bC$.

\subsection{Decomposition into vertical translation eigenfunctions}
\label{ssec: vertical translation eigenfunctions}

First, consider functions in the vertical 
strip~$\cstrip$. 
We look for
holomorphic functions~$F \colon \cstrip \to \bC$ 
with Riemann boundary values~\eqref{eq: RBV in strip},
which are furthermore
eigenfunctions for vertical translations, i.e.,
\begin{align}
\label{eq: vertical translation eigenfunction property}
F(z + \ii h) = \ceigvalsym(h) F(z) 
\qquad
\text{ for all $z \in \cstrip$ and $h \in \bR$.}
\end{align}
The vertical translation eigenfunction 
property~\eqref{eq: vertical translation eigenfunction property}
is clearly only possible if~$\ceigvalsym(h) = e^{p h}$ 
for some~$p$, 
and it also implies that the function must 
factorize as ${F(x+\ii y) = f(x) \, e^{p y}}$. Cauchy-Riemann 
equations then amount to~$f'(x) + \ii p f(x) = 0$, which yields
$f(x) \propto e^{-\ii p x}$.
The Riemann-boundary values~\eqref{eq: RBV in strip} in turn can be satisfied
only if 
$e^{\ii p} = f(\frac{-1}{2}) / f(\frac{+1}{2}) \in \ii \bR$,
i.e., $p \in \pi \bZ + \frac{\pi}{2}$. 
The functions of interest to us are therefore basically the analytic 
continuations of quarter-integer Fourier modes on the 
cross-section~$\ccrosssec$ (the argument makes one 
quarter-turn plus any number of half-turns from one end of the 
interval~$\ccrosssec$ to the other).

For indexing the Fourier modes (as well as the fermion modes in 
the vertex operator algebra later on),
we use the sets of positive half-integers
and of all half-integers denoted in what follows by
\begin{align}
\poshalfint 
:= \; & \big[ 0,+\infty \big) \cap \Big( \bZ + \half \Big) 
    = \set{\frac{1}{2}, \frac{3}{2}, \frac{5}{2}, \ldots} , \\
\nonumber
\pm \poshalfint 
:= \; & \poshalfint \cup(-\poshalfint) = \bZ + \half 
    = \set{\pm \frac{1}{2}, \pm \frac{3}{2}, \ldots} .
\end{align}

Then, for $k \in \pm \poshalfint$, we define the function
\begin{align}\label{eq: half integer Fourier mode}
\ccfFun_k (x+ \ii y) 
:= C_k \; \exp \big( -\ii \pi k x + \pi k y \big) ,
\qquad \text{ for } x + \ii y \in \cstrip ,
\end{align}
and its restriction to the cross section
\begin{align}\label{eq: half integer Fourier mode}
\ccffun_k (x) 
:= C_k \; e^{ -\ii \pi k x}  
\qquad \text{ for } x \in \ccrosssec ,
\end{align}
where the normalization constant is chosen as
\begin{align}\label{eq: normalization of half integer Fourier modes}
C_k := \; & e^{\ii \pi (-k/2-1/4)} 
\end{align}
to ensure
$\ccfFun_k \big( \frac{-1}{2} + \ii y \big) \in e^{-\ii \pi /4} \, \bR_+$ 
and $\|\ccffun_k\| = 1$.
These quarter-integer Fourier modes~\eqref{eq: half integer Fourier mode}
are illustrated in 
Figure~\ref{fig: continuous eigenfunctions}.
Let us start by checking that
they form an orthonormal basis of our function space~$\cfunctionsp$.
\begin{figure}
\centering
\subfigure[The Fourier mode~$\ccffun_{k}$ for~$k=1/2$.] 
{
    \includegraphics[width=.4\textwidth]{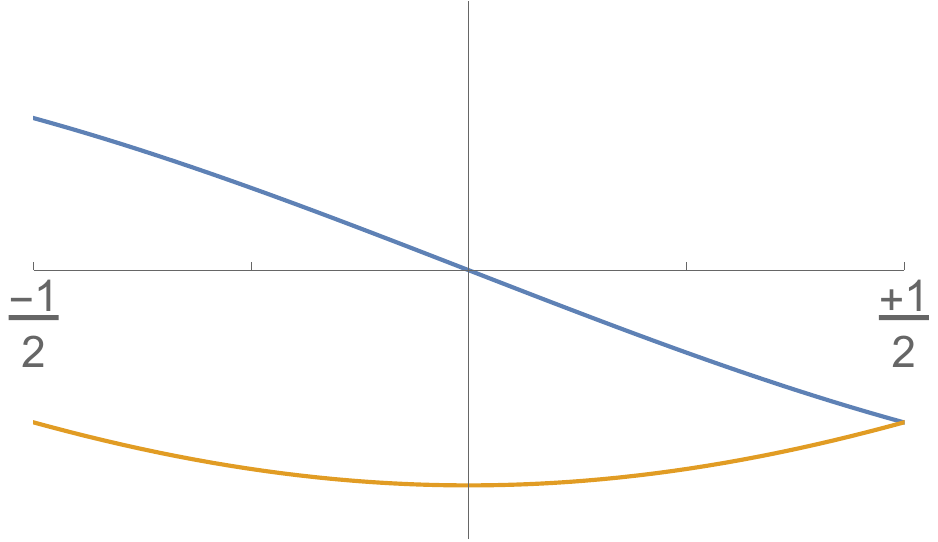}
	\label{sfig: ceigf P 1}
}
\hspace{2.0cm}
\subfigure[The Fourier mode~$\ccffun_{k}$ for~$k=-1/2$.] 
{
    \includegraphics[width=.4\textwidth]{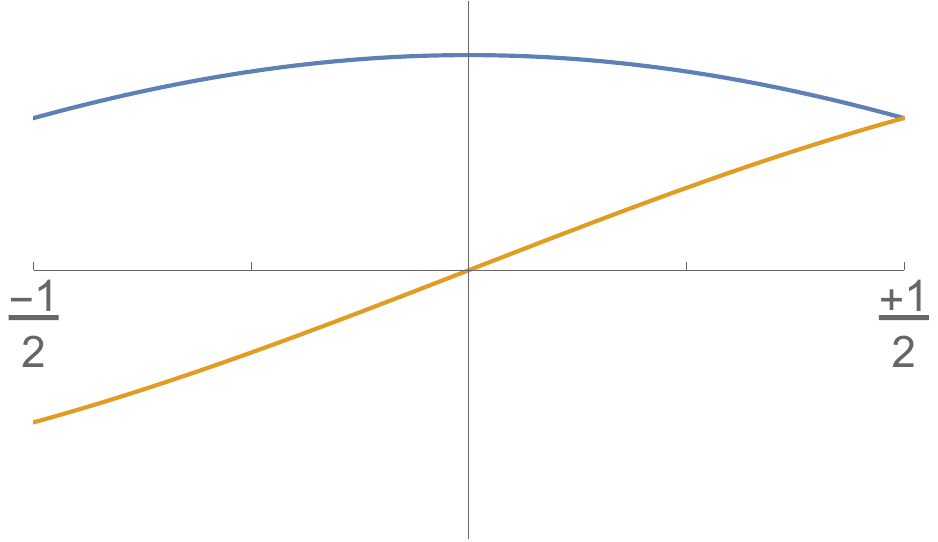}
	\label{sfig: ceigf M 1}
} \\
\subfigure[The Fourier mode~$\ccffun_{k}$ for~$k=3/2$.] 
{
    \includegraphics[width=.4\textwidth]{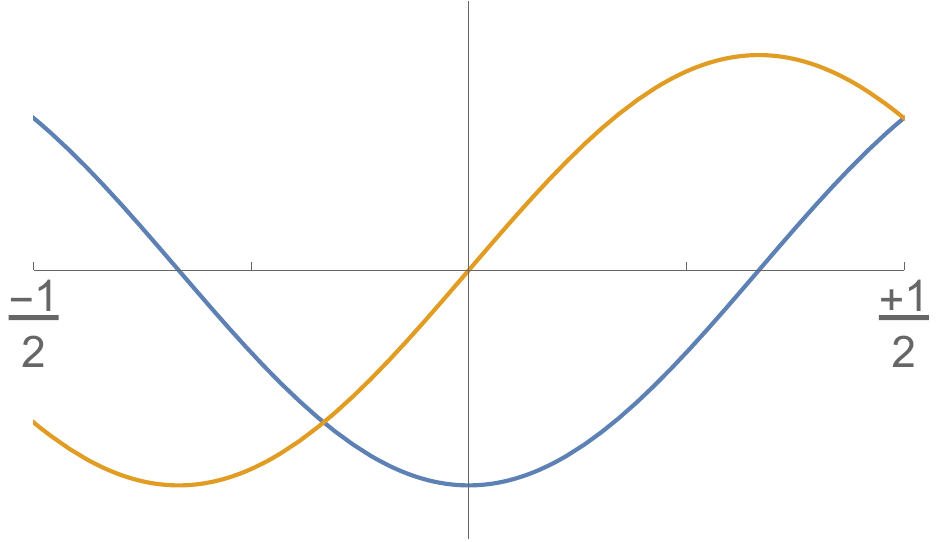}
	\label{sfig: ceigf P 3}
}
\hspace{2.0cm}
\subfigure[The Fourier mode~$\ccffun_{k}$ for~$k=-3/2$.] 
{
    \includegraphics[width=.4\textwidth]{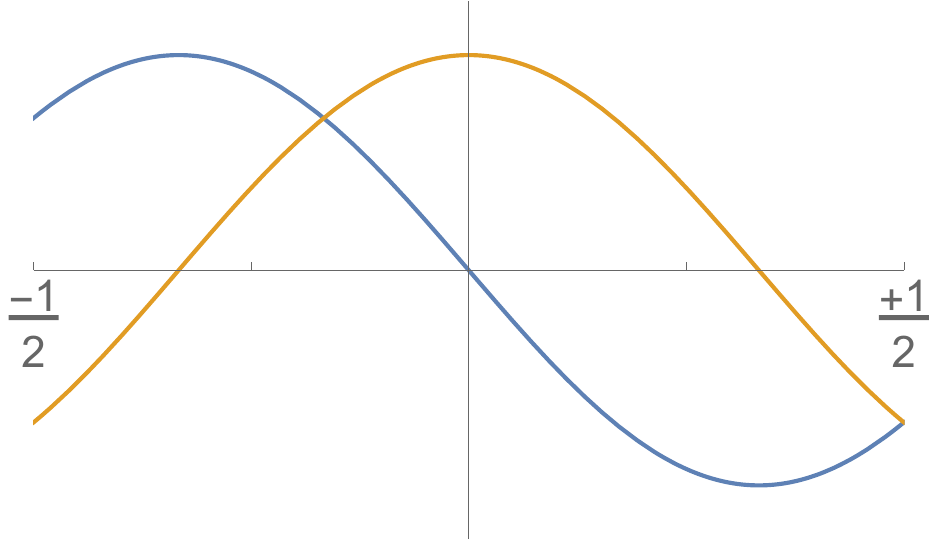}
	\label{sfig: ceigf M 3}
} \\
\subfigure[The Fourier mode~$\ccffun_{k}$ for~$k=5/2$.] 
{
    \includegraphics[width=.4\textwidth]{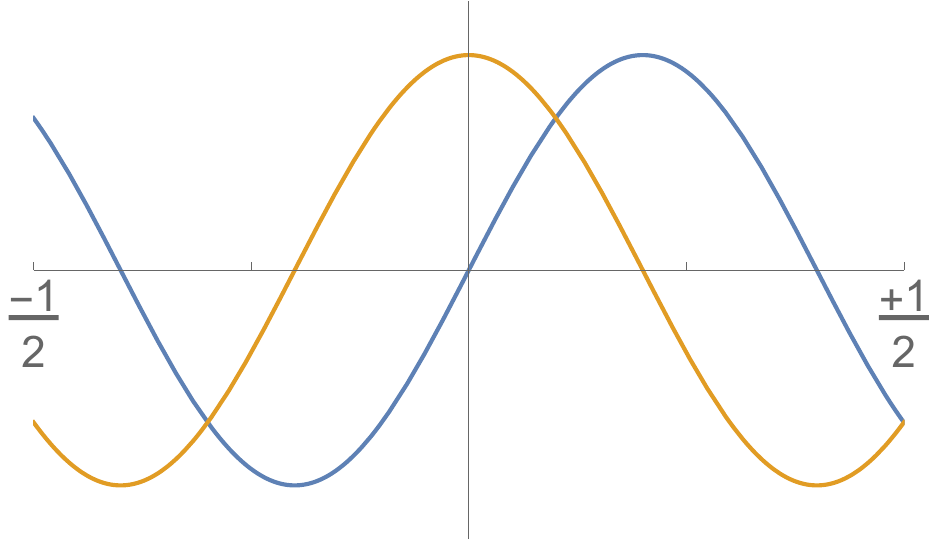}
	\label{sfig: ceigf P 5}
}
\hspace{2.0cm}
\subfigure[The Fourier mode~$\ccffun_{k}$ for~$k=-5/2$.] 
{
    \includegraphics[width=.4\textwidth]{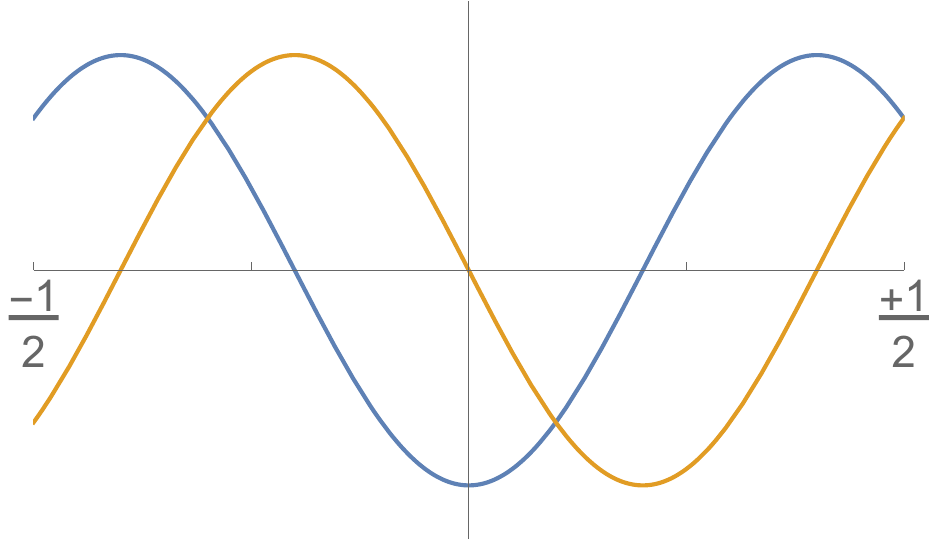}
	\label{sfig: ceigf M 5}
} \\
\subfigure[The Fourier mode~$\ccffun_{k}$ for~$k=7/2$.] 
{
    \includegraphics[width=.4\textwidth]{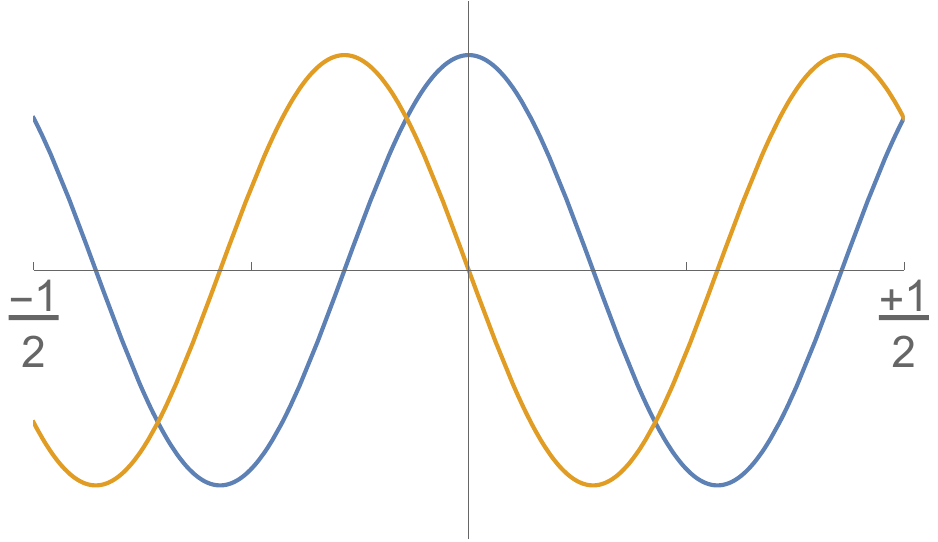}
	\label{sfig: ceigf P 7}
}
\hspace{2.0cm}
\subfigure[The Fourier mode~$\ccffun_{k}$ for~$k=-7/2$.] 
{
    \includegraphics[width=.4\textwidth]{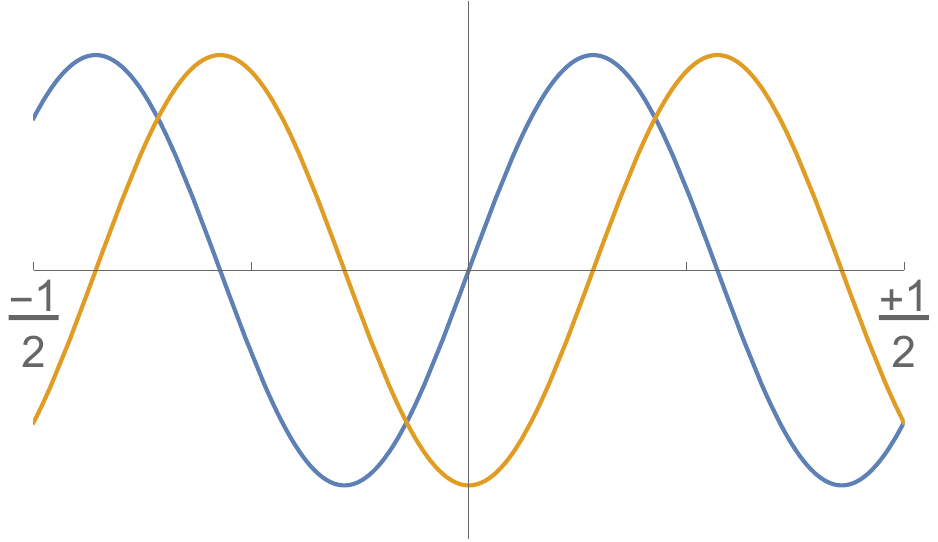}
	\label{sfig: ceigf M 7}
} \\
\label{fig: continuous eigenfunctions}
\caption{Real and imaginary parts (blue and orange)
of the restrictions~$\ccffun_k$ of the vertical translation 
eigenfunctions~$\ccfFun_k$ to the cross 
section~$\ccrosssec=[\mhalf,\half]$.
Riemann boundary conditions~\eqref{eq: RBV in strip}
fix the ratio of the real and imaginary parts
at the two endpoints, and our 
normalization~\eqref{eq: normalization of half integer Fourier modes}
fixes the phase at the left endpoint.}
\end{figure}
\begin{prop}\label{prop: half integer Fourier modes form a basis}
The collection~$(\ccffun_k)_{k \in \pm \poshalfint}$ is an orthonormal basis
for the real Hilbert space~$\cfunctionsp$.
In particular, for $k,k' \in \pm \poshalfint$ we have
\begin{align}\label{eq: orthonormality of half integer Fourier modes}
\innprod{\ccffun_k}{\ccffun_{k'}}  = \delta_{k,k'} .
\end{align}

\end{prop}
\begin{proof}
Orthonormality~\eqref{eq: orthonormality of half integer Fourier modes} 
is shown
by a routine trigonometric integral. It thus remains to show 
completeness of 
the collection~$(\ccffun_k)_{k \in \pm \poshalfint}$. 

Given $f \in \cfunctionsp$, a square-integrable complex valued function 
on the interval~$[-1/2,+1/2]$ of length~$1$, 
define a period~$4$ extension of~$f$ 
to~$\bR$ as follows.
First extend~$f$ to~$[-1/2,+3/2]$ 
by~$f(1-x) = \ii \, \overline{f(x)}$,
and then extend to~$\bR$ by the antiperiodicity 
condition~$f(x+2)=-f(x)$.
By ordinary Fourier series, we can write this extended square-integrable 
function on~$[-2,2]$ as
$f(x) = \sum_{n \in \bZ} c_n \, e^{-\ii n x \pi / 2} $,
with some complex coefficients~$c_n \in \bC$. 
It follows from antiperiodicity
that the even coefficients vanish, $c_{2m} = 0$ for all~$m\in\bZ$.
From the reflection property~$f(1-x) = \ii \, \overline{f(x)}$
it follows that the odd coefficients
satisfy~$c_{2k} + \ii^{2k-1} \, \overline{c_{2k}} = 0$
for all~$k \in \pm \poshalfint$,
i.e., that
$c_{2k}$ lies on the line~$\bR C_k$ in the complex plane. 
The terms in this 
Fourier series are therefore real multiples 
of the basis functions~$\ccffun_k$, 
$k \in \pm \poshalfint$. Completeness follows.
\end{proof}

\begin{rmk}

The functions obtained by imaginary multiplication
$x \mapsto \ii \, \ccffun_k(x)$,
and complex conjugation $x \mapsto \overline{\ccffun_k(x)}$ are 
obviously also square-integrable,
and related by~$\overline{\ccffun_k (x)} = \ii \, \ccffun_{-k} (x)$. 
Note that the expansions of both of these
in the basis~$(\ccffun_k)_{k \in \pm \poshalfint}$ 
of~
$\cfunctionsp$
(with real coefficients!) contain infinitely many terms.
The combination of the two operations, however, amounts
to simply changing the sign of the mode, 
$\overline{\ii \, \ccffun_k (x)} = \ccffun_{-k} (x)$,
as is evident also in Figure~\ref{fig: continuous eigenfunctions}.
\end{rmk}

The positive and negative modes 
form a splitting of the 
Hilbert space into two orthogonally complementary subspaces
\begin{align*}
\cfunctionsp
= \; & \cfspTpole \oplus \cfspTzero ,
\end{align*}
where
\begin{align}\label{eq: top poles and zeros def}
\cfspTpole 
	:= \; & \overline{\spn_\bR \set{ \ccffun_k \; \big| \; k > 0}} &
\cfspTzero
	:= \; & \overline{\spn_\bR \set{ \ccffun_k \; \big| \; k < 0}} .
\end{align}
We denote the orthogonal projections to these two subspaces by
\begin{align*}
\cprTpole \colon \; & \cfunctionsp \to \cfspTpole 
& 
\cprTpole (\ccffun_k) = \begin{cases}
                        0 & \text{ if $k<0$} \\
                        \ccffun_k & \text{ if $k>0$}
                        \end{cases} \\
\cprTzero \colon \; & \cfunctionsp \to \cfspTzero  
& 
\cprTzero (\ccffun_k) = \begin{cases}
                        \ccffun_k & \text{ if $k<0$} \\
                        0 & \text{ if $k>0$} .
                        \end{cases} 
\end{align*}

\subsection{Decomposition of holomorphic functions in the slit-strip}
\label{ssec: decompositions of functions in slit-strip}
With particularly the slit-strip~$\cslitstrip$ in mind
we also use functions
defined in the left and right halves of the strip,
\begin{align}\label{eq: half strip def}
\cstrip^{\lftsym} 
:= & \set{ z \in \bC \; \bigg| \; \frac{-1}{2} < \re(z) < 0} , &
\cstrip^{\rgtsym} 
:= & \set{ z \in \bC \; \bigg| \; 0 < \re(z) < \frac{+1}{2}} ,
\end{align}
and their restrictions to left and right halves of the cross sections
\begin{align}\label{eq: cross section}
\ccrosssec^{\lftsym} := \; & \Big[ \frac{-1}{2} , 0 \Big] , &
\ccrosssec^{\rgtsym} := \; & \Big[ 0 , \frac{+1}{2} \Big] .
\end{align}
The space of square-integrable functions has an orthogonal 
decomposition into 
functions with support in the left and right halves,
\begin{align*}
\cfunctionsp = \cfunctionsp_{\lftsym} \oplus \cfunctionsp_{\rgtsym} ,
\end{align*}
where we can interpret (modulo extension by zero to the other half)
\begin{align*}
\cfunctionsp_{\lftsym} = \; & L^2_{\bR}(\ccrosssec^{\lftsym}, \bC) , &
\cfunctionsp_{\rgtsym} = \; & L^2_{\bR}(\ccrosssec^{\rgtsym}, \bC) .
\end{align*}

The quarter-integer Fourier modes for the left 
and right halves are defined by
\begin{align}
\label{eq: half integer Fourier mode left}
\ccfFun^{\lftsym}_k (x+ \ii y) 
= \; & C^{\lftsym}_k \; \exp \big(- \ii 2 \pi k x + 2 \pi k y \big) ,
& 
\ccffun^{\lftsym}_k (x) 
= \; & C^{\lftsym}_k \, e^{-\ii 2 \pi k x}
\\
\label{eq: half integer Fourier mode right}
\ccfFun^{\rgtsym}_k (x+ \ii y) 
= \; & C^{\rgtsym}_k \; \exp \big( -\ii 2 \pi k x + 2 \pi k y \big) ,
&
\ccffun^{\rgtsym}_k (x) 
= \; & C^{\rgtsym}_k \, e^{-\ii 2 \pi k x} ,
\end{align}
for $k \in \pm \poshalfint$, where
we choose $C^{\lftsym}_k = \sqrt{2} \, e^{\ii \pi (-k-1/4)}$
and $C^{\rgtsym}_k = \sqrt{2} \, e^{-\ii \pi / 4}$
to ensure
${\ccfFun^{\lftsym}_k \big( \frac{-1}{2} + \ii y \big)} \in e^{-\ii \pi /4} \, 
\bR_+$ and
$\|\ccffun^\lftsym_k\| = 1$, and
${\ccfFun^{\rgtsym}_k \big( 0 + \ii y \big)} \in e^{-\ii \pi /4} \, \bR_+$ and
$\|\ccffun^\rgtsym_k\| = 1$, respectively.

The following orthonormal basis properties follow easily from
Proposition~\ref{prop: half integer Fourier modes form a basis} again.
\begin{prop}
The collections~$(\ccffun^{\lftsym}_k)_{k \in \pm \poshalfint}$ 
and~$(\ccffun^{\rgtsym}_k)_{k \in \pm \poshalfint}$ are orthonormal bases
for the real Hilbert spaces~$\cfunctionsp_\lftsym$ and~$\cfunctionsp_\rgtsym$,
respectively. These two collections combined form a basis for the real Hilbert 
space~$\cfunctionsp = \cfunctionsp_{\lftsym} \oplus \cfunctionsp_{\rgtsym}$.
\end{prop}

We can further split the functions with support on one of the two halves to 
those with negative or positive modes in the corresponding half, i.e., write
\begin{align*}
\cfunctionsp_{\lftsym}
= \; & \cfspLpole \oplus \cfspLzero , 
&
\cfunctionsp_{\rgtsym}
= \; & \cfspRpole \oplus \cfspRzero , 
\end{align*}
where
\begin{align}\label{eq: left and right poles and zeros def}
\cfspLpole 
	:= \; & \overline{\spn_\bR \set{ \ccffun^{\lftsym}_k \; \big| \; k < 0}} &
\cfspRpole 
	:= \; & \overline{\spn_\bR \set{ \ccffun^{\rgtsym}_k \; \big| \; k < 0}} 	
\\
\nonumber
\cfspLzero
	:= \; & \overline{\spn_\bR \set{ \ccffun^{\lftsym}_k \; \big| \; k > 0}}
&
\cfspRzero
	:= \; & \overline{\spn_\bR \set{ \ccffun^{\rgtsym}_k \; \big| \; k > 0}} .
\end{align}
\begin{rmk}
In~\eqref{eq: top poles and zeros def}, the poles corresponded to positive 
indices~$k>0$ and zeros to negative indices~$k<0$; 
the former Fourier modes are 
tending to infinity in the top extremity of the strip, 
and the latter to zero.
Here in~\eqref{eq: left and right poles and zeros def} 
we instead care about 
the asymptotics in the left and right downwards extremities of the 
slit-strip, where it is the modes with 
negative indices that tend to infinity 
and modes with positive indices that tend to zero~--- hence the opposite 
convention for the correspondence between labels and indices.
\end{rmk}

We have thus introduced two decompositions of
the function space~$\cfunctionsp$:
\begin{align}\label{eq: split to top poles and zeros}
\cfunctionsp
= \; & \cfspTpole \oplus \cfspTzero ,
\end{align}
and
\begin{align}\label{eq: split to left and right poles and zeros}
\cfunctionsp
= \; & \cfspLpole \oplus \cfspLzero \oplus \cfspRpole \oplus \cfspRzero .
\end{align}
The orthogonal projections to the two subspaces in 
decomposition~\eqref{eq: split to top poles and zeros} are 
denoted by~$\cprTpole \colon \cfunctionsp \to \cfspTpole$
and~$\cprTzero \colon \cfunctionsp \to \cfspTzero$. We denote the 
orthogonal projections onto the four subspaces in 
decomposition~\eqref{eq: split to left and right poles and zeros} by
\begin{align*}
\cprLpole \colon \; & \cfunctionsp \to \cfspLpole , &
\cprRpole \colon \; & \cfunctionsp \to \cfspRpole , \\
\cprLzero \colon \; & \cfunctionsp \to \cfspLzero , &
\cprRzero \colon \; & \cfunctionsp \to \cfspRzero .
\end{align*}

\subsubsection*{Singular parts of a function in the three extremities
of the slit-strip}
For a function~$f \in \cfunctionsp$, we call
\begin{align}
\nonumber
\cprTpole(f) \in \; & \cfspTpole 
& \text{ its \term{singular part at the top},} & \\
\label{eq: definition of singular parts}
\cprLpole(f) \in \; & \cfspLpole 
& \text{ its \term{singular part in the left leg},} & \\
\nonumber
\cprRpole(f) \in \; & \cfspRpole 
& \text{ its \term{singular part in the right leg}.} &
\end{align}
If $\cprTpole(f) = 0$ (resp. $\cprLpole(f)=0$ or
$\cprRpole(f)=0$),  we say that the function~$f$
admits a \term{regular extension} to the top
(resp. regular extension to the left 
leg or regular extension to the right leg).

The following result shows that a function is uniquely
characterized by its singular parts.
It is the analogue of the result that in bounded domains
holomorphic functions with Riemann boundary values
must vanish identically,
see~\cite{Hongler-thesis}. In our unbounded domains the
additional requirement is just regular extension to the
three infinite extremities. The proof technique is
a simple continuum version of the main tool we will
use in the discrete setup with s-holomorphic functions:
the (harmonic conjugate of the imaginary part of the)
integral of the square of the holomorphic function with 
Riemann boundary values.
\begin{lem}\label{lem:continuous-uniqueness}
If a function~$f \in \cfunctionsp$ admits regular extensions 
to the top, to the left leg, and to the right leg, then~$f \equiv 0$.
\end{lem}
\begin{proof}
First we will show that~$\re \int_{\aaa}^{\bbb} f(x)^2 \, \ud x = 0$.
By the assumption~$\cprTpole(f) = 0$, we can write
$f = \sum_{k'<0} c_{k'}\ccffun_{k'}$ with real 
coefficients~$c_{k'}$ which are square summable,
${\sum_{k'<0} c_{k'}^2 < \infty}$.
To obtain smooth approximations, for $N \in \bN$ define
the partial sum
\begin{align*}
f_N := & \sum_{-N < k' < 0} c_{k'} \ccffun_{k'} , &
F_N := & \sum_{-N < k' < 0} c_{k'} \ccfFun_{k'} ,
\end{align*}
so that $f_N \to f$ in~$\cfunctionsp$
and $F_N$ is a holomorphic function in the top half
\[\cstrip^{\topsym} := \set{z\in\cstrip \; \big| \; \im(w)>0} \]
which extends smoothly to the boundary~$\bdry \cstrip^{\topsym}$, 
coincides with~$f_N$ on the
cross-section~$\ccrosssec \subset \bdry \cstrip^{\topsym}$,
and has Riemann boundary values~\eqref{eq: RBV in strip} on the
left and right boundaries. 
Moreover, $F_N(x+\ii y)$ decays exponentially as~$y\to+\infty$.
By Cauchy's integral theorem for~$F_N^2$ 
along~$\bdry \cstrip^{\topsym}$ with Riemann boundary values
$F_N(\pm\half + \ii y)^2 = \pm \ii \, \big| F_N(\pm\half + \ii y) \big|^2$, 
we get
\begin{align*}
\int_{\aaa}^{\bbb} f_N(x)^2 \, \ud x \;
= \; & + \ii \int_{0}^{+\infty} F_N \big(-\half+\ii y \big)^2 \; \ud y 
    - \ii \int_{0}^{+\infty} F_N \big(+\half+\ii y \big)^2 \; \ud y \\
= \; & \int_{0}^{+\infty} \big| F_N(-\half+\ii y) \big|^2 \; \ud y 
    + \int_{0}^{+\infty} \big| F_N(+\half+\ii y) \big|^2 \; \ud y 
\; \geq \; 0 .
\end{align*}
Since~$f_N \to f$ in~$\cfunctionsp$ as $N \to \infty$, we conclude
\[ \int_{\aaa}^{\bbb} f(x)^2 \, \ud x \geq 0 . \]
Entirely similar arguments in the left and the right legs yield
\[ \int_{\aaa}^{0} f(x)^2 \, \ud x \leq 0 , \qquad
\int_{0}^{\bbb} f(x)^2 \, \ud x \leq 0 . \]
Together these observations imply that $\int_{\aaa}^{\bbb} f(x)^2 \, \ud x = 0$, 
and in particular also 
\[ \int_{0}^{+\infty} \Big| F_N \big(\half+\ii y \big) \Big|^2 \; \ud y
\; \longrightarrow \; 0 
\qquad \text{ as } N \to \infty . \]

Then consider $F \colon \cstrip^{\topsym} \to \bC$
defined as~$F=\sum_{k'<0} c_{k'} \ccfFun_{k'}$. 
Since~$\left| \ccfFun_{k'}(x+iy) \right| \leq e^{-k'\pi y}$,
we have $F_N \to F$ 
uniformly on $\set{ z \in \cstrip \; \big| \; \im(z) > \eps}$ 
for any~$\eps>0$, and~$F$ is holomorphic in~$\cstrip^{\topsym}$
and smooth 
in~$\overline{\cstrip^{\topsym}}\setminus\ccrosssec$. 
But we now have
\begin{align*}
\int_{\eps}^{+\infty} \Big| F \big(\half+\ii y \big) \Big|^2 \; \ud y
= \lim_{N \to \infty}
    \int_{\eps}^{+\infty} \Big| F_N \big(\half+\ii y \big) \Big|^2 \; \ud y
= 0
\end{align*}
for any~$\eps>0$, so $F$ vanishes identically on
the right boundary vertical line (similarly for left).
Vanishing on a line segment implies~$F \equiv 0$, and therefore 
we get that~$c_{k'} = 0$ for all~$k'$, and also~$f \equiv 0$.
\end{proof}

\subsubsection*{Pulled-back monomials}

By Lemma~\ref{lem:continuous-uniqueness} above,
the singular parts~\eqref{eq: definition of singular parts}
uniquely characterize a function~$f \in \cfunctionsp$.
It is therefore natural to introduce basis
functions, which have exactly one singular Fourier mode of 
a given order in one of the three extremities 
of the slit-strip, and which are regular in the 
other two extremities. It is easier to first 
construct functions which are a mixture with 
finitely many singular Fourier modes, and to then recursively 
extract the ones with a single singular Fourier mode.

In the upper half plane
\begin{align*}
\bH = \set{ w \in \bC \; \Big| \; \im(w) > 0} ,
\end{align*}
the Riemann boundary values~\eqref{eq:rbv_cf} 
amount to the requirement that the functions are purely imaginary on the real 
axis. Therefore imaginary constant multiples of Laurent monomial functions
centered on the real axis, $w 
\mapsto \ii \, (w-c)^n$, $n \in \bZ$, $c \in \bR$, are appropriate singular 
modes in the half-plane.
Conformal transformation as $\half$-forms preserves the Riemann boundary 
values~\eqref{eq:rbv_cf}.
This guides the construction below.

Consider the conformal map
\begin{align}\label{eq: uniformizing map of the slit-strip}
\mapSS
\colon \; & \cslitstrip \to \bH &
\mapSS(z) 
= \; & \half \, \sqrt{1-e^{-2\ii\pi z}} 
\end{align}
from the slit-strip to the upper half-plane, where the branch of the square root 
is such that it always has a positive imaginary part.
It maps the top extremity of the slit-strip to~$+\ii \cdot \infty$,
and the left and right downwards extremities to~$-\half$ and~$+\half$, 
respectively.
\begin{figure}[tb]
\centering
\subfigure[The mapping~$z \mapsto \ii \, e^{-\ii \pi z}$ is conformal from the 
the strip~$\cstrip$ to the half plane~$\bH$: the level lines of its real and 
imaginary parts are shown here.]
{
    \includegraphics[width=.33\textwidth]{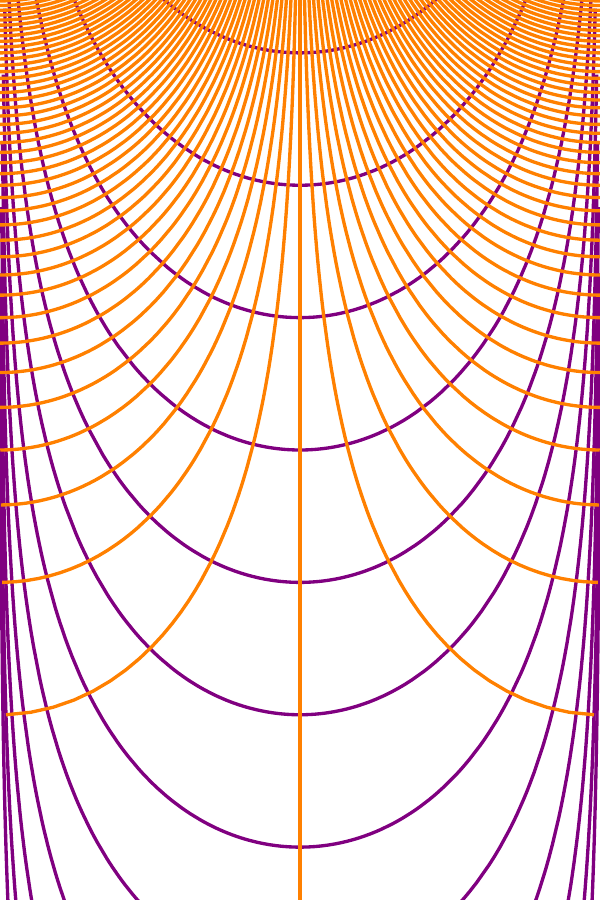}
	\label{sfig: confmap to strip}
}
\hspace{2.5cm}
\subfigure[
The mapping~$z \mapsto \mapSS(z)$ is conformal from the 
the slit-strip~$\cslitstrip$ to the half plane~$\bH$: the level lines of its 
real and imaginary parts are shown here.]
{
    \includegraphics[width=.33\textwidth]{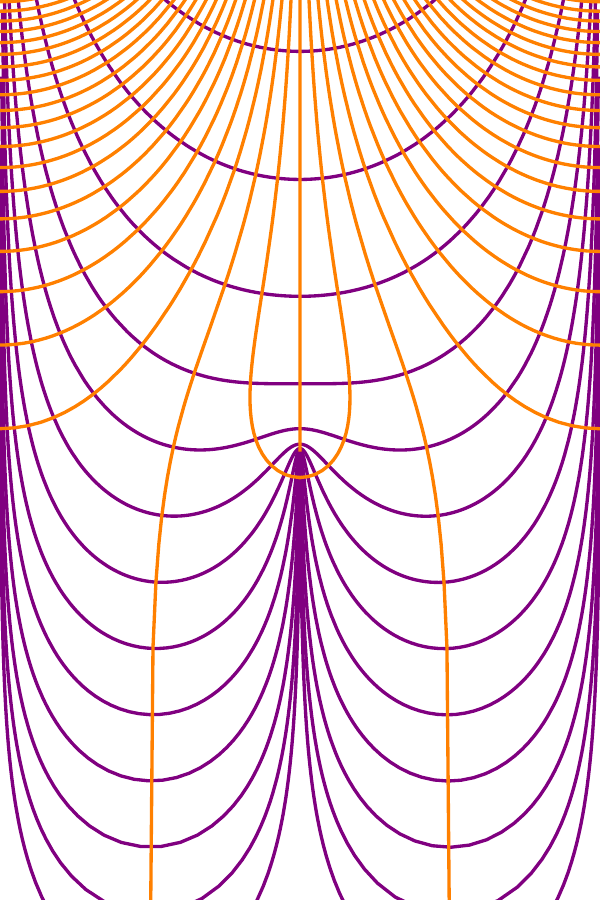}
	\label{sfig: confmap to slit-strip}
}
\caption{Illustrations of conformal maps to the strip and the slit-strip.
}
\label{fig: conformal maps to strip and slit-strip}
\end{figure}
\begin{figure}[tb]
\centering
\includegraphics[width=.7\textwidth]{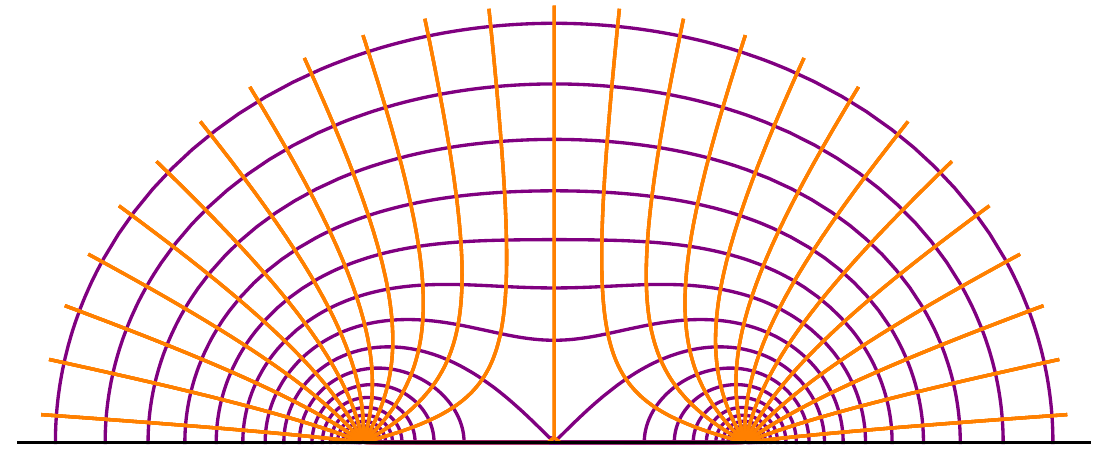}
\caption{Illustration of the conformal map $\mapSS \colon 
\cslitstrip \to \bH$ from the slit-strip to the half-plane:
the images in the half-plane of horizontal and vertical lines in the slit-strip 
are shown here.
}
\label{fig: conformal map from slit-strip}
\end{figure}
Illustrations of this conformal 
map~\eqref{eq: uniformizing map of the slit-strip} are given in
Figures~\ref{fig: conformal maps to strip and slit-strip}
and~\ref{fig: conformal map from slit-strip}.
Note the asymptotics in the three extremities of the slit-strip
\begin{align*}
\mapSS(z) = \; & \frac{\ii}{2} \, e^{-\pi\ii z} + \OO(e^{\pi \ii z}) , &
\mapSS'(z) = \; & \frac{\pi}{2} \, e^{-\pi\ii z} + \OO(e^{\pi \ii z})&
\text{ in the top,} \\
\mapSS(z) = \; & 
    \frac{-1}{2}+\frac{1}{4} \, e^{-2\pi\ii z} + \OO(e^{-4\pi\ii z}) , &
\mapSS'(z) = \; &
    \frac{-\pi\ii}{2}e^{-2\pi\ii z} + \OO(e^{-4\pi\ii z}) &
\text{ in the left leg,} \\
\mapSS(z) = \; & 
    \frac{+1}{2}-\frac{1}{4} \, e^{-2\pi\ii z} + \OO(e^{-4\pi\ii z}) , &
\mapSS'(z) = \; &
    \frac{+\pi\ii}{2}e^{-2\pi\ii z} + \OO(e^{-4\pi\ii z}) &
\text{ in the right leg.}
\end{align*}

To use the conformal map~$\mapSS$ for unique pull-backs of $\half$-forms,
let us fix a branch of the square root of the derivative\footnote{%
Note that since~$\mapSS$ is conformal, the derivative~$\mapSS'$ is non-vanishing 
in the whole domain~$\cslitstrip$. As this domain is simply connected,
it is possible to choose a single-valued branch of~$\sqrt{\mapSS'}$ 
on~$\cslitstrip$. There are two possible branch choices, which differ by a 
sign, and for definiteness we fix one of them here.
In our final results of the series, an even number of these square
roots will appear as factors, so the results will actually be 
independent of the branch choice made here.}
\begin{align*}
\sqrt{\mapSS'} \, \colon \cslitstrip \to \bC \setminus \set{0}
\end{align*}
so that $\sqrt{\mapSS'(x^\lftsym_0)} \in e^{\ii \pi / 4} \, \bR_+$ for 
boundary points~$x^\lftsym_0$ on the left boundaries
and $\sqrt{\mapSS'(x^\rgtsym_0)} \in e^{-\ii \pi / 4} \, \bR_+$ for 
boundary points~$x^\rgtsym_0$ on the right boundaries.\footnote{%
The left boundary is taken to include both the 
case~$\re(x^\lftsym_0)=-\half$, and
the case that $x^\lftsym_0$ is a prime end on the right side of the slit.
The right boundary is taken to include the
case~$\re(x^\rgtsym_0)=+\half$, and 
the case that $x^\rgtsym_0$ is a prime end on the left side of the slit.}
With this branch choice the asymptotics in the 
three extremities of the slit-strip are
\begin{align*}
\sqrt{\mapSS'(z)} 
= \; & \sqrt{\frac{\pi}{2}} \, e^{-\ii \pi z / 2} 
    + \OO (e^{\ii 3 \pi z / 2}) &
\text{ in the top,} & \\
\sqrt{\mapSS'(z)} 
= \; & \sqrt{\frac{\pi}{2}} \, e^{-\ii \pi (z+\frac{1}{4})} 
    + \OO (e^{-\ii 3 \pi z}) &
\text{ in the left leg,} & \\
\sqrt{\mapSS'(z)} 
= \; & \sqrt{\frac{\pi}{2}} \, e^{-\ii \pi (z-\frac{1}{4})} 
    + \OO (e^{-\ii 3 \pi z}) &
\text{ in the right leg.} &
\end{align*}
Define, for $k \in \pm \poshalfint$
(although we will primarily use the case of positive 
half-integer $k \in \poshalfint$), functions
$\cmixPoleT{k} , \cmixPoleL{k} , \cmixPoleR{k} \colon \cslitstrip \to \bC$ by 
the formulas
\begin{align}
\nonumber
\cmixPoleT{k}(z)
    := \; & \ii \; \mapSS(z)^{k-\half} \; \sqrt{\mapSS'(z)} \\
\label{eq: pulled-back monomials}
\cmixPoleL{k}(z)
    := \; & \ii \; \big( \mapSS(z) + \half \big)^{-k-\half} 
        \; \sqrt{\mapSS'(z)} \\
\nonumber
\cmixPoleR{k}(z)
    := \; & \ii \; \big( \mapSS(z) - \half \big)^{-k-\half} 
        \; \sqrt{\mapSS'(z)} .
\end{align}
The functions~\eqref{eq: pulled-back monomials} are holomorphic and have 
Riemann boundary values~\eqref{eq: RBV on the slit part} in the 
slit-strip~$\cslitstrip$. Their asymptotics in the corresponding 
extremities are given by
\begin{align*}
\cmixPoleT{k}(z) 
= \; & \ii \, \sqrt{\frac{\pi}{2}} \, \Big(\frac{\ii}{2} \Big)^{k-\half}
        \, e^{-\ii \pi k z} + \OO(e^{-\ii \pi (k-2) z}) &
\text{ in the top,} & \\
\cmixPoleL{k}(z) 
= \; &  e^{\ii \pi / 4} \, \sqrt{\frac{\pi}{2}} \; 4^{k+\half}
        \; e^{\ii 2 \pi k z} + \OO(e^{\ii 2 \pi (k-1) z}) &
\text{ in the left leg,} & \\
\cmixPoleR{k}(z) 
= \; & - e^{-\ii \pi / 4} \sqrt{\frac{\pi}{2}} \, (-4)^{k+\half} 
        \, e^{\ii 2 \pi k z} + \OO(e^{\ii 2 \pi (k-1) z}) &
\text{ in the right leg.} &
\end{align*}
In particular, for any positive half-integer~$k \in \poshalfint$,
the singular parts of $\cmixPoleT{k}, \cmixPoleL{k}, \cmixPoleR{k}$ in the 
corresponding extremities contain finitely many singular Fourier modes.
Moreover, it is easy to see that these functions are regular in the other two 
extremities.

\subsubsection*{Pure pole functions in the slit-strip}
From the pulled-back monomials above, 
through a simple upper triangular transformation, we can construct
functions characterized by a single Fourier mode as their singular parts.
The functions are illustrated in 
Figure~\ref{fig: continuous pole functions}.
\begin{prop}
\label{prop: pure pole functions}
For all positive half-integers~$k \in \poshalfint$,
there exist functions
\[ \cpoleT{k}, \; \cpoleL{k}, \; \cpoleR{k} \; \in \; \cfunctionsp \]
characterized by the following singular parts:
\begin{align}
\nonumber
\cprTpole(\cpoleT{k}) = \; & \ccffun_k , &
\cprLpole(\cpoleT{k}) = \; & 0 , &
\cprRpole(\cpoleT{k}) = \; & 0 , \\ 
\label{eq:cpole-asymptotic}
\cprTpole(\cpoleL{k}) = \; & 0 , &
\cprLpole(\cpoleL{k}) = \; & \ccffun^{\lftsym}_{-k} , &
\cprRpole(\cpoleL{k}) = \; & 0 , \\ 
\nonumber
\cprTpole(\cpoleR{k}) = \; & 0 , &
\cprLpole(\cpoleR{k}) = \; & 0 , &
\cprRpole(\cpoleR{k}) = \; & \ccffun^{\rgtsym}_{-k} .
\end{align}
The functions~$\cpoleT{k}, \cpoleL{k}, \cpoleR{k}$ are the restrictions to the 
cross-section~$\ccrosssec$ of globally defined holomorphic 
functions $\cPoleT{k}, \cPoleL{k}, \cPoleR{k} \colon \cslitstrip \to \bC$ with 
Riemann boundary values on the slit-strip, which we call the \term{pure pole 
functions}.
We can express these pure pole functions as finite linear combinations of the 
pulled-back monomials~\eqref{eq: pulled-back monomials} and vice versa,
\begin{align*}
\cPoleT{k} 
= \; & \sum_{0 < k'\leq k} \mixToPoleT{k}{k'} \, \cmixPoleT{k'} , &
\cPoleL{k} 
= \; & \sum_{0 < k'\leq k} \mixToPoleL{k}{k'} \, \cmixPoleL{k'} , &
\cPoleR{k} 
= \; & \sum_{0 < k'\leq k} \mixToPoleR{k}{k'} \, \cmixPoleR{k'} , & \\
\cmixPoleT{k} 
= \; & \sum_{0 < k'\leq k} \poleToMixT{k}{k'} \, \cPoleT{k'} , &
\cmixPoleL{k} 
= \; & \sum_{0 < k'\leq k} \poleToMixL{k}{k'} \, \cPoleL{k'} , &
\cmixPoleR{k} 
= \; & \sum_{0 < k'\leq k} \poleToMixR{k}{k'} \, \cPoleR{k'} , & 
\end{align*}
with certain real coefficients
$\mixToPoleT{k}{k'}, \mixToPoleL{k}{k'}, \mixToPoleR{k}{k'},
\poleToMixT{k}{k'}, \poleToMixL{k}{k'}, \poleToMixR{k}{k'}$.
\end{prop}
\begin{figure}
\centering
\subfigure[The pole function~$\cpoleT{k}$ for~$k=1/2$.] 
{
    \includegraphics[width=.4\textwidth]{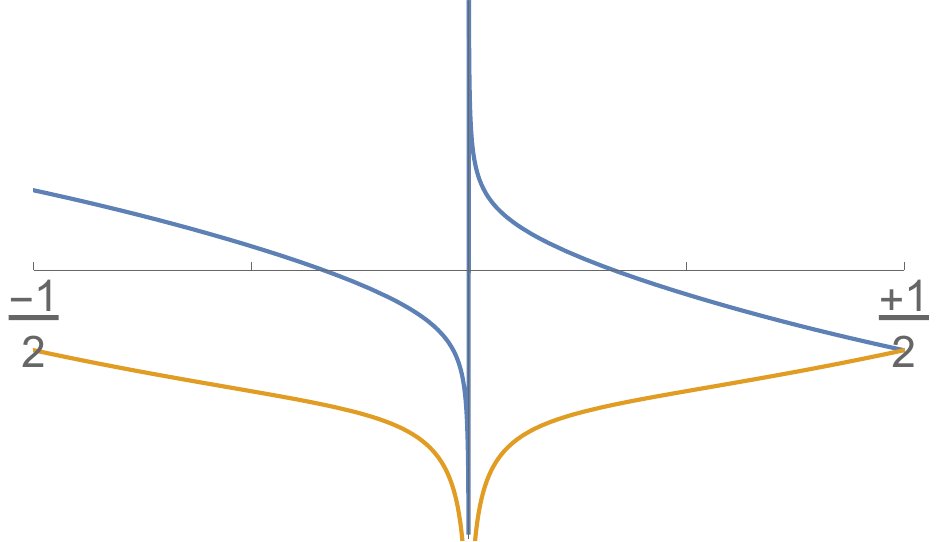}
	\label{sfig: ceigf P 1}
}
\hspace{2.0cm}
\subfigure[The pole function~$\cpoleR{k}$ for~$k=1/2$.] 
{
    \includegraphics[width=.4\textwidth]{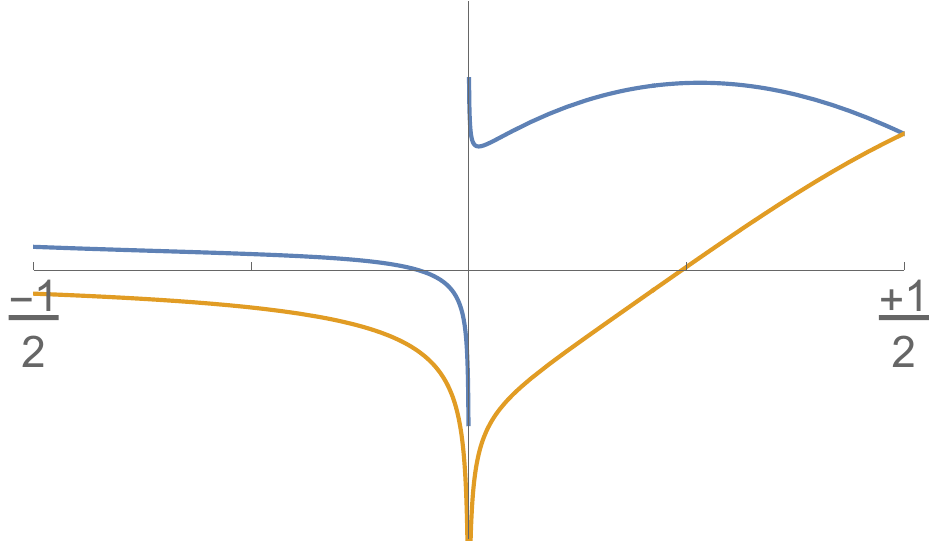}
	\label{sfig: ceigf M 1}
} \\
\subfigure[The pole function~$\cpoleT{k}$ for~$k=3/2$.] 
{
    \includegraphics[width=.4\textwidth]{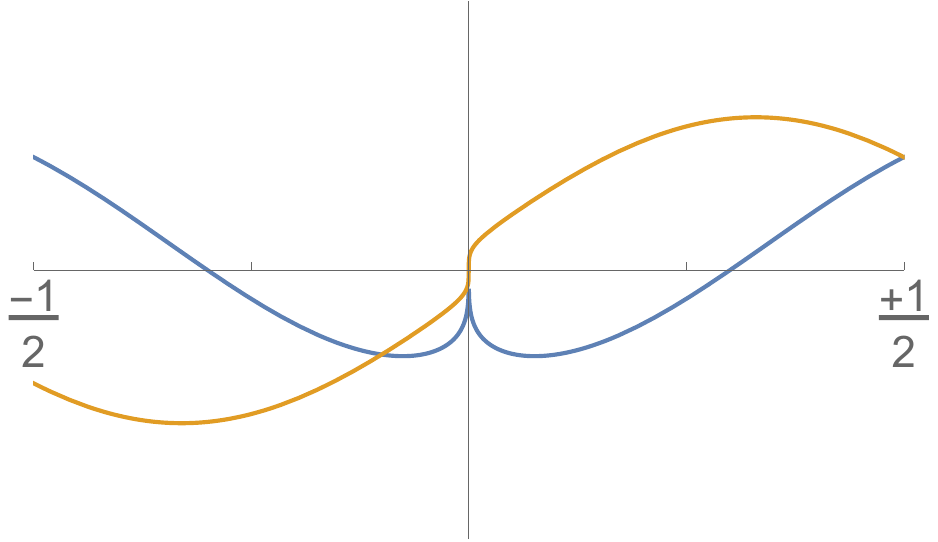}
	\label{sfig: ceigf P 3}
}
\hspace{2.0cm}
\subfigure[The pole function~$\cpoleR{k}$ for~$k=3/2$.] 
{
    \includegraphics[width=.4\textwidth]{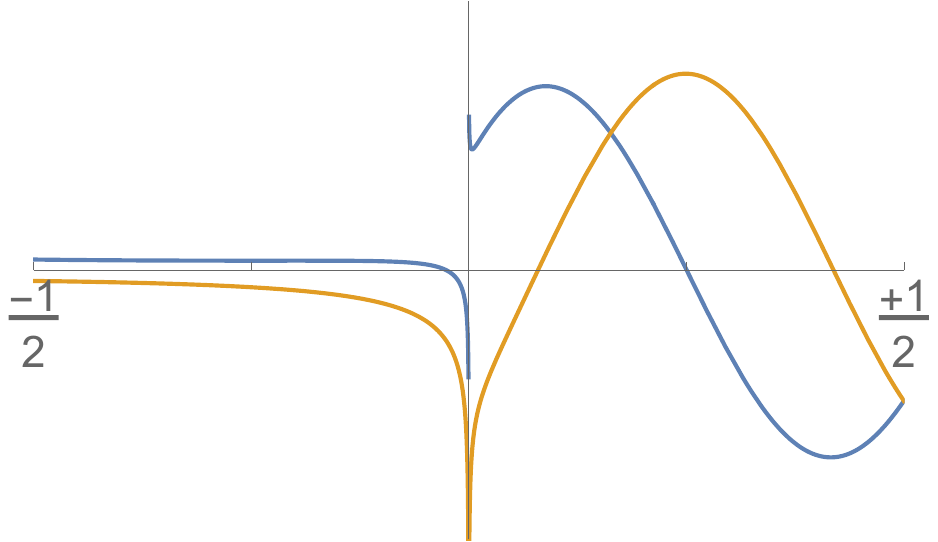}
	\label{sfig: ceigf M 3}
} \\
\subfigure[The pole function~$\cpoleT{k}$ for~$k=5/2$.] 
{
    \includegraphics[width=.4\textwidth]{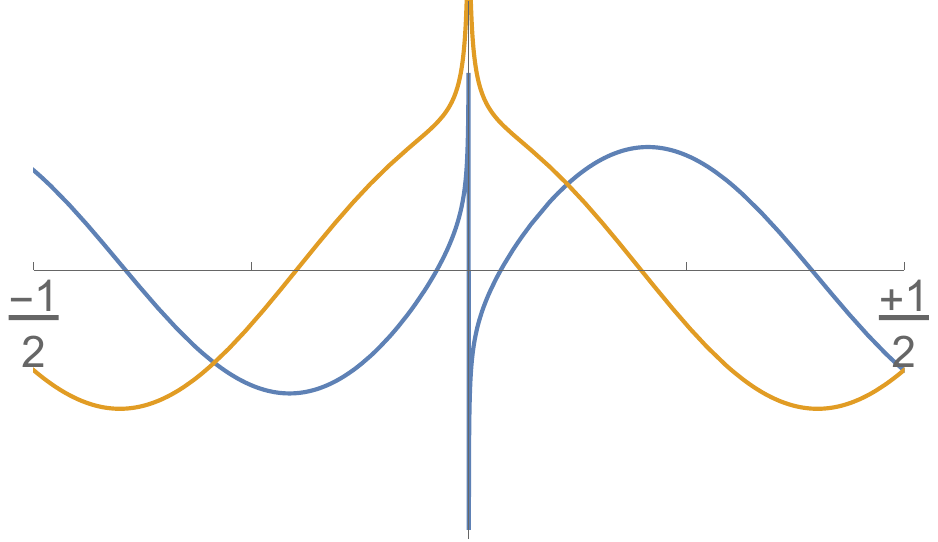}
	\label{sfig: ceigf P 5}
}
\hspace{2.0cm}
\subfigure[The pole function~$\cpoleR{k}$ for~$k=5/2$.] 
{
    \includegraphics[width=.4\textwidth]{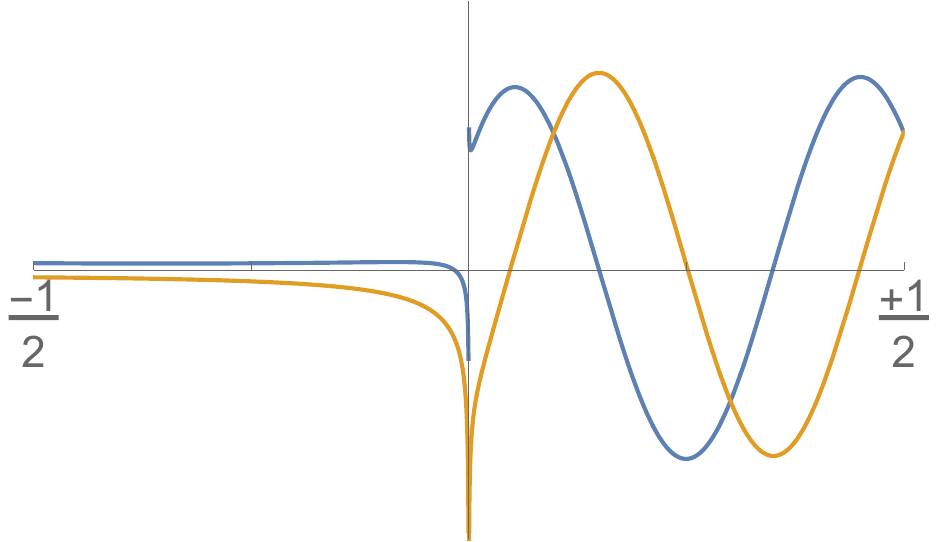}
	\label{sfig: ceigf M 5}
} \\
\subfigure[The pole function~$\cpoleT{k}$ for~$k=7/2$.] 
{
    \includegraphics[width=.4\textwidth]{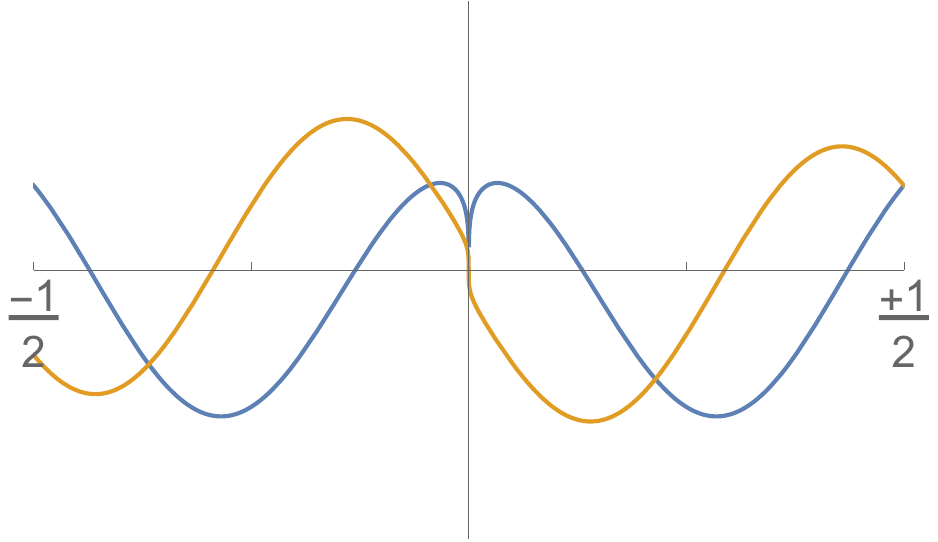}
	\label{sfig: ceigf P 7}
}
\hspace{2.0cm}
\subfigure[The pole function~$\cpoleR{k}$ for~$k=7/2$.] 
{
    \includegraphics[width=.4\textwidth]{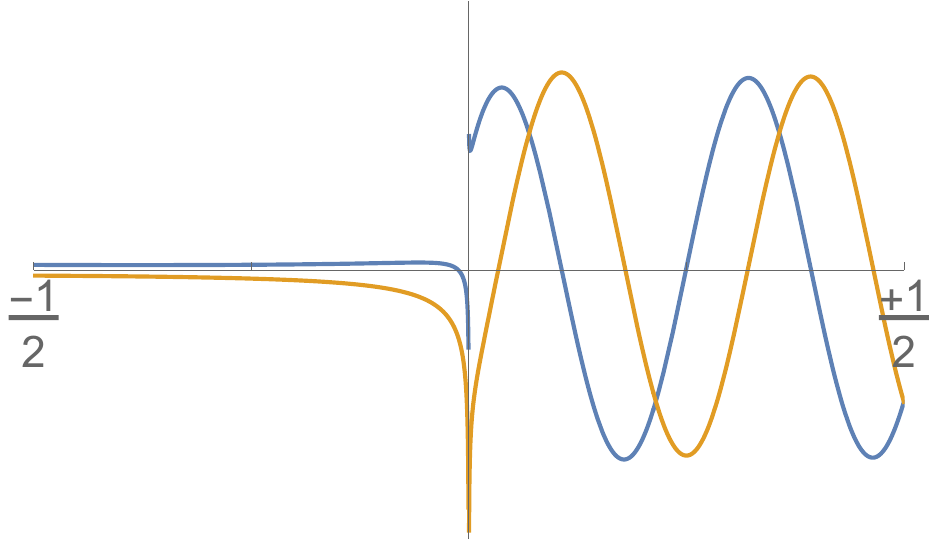}
	\label{sfig: ceigf M 7}
} \\
\label{fig: continuous pole functions}
\caption{Real and imaginary parts (blue and orange)
of the restrictions~$\cpoleT{k}, \cpoleR{k}$ to~$\ccrosssec$
of the pure pole functions~$\cPoleT{k}, \cPoleR{k}$ for the 
top and right extremities. (The pure pole functions~$\cPoleL{k}$
for the left leg, and their restrictions~$\cpoleL{k}$, 
are similar to the ones for the right leg.)}
\end{figure}
\begin{proof}

We sketch the recursive construction of $\cPoleT{k}$ in terms 
of pulled-back monomials; the other cases are analogous, and the resulting 
functions are uniquely characterized by \eqref{eq:cpole-asymptotic} from 
Lemma~\ref{lem:continuous-uniqueness}.

First, note that we may 
define~$\cPoleT{1/2}=-\sqrt{\frac{2}{\pi}}\cmixPoleT{1/2}$. 
If~$\cPoleT{k}$ for~$k=1/2,\ldots,k'$ are already constructed,
we may define
\begin{align*}
\cPoleT{k'+1}
= \; & \frac{(-1)^{k+\frac12}2^k}{\sqrt{\pi}} \, \cmixPoleT{k'+1}
    - \sum_{0<k\leq k'} 
      \innprod{\ccffun_k}{\tilde{\mathfrak{q}}_{\topsym;k'+1}} \, \cPoleT{k},
\end{align*}
where~$\tilde{\mathfrak{q}}_{\topsym;k}$ is the restriction of~$\cmixPoleT{k}$ 
to the real line.

Given the existence of pure pole functions, we may express the pulled-back 
monomials in terms of them as follows. By definition, there exist real 
coefficients~$a_{-k'}$ such that
\begin{align*}
\cPoleT{k}(z) 
= \; & \ccfFun_{k}(z) 
    + \sum_{k'>0} a_{-k'} \, \ccfFun_{-k'}(z)
\qquad \text{ for } z \in \cslitstrip \cap \bH . 
\end{align*}

Consider the conformal map~$\tilde{w} = \ii \, e^{-\pi \ii z} $ 
from the slit-strip to the slit upper half-plane~${\mathbb H \setminus [0,\ii]}$. 
Under a change of variable analogous to~\eqref{eq: pulled-back monomials},
\begin{align*}
\ccfFun_{k}(z(\tilde{w}))\cdot \left[ \frac{dz}{d\tilde{w}}\right]^{\frac12}
& =\frac{e^{-\ii\pi k}}{\sqrt{\pi}}\tilde{w}^{k-\frac12};
\end{align*}
and thus
\begin{align*}
\cPoleT{k}(z(\tilde{w}))\cdot \left[ \frac{dz}{d\tilde{w}}\right]^{\frac12}
& = \frac{e^{-\ii\pi k}}{\sqrt{\pi}}\tilde{w}^{k-\frac12}+\sum_{k'>0} 
\frac{e^{\ii\pi k'}a_{-k'}}{\sqrt{\pi}}\tilde{w}^{-k'-\frac12} .
\end{align*}
Since $\mapSS(z)=\frac12\sqrt{1+\tilde{w}(z)^2}$ and $\cmixPoleT{k}$ is 
defined precisely as pullbacks of half-plane monomials,
\begin{align*}
\cmixPoleT{k}(z(\tilde{w}))\cdot 
    \left[ \frac{dz}{d\tilde{w}}\right]^{\frac12}
= \; & \ii \left(\frac{\sqrt{1+\tilde{w}^2}}{2} \right)^{k-\frac12}
    \left[ \frac{\tilde{w}}{2\sqrt{1+\tilde{w}^2}}\right]^{\frac12}
= \frac{\ii \sqrt{\tilde{w}}}{2^k}
    \left( 1+\tilde{w}^2 \right)^{\frac{k-1}{2}} \\
= \; & \frac{\ii \tilde{w}^{k-\frac12} }{2^k}
    \left(1+ {\frac{k-1}{2} \choose{1}}\tilde{w}^{-2}+\cdots \right)
\mbox{ as }\tilde{w}\to\infty.
\end{align*}
The above binomial expansion only has finitely many terms with 
nonnegative powers of~$w$, and their coefficients are precisely the 
coefficients in the expansion of~$\cmixPoleT{k}$ in terms of 
the pure pole functions:
\begin{align*}
\cmixPoleT{k} 
= \; & \sum_{0 < k' \leq k} 
    \frac{\sqrt{\pi} e^{\ii \pi k'}\ii}{2^{k}}{\frac{k-1}{2}
    \choose{\frac{k-k'}{2}}} \cPoleT{k'} ,
\end{align*}
where the binomial coefficient is taken as zero when $\frac{k-k'}{2}$ 
is not an integer. The other cases are similar.
\end{proof}
\begin{rmk}
The coefficients in Proposition~\ref{prop: pure pole functions}
implement a change of basis, which reflects the relationship of 
the geometry of the slit-strip~$\cslitstrip$ 
and the half-plane~$\bH$ via the 
conformal map~\eqref{eq: uniformizing map of the slit-strip} between them.
These constants of geometric origin appear in our main result, which 
reconstructs the structure constants of the vertex operator algebra
(the Ising conformal field theory) from
the scaling limits of the fusion coefficients of the Ising model in 
lattice slit-strips.

Given that these constants thus account for the most nontrivial 
geometric input to our main result, we note that 
explicit expressions for them can be obtained similarly
to the above proof by expanding monomials of the map 
(between the slit half-plane and the half-plane)
$\psi(\tilde{w}) =  
\frac{\sqrt{1+\tilde{w}^2}}{2}$ and its inverse.
In the following, we will write $w$ for the variable in
the half-plane, and $\tilde{w}$ for the slit half-plane;
the boundary point $0$ in the slit half-plane corresponds
to two prime ends $0_-,0_+$, approached from left and 
right respectively.

Also note that we are considering holomorphic functions
in a neighborhood of a boundary pole where they have 
Riemann boundary values. 
Using a pull-back to the half-plane, we may define 
their Schwarz reflection, which may be uniquely 
expanded in a Laurent series: this guarantees the 
existence and uniqueness of the expansions.

The constants can then be expressed as follows:
\begin{align}
\mixToPoleT{k}{k'} 
&= \mbox{ $(-\ii)\cdot$(coefficient of $w^{k'-\frac12}$ 
  of $\frac{e^{-\ii\pi k}}{\sqrt{\pi}}{\left(\sqrt{4w^2-1}\right)^{k-\frac12}} 
  \left[ \frac{4w}{\sqrt{4w^2-1}}\right]^{\frac12}$ around ${w}=\infty$)}, 
\\ \nonumber
\mixToPoleL{k}{k'} 
&= \mbox{ $(-\ii)\cdot$(coefficient of $(w+\frac12)^{-k'-\frac12}$ 
  of $\sqrt{2}\ii \left({\sqrt{4w^2-1}}\right)^{-2k-\frac12} 
  \left[ \frac{4w}{\sqrt{4w^2-1}}\right]^{\frac12}$ around ${w}=-\frac12$)},
\\ \nonumber
\mixToPoleR{k}{k'} 
&= \mbox{ $(-\ii)\cdot$(coefficient of $(w-\frac12)^{-k'-\frac12}$ 
  of $\sqrt{2}e^{-\pi \ii k} \left({\sqrt{4w^2-1}}\right)^{-2k-\frac12} 
  \left[ \frac{4w}{\sqrt{4w^2-1}}\right]^{\frac12}$ around ${w}=\frac12$)} , \\ 
\nonumber
\poleToMixT{k}{k'} 
&= \mbox{ $e^{\ii \pi k'}\sqrt{\pi}\cdot$
  (coefficient of $\tilde{w}^{k'-\frac12}$
  of $\ii\left(\frac{\sqrt{1+\tilde{w}^2}}{2} \right)^{k-\frac12}
  \left[ \frac{\tilde{w}}{2\sqrt{1+\tilde{w}^2}}\right]^{\frac12}$ 
  around $\tilde{w}=\infty$)}, 
\\ \nonumber
\poleToMixL{k}{k'} 
&= \mbox{ $\frac{1}{\sqrt{2}\ii}\cdot$(coefficient of $\tilde{w}^{2k'-\frac12}$ 
  of $\ii\left(\frac{\sqrt{1+\tilde{w}^2}}{2}+\frac12 \right)^{-k-\frac12}
  \left[ \frac{\tilde{w}}{2\sqrt{1+\tilde{w}^2}}\right]^{\frac12}$ 
  around $\tilde{w}=0_-$)},
\\ \nonumber
\poleToMixR{k}{k'} 
&=\mbox{ $\frac{e^{\pi \ii k'}}{\sqrt{2}}\cdot$(coefficient 
  of $\tilde{w}^{2k'-\frac12}$ 
  of $\ii\left(\frac{\sqrt{1+\tilde{w}^2}}{2}-\frac12 \right)^{-k-\frac12}
  \left[ \frac{\tilde{w}}{2\sqrt{1+\tilde{w}^2}}\right]^{\frac12}$ 
  around $\tilde{w}=0_+$)} .
\end{align}
In particular, we have
\begin{align*}
\mixToPoleT{k}{k'} 
& = \left(\frac{-1}{4}\right)^{\frac{k-k'}{2}}{{\frac{k-1}2} 
    \choose{\frac{k-k'}{2} }}, &
\poleToMixT{k}{k'} 
& = \frac{\sqrt{\pi} e^{\ii \pi k'}\ii}{2^{k}}{\frac{k-1}{2} 
    \choose{\frac{k-k'}{2}}} .
\end{align*}
\end{rmk}


\section{Discrete function spaces and decompositions}%
\label{sec: discrete complex analysis}

In this section we study functions on discretized domains,
which have the properties analogous to holomorphicity and
Riemann boundary values. We introduce the 
spaces of functions analogous to the continuum case,
and find analogous distinguished functions:
vertical translation eigenfunctions in the lattice strip,
and functions with prescribed singularities in the
extremities of the lattice slit-strip.

As the appropriate notion of discrete holomorphicity 
we use s-holomorphicity. This notion and its powerful uses together
with Riemann boundary values were pioneered 
by Smirnov~\cite{Smirnov-towards_conformal_invariance, 
Smirnov-conformal_invariance_in_RCM_1}, and have been developed into
an extremely powerful tool for the study of the Ising 
model~\cite{CS-discrete_complex_analysis_on_isoradial_graphs, 
CS-universality_in_Ising}.
We have chosen a route to the main result of this
series which avoids entirely the use of the notions of 
s-holomorphic poles~\cite{HS-energy_density, 
Hongler-thesis} and s-holomorphic spinors~\cite{Izyurov-thesis, 
CI-holomorphic_spinor_observables, 
CHI-conformal_invariance_of_spin_correlations}
and square root singularities.
The quintessential trick
for s-holomorphic solutions to Riemann boundary value problem
is the ``imaginary part of the integral of the square'',
and we will be able to employ it largely in its most standard
incarnation in Section~\ref{sec: convergence of functions}.

In Section~\ref{ssec: discrete domains}
we introduce the discrete domains, and in
Section~\ref{ssec: discrete complex analysis definitions}
give the definitions of the
needed notions of discrete complex analysis and of the 
space of functions of interest to us. 
In Section~\ref{sub: functions in lattice strip}
we study the vertical translation eigenfunctions in the
strip, and the associated decomposition of the function space.
In Section~\ref{sub: distinguished s-holomorphic functions in the slit-strip}
we find the distinguished functions in the lattice slit-strip,
which have prescribed singularities in the extremities.

\subsection{The lattice strip and the lattice slit-strip}
\label{ssec: discrete domains}

The lattice analogues of the continuum strip and slit-strip 
domains~$\cstrip$ and~$\cslitstrip$ will be certain 
square grid discretizations of these domains.

Fix two integers
\begin{align*}
\lft,\rgt \in \bZ , \qquad
\lft < 0 < \rgt ,
\end{align*}
which represent the (horizontal) positions of the 
left and right boundaries. The slit will always be placed
at the horizontal position~$\mdpt$. The \term{width} of the strip
(in lattice units) is
\begin{align}\label{eq: strip width}
\width := \rgt - \lft \; \in \, \bN .
\end{align}
For simplicity of notation, we only carry the superscript label
for width~$\width$ in the notation to indicate the discretization,
although
information about~$\lft$ and~$\rgt$ is in fact important as well.
We mostly care about a symmetric situation 
(equal widths for the left and right substrip) in which
$\rgt = -\lft = \half \width$ and the limit of large even integer
widths~$\width \to \infty$,
but more general choices are possible and at times in fact clearer.

\subsubsection*{The lattice strip}

\begin{figure}[tb]
\centering
\subfigure[The square grid strip~$\dstrip$.] 
{
  \includegraphics[width=.35\textwidth]{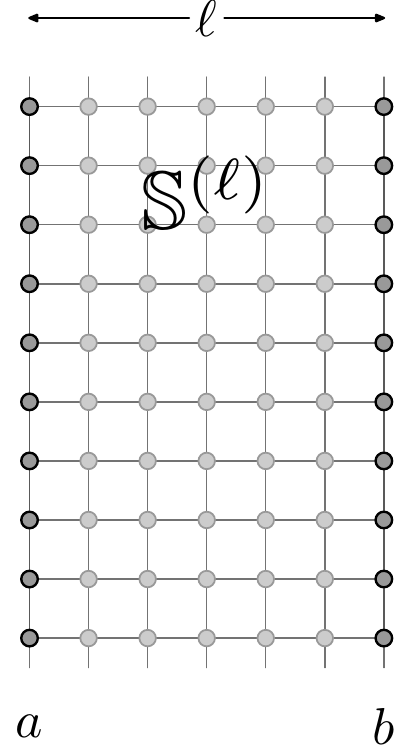}
  \label{sfig: lattice strip}
}
\hspace{2.5cm}
\subfigure[The square grid slit-strip~$\dslitstrip$.] 
{
  \includegraphics[width=.35\textwidth]{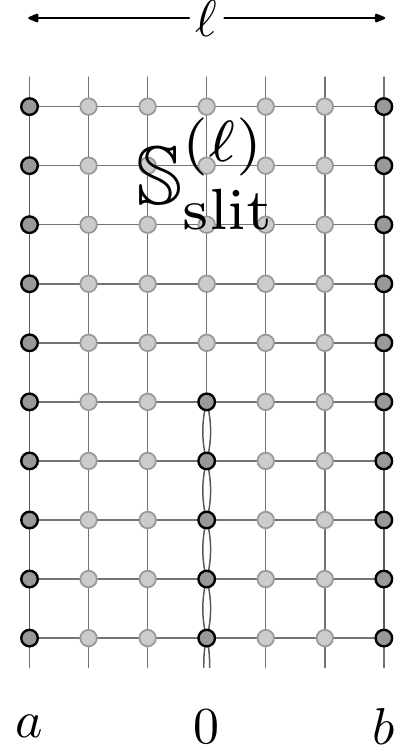}
  \label{sfig: lattice slit-strip}
}
\caption{The discrete strip and slit-strip graphs.
}
\label{fig: square grid strip and slit-strip}
\end{figure}
A discretized version of the cross-section~$\ccrosssec$
is the integer interval
\begin{align}\label{eq: lattice interval}
\crosssec 
:= \; & \set{\lft , \lft+1 , \lft+2 , \ldots , 
    \rgt - 1 , \rgt} 
    =: \dinterval{\lft}{\rgt} .
\end{align}
The \term{lattice strip} is then defined as the graph with vertex set
\begin{align}\label{eq: lattice strip}
\dstrip 
:= \; & \set{ x + \ii y \; \Big| \; 
    x \in \crosssec , \; y \in \bZ }
= \dinterval{\lft}{\rgt} \times \bZ ,
\end{align}
and with nearest neighbor edges as in
Figure~\ref{sec: discrete complex analysis}.\ref{sfig: lattice strip}~---
in other words, the lattice strip is
seen as an induced subgraph~$\dstrip \subset \bZ^2$ of
the ordinary square lattice.

The set of all edges of the lattice strip is denoted by~$\dE(\dstrip)$.
We will identify edges with their midpoints, so that
vertical edges are of the form~$x + \ii y' \in \dE(\dstrip)$
with $x \in \crosssec$ and $y' \in \bZ + \half$,
and horizontal edges are of the form $x' + \ii y \in \dE(\dstrip)$
with $y \in \bZ$ and $x' \in \crosssecdual$
in the half-integer interval defined by
\begin{align}\label{eq: dual lattice interval}
\crosssecdual 
:= \; & \set{\lft + \frac{1}{2} , \lft + \frac{3}{2} , \ldots ,
			\rgt - \frac{3}{2} , \rgt - \frac{1}{2}} 
	=: \dintervaldual{\lft}{\rgt} .
\end{align}
In fact, it is this half-integer interval~$\crosssecdual$
on which our functions will be defined.

\subsubsection*{Lattice slit-strip}

The lattice slit-strip will be the (multi-)graph with the 
same set of vertices
\begin{align}\label{eq: lattice strip}
\dslitstrip 
:= \; & \set{ x + \ii y \; \Big| \; 
    x \in \crosssec , \; y \in \bZ }
= \dinterval{\lft}{\rgt} \times \bZ 
\end{align}
as the lattice strip, and otherwise also 
the same set of edges, except that there are double 
edges between nearest neighbors $\mdpt + \ii y$ and $\mdpt + \ii (y-1)$
for $y \leq 0$, i.e., along the slit. This is illustrated in 
Figure~\ref{sec: discrete complex analysis}.\ref{sfig: lattice slit-strip}.
The two different edges 
between nearest neighbors on the slit part have exactly the
same roles as the two different prime-ends corresponding to the
same boundary point on the slit in the continuum
slit-strip~$\cslitstrip$ of 
Figure~\ref{sec: function spaces}.\ref{sfig: cont slit-strip}:
one is thought to belong to (the boundary of) the left substrip and the other 
to (the boundary of) the right substrip.

The set of all edges of the lattice slit-strip is 
denoted by~$\dE(\dslitstrip)$.
Despite the presence of multi-edges, we continue to abuse 
the notation and usually
label edges of the lattice slit-strip by their 
midpoints. 
For any pair of edges along the slit which have coinciding
midpoints, we trust that it will always be sufficiently clear from 
the context which one of the two edges is meant (it should be clear
whether we are considering the left or the right substrip).

\subsection{S-holomorphicity, Riemann boundary values, and function spaces}
\label{ssec: discrete complex analysis definitions}

In this discrete setting, functions will be defined on the
edges of the graph.
We allow for functions defined on subgraphs as well,
so let $\dV$ denote the relevant (sub)set of vertices 
($\dV \subset \dstrip = \dslitstrip$)
and $\dE$ the relevant (sub)set of edges
($\dE \subset \dE(\dstrip)$ or $\dE \subset \dE(\dslitstrip)$).\footnote{We
furthermore always take the subgraphs to consist of all
vertices and edges adjacent to some connected set of faces of the 
lattice strip or the slit-strip, so the 
subsequent definitions in fact contain nontrivial requirements
regarding the functions.}
We consider functions
\begin{align*}
F \colon \dE \to \bC .
\end{align*}

\subsubsection*{S-holomorphicity}

\begin{figure}[tb]
\includegraphics[width=.4\textwidth]{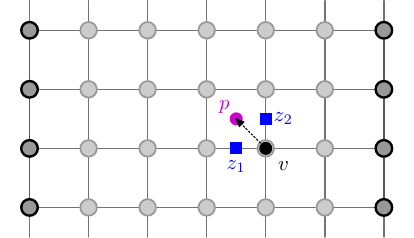}
\hspace{2cm}
\includegraphics[width=.4\textwidth]{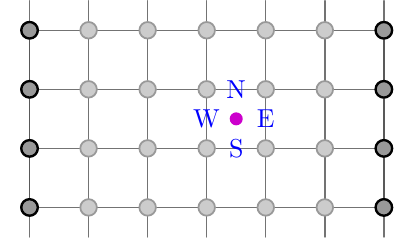}
\caption{S-holomorphicity is a condition for the values of
a function on pairs of edges~$z_1 , z_2$
adjacent to the same face~$p$ and vertex~$v$.}
\label{fig: s-holomorphicity}
\end{figure}
A function~$F \colon \dE \to \bC$ is said to be \term{s-holomorphic}, 
if for all pairs of edges $z_1, z_2 \in \dE$ which are adjacent,
in the sense that both are adjacent to 
the same face~$p$ and the same vertex~$v$, we have
\begin{align}\label{eq: s-holomorphicity}
F(z_1) + \frac{\ii \, |v-p|}{v-p} \, \overline{F(z_1)}
= F(z_2) + \frac{\ii \, |v-p|}{v-p} \, \overline{F(z_2)} .
\end{align}
Equivalently, the values of~$F$ at~$z_1$ and~$z_2$ have
the same projections to the line~$\sqrt{\ii/(v-p)} \, \bR$
in the complex plane. Depending on the position of the
adjacent edges~$z_1, z_2$ with respect to the face, this
line is one among four possibilities. In view of this,
yet another explicit
way of writing the s-holomorphicity condition is the following.
Define the constant
\begin{align*}
\lambda = e^{\ii \pi / 4} = \frac{1+\ii}{\sqrt{2}}
\end{align*}
that we will keep using throughout,
as is common in related literature.
The s-holomorphicity condition is equivalent to requiring
that when $\Nth, \Est, \Sth, \Wst$ are the four edges surrounding
a face as in Figure~\ref{fig: s-holomorphicity}, then
\begin{align*}
F(\Nth) + \eighthrootbar \overline{F(\Nth)} 
    & =  F(\Wst) + \eighthrootbar \overline{F(\Wst)} , \quad &
F(\Nth) + \eighthroot \overline{F(\Nth)} 
    & =  F(\Est) + \eighthroot \overline{F(\Est)} , \\
F(\Sth) + \eighthrootfive \overline{F(\Sth)} 
    & =  F(\Wst) + \eighthrootfive \overline{F(\Wst)} , \quad &
F(\Sth) + \eighthrootthree \overline{F(\Sth)} 
    & =   F(\Est) + \eighthrootthree \overline{F(\Est)} .
\end{align*}

S-holomorphicity implies (but is not implied by) the usual discretized 
Cauchy-Riemann
equations
\begin{align*}
F(z+\frac{1}{2}) - F(z-\frac{1}{2}) 
= -\ii \, \Big( F(z+\frac{\ii}{2}) - F(z-\frac{\ii}{2}) \Big)
\end{align*}
around any face or vertex~$z$ where all the needed values are defined.
This at least gives the interpretation for 
s-holomorphicity as a notion of
discrete holomorphicity, which may not have been apparent
directly from definition~\eqref{eq: s-holomorphicity}.
Note, however, that
s-holomorphicity is an $\bR$-linear condition for the 
complex-valued function~$F$~--- not $\bC$-linear!

\subsubsection*{Riemann boundary values}
The discrete version of Riemann boundary values is defined
very analogously to the continuum version.
If~$z \in \dE$ is a boundary edge, and $\tccw(z)$ 
denotes the unit complex number in the direction of the 
tangent to the boundary oriented counterwlockwise
(i.e., so that the face to the left of the oriented edge is
a part of the discrete domain), then 
$F$ is said to have a \term{Riemann boundary value} at~$z$ if
\begin{align}\label{eq:rbv_df}
F(z) \in \ii \, \tccw(z)^{-1/2} \, \bR.
\end{align}

We will only use Riemann boundary values on the  
boundaries of the lattice strip and lattice slit-strip.
These boundaries are taken to consist
of the vertical edges on the left and on the right 
as well as on the slit.
The vertical edges on the left boundary are of the form
$z=\lft + \ii y'$ with $y' \in \bZ + \half$ and their
counterclockwise tangent points downwards, $\tccw(z)=-\ii$.
The vertical edges on the right boundary are of the form
$z=\rgt + \ii y'$ with $y' \in \bZ + \half$ and their
counterclockwise tangent points upwards, $\tccw(z)=+\ii$.
The requirement of Riemann boundary values for functions
in the lattice strip~$\dstrip$ are thus\footnote{In the article
\cite{HKZ-discrete_holomorphicity_and_operator_formalism}
an unfortunate misprint occurs at the statement of the Riemann
boundary values: the conditions on the left and right
boundaries are reversed. Formulas~\eqref{eq: dRBV in strip}
here are correct with the conventions used in both articles.}
\begin{align}\label{eq: dRBV in strip}
F \big( \lft + \ii y' \big) \in e^{-\ii \pi / 4} \, \bR
\qquad \text{ and } \qquad
F \big( \rgt + \ii y' \big) \in e^{+\ii \pi / 4} \, \bR
\end{align}
for $y' \in \bZ + \half$.

The slit part of the boundary has doubled edges,
one for the left side of the slit (which acts as a right
boundary for the left substrip) and one for the right side 
(which acts as a left boundary to the right substrip).
Denoting these edges respectively by
$z = \mdpt^- + \ii y'$ 
and $z = \mdpt^+ + \ii y'$, 
the Riemann boundary values on the slit part are
\begin{align}\label{eq: dRBV on the slit part}
F \big( \mdpt^- + \ii y' \big) \in e^{+\ii \pi / 4} \, \bR 
\qquad \text{ and } \qquad
F \big( \mdpt^+ + \ii y' \big) \in e^{-\ii \pi / 4} \, \bR
\end{align}
for $y' \in \bZ + \half$, $y' < 0$.

Note the analogy of~\eqref{eq: dRBV in strip}
and~\eqref{eq: dRBV on the slit part}
with~\eqref{eq: RBV in strip} 
and~\eqref{eq: RBV on the slit part},
and note once more that Riemann boundary values 
are~$\bR$-linear conditions for the complex-valued 
function~$F$.

\subsubsection*{Functions on the discrete cross-section}
Analogously to the continuum approach in 
Section~\ref{ssec: functions in strip and slit-strip},
we study s-holomorphic functions with Riemann boundary values
on the lattice strip~$\dstrip$
and lattice slit-strip~$\dslitstrip$ through their restrictions
to the horizontal cross-sections at height zero. 
These cross-sections consist of the~$\width$ horizontal
edges whose midpoints are $x' \in \crosssecdual$
as in~\eqref{eq: dual lattice interval}.

The space 
\begin{align}\label{eq: discrete function space}
\dfunctionsp := \; & \bC^{\crosssecdual} 
\end{align}
is thought of as the space of all complex valued functions
$f \colon \crosssecdual \to \bC$ on the 
discrete cross-section~$\crosssecdual$.
Because the main operations we consider are $\bR$-linear,
we interpret~$\dfunctionsp$ as a real vector space
of dimension
\begin{align*}
\dmn_\bR (\dfunctionsp) = 2 \width 
    .
\end{align*}
We equip it with the inner product defined
by the formula
\begin{align}\label{eq: inner product of discrete functions}
\innprod{f}{g} = \sum_{x' \in \crosssecdual}
    \Big( \re \big( f(x') \big) \, \re \big( g(x') \big) 
        + \im \big( f(x') \big) \, \im \big( g(x') \big) \Big) 
\end{align}
for~$f,g \in \dfunctionsp$.
This inner product induces the familiar norm
\begin{align*}
\| f \| = \Big(\sum_{x' \in \crosssecdual} |f(x')|^2 \Big)^{1/2} .
\end{align*}

\subsection{Vertical translation eigenfunctions in the lattice strip}
\label{sub: functions in lattice strip}

In Section~\ref{ssec: vertical translation eigenfunctions} we
saw that among the holomorphic functions with Riemann boundary
values in the strip, the vertical translation eigenfunctions
were exactly the extensions of quarter integer Fourier modes
on the cross-section. 
Here we address the analogous discrete question.

\subsubsection*{Discrete analytic continuation by one vertical step}
The operation of discrete analytic continuation by one vertical 
step was considered 
in~\cite{HKZ-discrete_holomorphicity_and_operator_formalism}.
We take from there the following result, whose proof is a
straightforward calculation from 
definitions~\eqref{eq: s-holomorphicity} 
and~\eqref{eq: dRBV in strip}.
\begin{prop}[{\cite[Lemmas~4 
 and~6]{HKZ-discrete_holomorphicity_and_operator_formalism}}]
\label{prop: discrete analytic continuation operator}
For any~$f \in \dfunctionsp$, 
there exists a unique 
function~$F \colon \dE(\dstrip) \to \bC$
which is s-holomorphic and has Riemann boundary values in the
lattice strip~$\dstrip$, and whose restriction to the
discrete cross-section 
coincides with~$f$:
\begin{align*}
F(x') = f(x') 
\qquad \text{for all } x' \in \crosssecdual \subset \dE(\dstrip) .
\end{align*}
If we define a new function $\propag f \colon \crosssecdual \to \bC$
in terms of this extension~$F$ by
\begin{align*}
\big(\propag f \big) (x') := F(x'+ \ii)
\qquad \text{for } x' \in \crosssecdual ,
\end{align*}
then $\propag f$ is explicitly given by
\begin{align*}
\big(\propag f \big) (x') = \; &
    2 \, f(x') 
  + \frac{\lambda^3}{\sqrt{2}} \, f(x'+1) 
  +\frac{\lambda^{-3}}{\sqrt{2}} \, f(x'-1) \\
& - \sqrt{2} \; \overline{f(x')} 
  + \frac{1}{\sqrt{2}} \,\overline{f(x'+1)}
  + \frac{1}{\sqrt{2}} \, \overline{f(x'-1)}
\qquad & \text{for $x' \neq \lft+\half , \rgt - \half$,} \\
\big(\propag f \big) (x'_\lftsym) = \; &
    \Big( 1+\frac{1}{\sqrt{2}} \Big) \,  f(x'_\lftsym) 
  + \frac{\lambda^3}{\sqrt{2}} \,  f(x'_\lftsym+1) \\
& + \Big( \lambda^3 
      + \frac{\lambda^{-3}}{\sqrt{2} } \Big) \, \overline{f(x'_\lftsym)}
  + \frac{1}{\sqrt{2}} \, \overline{f(x'_\lftsym+1)} 
\qquad & \text{for $x'_\lftsym = \lft+\half$,} \\
\big(\propag f \big) (x'_\rgtsym) = \; &
    \Big(1+\frac{1}{\sqrt{2}} \Big) \, f(x'_\rgtsym)
  + \frac{\lambda^{-3}}{\sqrt{2}} \, f(x'_\rgtsym-1) \\
& + \Big( \lambda^{-3} 
      + \frac{\lambda^{3}}{\sqrt{2}} \Big) \, \overline{f(x'_\rgtsym)}
  + \frac{1}{\sqrt{2}} \, \overline{f(x'_\rgtsym-1)} 
\qquad & \text{for $x'_\rgtsym = \rgt-\half$,}
\end{align*}
where~$\lambda = e^{\ii \pi/4}$.
The mapping $f \mapsto \propag f$ defines a linear operator
\begin{align*}
\propag \colon \dfunctionsp \to \dfunctionsp .
\end{align*}
\end{prop}

By the above it is clear that
$F \colon \dE(\dstrip) \to \bC$ is a vertical translation
eigenfunction if and only if its restriction $f = F \big|_{\crosssecdual}$
to the discrete cross-section is an eigenfunction of the 
operator~$\propag$. More precisely for
any~$\deigvalsym \neq 0$, the property
\begin{align*}
F(z + \ii h) = \deigvalsym^h \; F(z)
\qquad \text{ for all } z \in \dE(\dstrip) \text{ and } h \in \bZ
\end{align*}
is equivalent to
\begin{align*}
\propag f = \deigvalsym \, f .
\end{align*}

In~\cite{HKZ-discrete_holomorphicity_and_operator_formalism} 
many qualitative properties of the spectrum of~$\propag$
were proven directly: 
it is symmetric, invertible, conjugate to its own inverse,
all eigenvalues have multiplicity one, and $1$ is not an eigenvalue.
Moreover, the complexification of~$\propag$ was shown to be
conjugate to the induced rotation
of the Ising transfer matrix with locally monochromatic boundary
conditions, whose spectrum is 
well-known~\cite{AM-transfer_matrix_for_a_pure_phase, 
Palmer-planar_Ising_correlations}.
We will need the eigenvectors and eigenvalues explicitly, and
we thus rederive such properties via 
a direct calculation below.

The following qualitative property related to reflections
across the cross-section is, however, 
instructive and useful to note first.
\begin{rmk}\label{rmk: reflection and inversion of eigenvalues}
For any vertical translation eigenfunction
that is exponentially growing in the
upwards direction, 
there is a corresponding vertical
translation eigenfunction
that is exponentially growing in the
downwards direction. Namely if
$F \colon \dE(\dstrip) \to \bC$ is s-holomorphic and has
Riemann boundary values, then also
$\widetilde{F} \colon \dE(\dstrip) \to \bC$ defined by
\begin{align*}
\widetilde{F} (x + \ii y) = -\ii \, \overline{F(x - \ii y)}
\end{align*}
is s-holomorphic and has Riemann boundary values.
If~$F$ satisfies $F(z + \ii h) = \deigvalsym^h \, F(z)$
with $\deigvalsym \in \bR \setminus \set{0}$,
then~$\widetilde{F}$ satisfies 
$\widetilde{F}(z + \ii h) = (1/\deigvalsym)^{h} \, \widetilde{F}(z)$.

In the function space~$\dfunctionsp$, the corresponding
operation
\begin{align*}
f \mapsto \tilde{f} = \refl f, \qquad
(\refl f)(x') := - \ii \, \overline{f(x')}
\end{align*}
is a unitary involution~$\refl \colon \dfunctionsp \to \dfunctionsp$
by which~$\propag$ is conjugate to
its inverse~$\propag^{-1}$.
\end{rmk}

\subsubsection*{A dispersion relation}

The vertical translation eigenfunctions in the continuum strip are
essentially the quarter-integer Fourier modes.
The vertical translation eigenfunctions in the lattice strip
turn out to be mixtures of two discrete
Fourier modes with opposite frequencies.
Which frequencies can appear is ultimately
determined by the boundary conditions.
Before addressing that, let us observe that
a relation between the vertical translation eigenvalue 
and the frequency 
of the Fourier mode is obtained
from the first of the formulas of
Proposition~\ref{prop: discrete analytic continuation operator},
which governs the discrete analytic continuation away from boundaries.
\begin{lem}\label{lem: dispersion relation}
Let $\omega \in \bR \setminus 2 \pi \bZ$,
and let~$\deigvalsym$ be a solution to the
equation
\begin{align}\label{eq: dispersion relation}
\deigvalsym^2 + \big( 2 \cos(\omega) - 4 \big) \, \deigvalsym + 1 = 0.
\end{align}
Define
\begin{align*}
f(x') = C^+ \, e^{+\ii \omega x'} + C^- \, e^{-\ii \omega x'} ,
\qquad \text{for } x' \in \crosssecdual ,
\end{align*}
where the constants $C^+, C^- \in \bC$ are related by
\begin{align*}
C^- = \frac
        {2 + \sqrt{2} \, \cos(\frac{3 \pi}{4} + \omega) - \deigvalsym}
        {\sqrt{2} \big( 1 - \cos(\omega) \big)}
    \; \overline{C^+} .
\end{align*}
Then for 
any~$x' \in \crosssecdual \setminus \{ \lft+\half , \rgt-\half \}$,
we have
$\big( \propag f \big) (x') = \deigvalsym \,  f(x') $.

\end{lem}
\begin{proof}
Inserting the defining formula of~$f(x')$ into the 
the explicit expression for~$\big( \propag f \big) (x')$
from 
Proposition~\ref{prop: discrete analytic continuation operator},
a straightforward calculation yields
\begin{align*}
\big( \propag f \big) (x') - \deigvalsym \, f(x') 
= \; & \phantom{+} e^{+\ii \omega x'} \left( 
    C^+ \Big( 2 + \sqrt{2} \cos \big( \frac{3 \pi}{4} + \omega \big) 
            - \deigvalsym \Big) 
    + \overline{C^-} \Big( \sqrt{2} \cos(\omega) - \sqrt{2} \Big) \right) \\
& + e^{-\ii \omega x'} \left( 
    C^- \Big( 2 + \sqrt{2} \cos \big( \frac{3 \pi}{4} - \omega \big) 
            - \deigvalsym \Big) 
    + \overline{C^+} \Big( \sqrt{2} \cos(\omega) - \sqrt{2} \Big) \right) .
\end{align*}
Using the relationship between the constants~$C^+ , C^-$,
the coefficient of~$e^{+\ii \omega x'}$ above vanishes immediately.
It remains to check the vanishing of
the coefficient of~$e^{-\ii \omega x'}$.

For this purpose, observe that
using the trigonometric identities
$\cos \big( \frac{3 \pi}{4} + \omega \big) 
+ \cos \big( \frac{3 \pi}{4} - \omega \big) = -\sqrt{2} \, \cos(\omega)$
and
$\cos \big( \frac{3 \pi}{4} + \omega \big) 
\cos \big( \frac{3 \pi}{4} - \omega \big) = \half \cos(2 \omega)$
we can write
\begin{align*}
& \Big( 2 + \sqrt{2} \, \cos(\frac{3 \pi}{4} + \omega) - \deigvalsym \Big)
    \Big( 2 + \sqrt{2} \, \cos(\frac{3 \pi}{4} - \omega) - \deigvalsym \Big) \\
= \; & \deigvalsym^2 - \big( 4 - 2 \cos(\omega) \big) \deigvalsym
    + 4 + \cos(2\omega) -4 \cos(\omega) .
\end{align*}
When~$\deigvalsym$ is a solution to~\eqref{eq: dispersion relation},
this expression further simplifies to
\begin{align*}
3 + \cos(2\omega) -4 \cos(\omega) = 2 \big( 1 - \cos(\omega) \big)^2 .
\end{align*}
We thus see that an alternative equivalent form
of the relationship between 
the constants~$C^+ , C^-$ is
\begin{align*}
C^- = \frac{\sqrt{2} \big( 1 - \cos(\omega) \big)}
        {2 + \sqrt{2} \, \cos(\frac{3 \pi}{4} - \omega) - \deigvalsym}%
    \; \overline{C^+} .
\end{align*}
From this relationship we immediately see the
desired vanishing of the
coefficient of~$e^{-\ii \omega x'}$, so the proof is complete.
\end{proof}

\begin{rmk}
For a given~$\omega$, Equation~\eqref{eq: dispersion relation} 
has two roots, which are positive real numbers and 
inverses of each other. 
This reflects the observation from 
Remark~\ref{rmk: reflection and inversion of eigenvalues}
by which vertical translation eigenfunctions can be reflected
to produce eigenfunctions with the inverse eigenvalue.
Indeed
for functions of the form
$f(x') = C^+ \, e^{+\ii \omega x'} + C^- \, e^{-\ii \omega x'}$
as above, $(\refl f)(x') = -\ii \overline{f(x')}$
is also of the same form (with different coefficients),
and it has the inverse eigenvalue for~$\propag$.
\end{rmk}

\subsubsection*{Boundary conditions and equation on frequencies}

By the above calculation, the discrete Fourier modes of 
any frequency~$\omega \in \bR$ and its opposite can be combined
to satisfy the equation~$(\propag f)(x') = \deigvalsym \, f(x')$
for $x'$ not adjacent to the boundaries, provided that~$\omega$ 
and~$\deigvalsym$ are related by~\eqref{eq: dispersion relation}.
However, such functions can
satisfy the equation~$(\propag f)(x') = \deigvalsym \, f(x')$
near the boundaries only if the frequency is chosen judiciously.
It is sufficient for us to prove that
under a certain hypothesis on the frequency~$\omega$,
an eigenfunction exists; simple counting afterwards 
will show that all eigenfunctions are thus found.
\begin{lem}\label{lem: existence of eigenfunctions}
Suppose that $\omega \in \bR $ 
is a solution to
\begin{align}
\label{eq: allowed frequency equation} 
\frac{\cos\big( (\width+\half) \, \omega \big)}
{\cos\big( (\width-\half) \, \omega \big)}
= 3-2\sqrt{2} ,
\end{align}
and let $\deigvalsym$ be a solution to~\eqref{eq: dispersion relation}.
Then there exists non-zero $f \in \dfunctionsp$ such that
$\propag f = \deigvalsym \, f$.
\end{lem}
\begin{proof}
Let us denote by
\begin{align*}
R = R(\omega, \deigvalsym) = \frac
        {2 + \sqrt{2} \, \cos(\frac{3 \pi}{4} + \omega) - \deigvalsym}
        {\sqrt{2} \big( 1 - \cos(\omega) \big)} .
\end{align*}
the ratio in Lemma~\ref{lem: dispersion relation}
which relates the coefficient of one Fourier
mode to the complex conjugate of the coefficient of the opposite one.
We will show that~$f \in \dfunctionsp$ of the form
\[ f(x') = C \, e^{+\ii \omega x'} + R \overline{C} \, e^{-\ii \omega x'} \]
works, with suitably chosen $C \neq 0$.

Given the result of Lemma~\ref{lem: dispersion relation}, the only
remaining properties to verify are
$(\propag f)(x'_\lftsym) = \deigvalsym \, f(x'_\lftsym)$ and
$(\propag f)(x'_\rgtsym) = \deigvalsym \, f(x'_\rgtsym)$,
where $x'_\lftsym = \lft + \half$ and 
$x'_\rgtsym = \rgt - \half$. In view of 
Proposition~\ref{prop: discrete analytic continuation operator},
these amount to the equations
\begin{align*}
0 = \; & 
      C \, e^{\ii \omega x'_{\lftsym}} 
      \Big( A^+ + R B^+ \Big) 
    + \overline{C} \, e^{-\ii \omega x'_{\lftsym}} 
      \Big( B^- + R A^- \Big) \\
0 = \; & 
      C \, e^{\ii \omega x'_{\rgtsym}} 
      \Big( \overline{A^+} + R \overline{B^+} \Big) 
    + \overline{C} \, e^{-\ii \omega x'_{\rgtsym}} 
      \Big( \overline{B^-} + R \overline{A^-} \Big) ,
\end{align*}
where
\begin{align*}
A^\pm = A^\pm(\omega , \deigvalsym) 
:= \; & 1 + \frac{1}{\sqrt{2}} 
    + \frac{\lambda^3}{\sqrt{2}} e^{\pm \ii \omega} - \deigvalsym , \\
B^\pm = B^\pm(\omega) 
:= \; & \lambda^3 + \frac{\lambda^{-3}}{\sqrt{2}} 
    + \frac{1}{\sqrt{2}} e^{\pm \ii \omega} .
\end{align*}
Either one of the two equations fixes the argument of~$C$ modulo~$\pi$,
in that one can solve for~$C / \overline{C}$ from them. The former 
equation requires
\begin{align*}
C / \overline{C} \,
= \; & - e^{- \ii 2 \omega x'_{\lftsym}} \; \frac{B^- + R A^-}{A^+ + R B^+} ,
\end{align*}
and after first taking complex conjugates, the second requires
\begin{align*}
C / \overline{C} \,
= \; & - e^{- \ii 2 \omega x'_{\rgtsym}} \; \frac{A^+ + R B^+}{B^- + R A^-} .
\end{align*}
Evidently the modulus of these expressions are the inverses of each
other, so when we have the equality of the two,
the existence of nonzero~$C \in \bC$ satisfying the
eigenfunction requirements follows.
The equality of the two expressions simply reads
\begin{align*}
\frac{\big( A^+ + R B^+ \big)^2}{\big( B^- + R A^- \big)^2}
= e^{\ii 2 (\width - 1) \omega} .
\end{align*}
Our goal is therefore to show that this equality follows
from our 
assumption~\eqref{eq: allowed frequency equation}
(in fact the two are equivalent). 
The numerator and denominator on the left hand side are by construction
expressed as polynomials in~$\deigvalsym$ of degrees~$2$ and~$4$, but 
since by assumption $\deigvalsym$ satisfies the quadratic 
equation~\eqref{eq: dispersion relation}, 
we can reduce both
to first order polynomials in~$\deigvalsym$.
This is straightforward but slightly tedious.\footnote{
A symbolic computation 
in the quotient 
$\mathbb{F}[\deigvalsym] / 
\big\langle \deigvalsym^2 +(q+q^{-1}-4) \deigvalsym +1 \big\rangle$
of the polynomial ring~$\mathbb{F}[\deigvalsym]$
over the field~$\mathbb{F} = \bC(q)$ of rational functions 
of~$q=e^{\ii \omega}$ is a quick 
way to check the formulas.}
We first simplify the numerator 
and denominator before taking the squares
\begin{align*}
A^+(\omega,\deigvalsym) + R(\omega,\deigvalsym) \, B^+(\omega)
= \; & \frac{ (1-\sqrt{2}) \lambda \, \deigvalsym 
    + e^{-\ii \omega} \, \deigvalsym 
    +\lambda^3 \, e^{-\ii \omega} 
    + \ii (1-\sqrt{2}) }{2 (1-\cos(\omega))} , \\
B^-(\omega) + R(\omega,\deigvalsym) \,A^-(\omega,\deigvalsym) 
= \; & \frac{ (\sqrt{2} - 1) \, \deigvalsym - \lambda e^{\ii \omega} \, 
\deigvalsym 
   + \ii \, e^{\ii \omega} 
   + (\sqrt{2} - 1) \lambda^{3} }{2 (1-\cos(\omega))} .
\end{align*}
and then take the squares and simplify further to
\begin{align*}
\big(A^+ + R B^+ \big)^2 
= \; & \frac{(e^{\ii \omega}-1) \, e^{-\ii 3 \omega}}{4 (1-\cos(\omega))^2} 
    \, \Big( \big( 1 + \ii \, e^{\ii 2 \omega} 
        - 2\sqrt{2} \lambda e^{\ii \omega} \big) \deigvalsym 
    + \sqrt{2} \lambda \, e^{\ii \omega} \Big)
    \, \big( 1 - (3-2 \sqrt{2}) e^{\ii \omega} \big) \\
\big( B^- + R \, A^- \big)^2
= \; & \frac{(e^{\ii \omega}-1) \, e^{-\ii \omega}}{4 (1-\cos(\omega))^2} 
    \, \Big( \big( 1 + \ii \, e^{\ii 2 \omega} 
        - 2\sqrt{2} \lambda e^{\ii \omega} \big) \deigvalsym 
    + \sqrt{2} \lambda \, e^{\ii \omega} \Big)
    \, \big( (3-2 \sqrt{2}) - e^{\ii \omega} \big) .
\end{align*}
Cancellations in the ratio of these two yield
\begin{align*}
\frac{\big( A^+ + R B^+ \big)^2}{\big( B^- + R A^- \big)^2}
= \frac{\left(3-2 \sqrt{2} \right) e^{\ii \omega} - 1 }
   {e^{\ii \omega }-(3-2 \sqrt{2})} \; e^{-2 \ii \omega} .
\end{align*}
The desired equality of the expressions thus ultimately amounts to
\begin{align*}
\frac{\left(3-2 \sqrt{2} \right) e^{\ii \omega} - 1 }
   {e^{\ii \omega }-(3-2 \sqrt{2})}
= e^{\ii 2 \width \omega} ,
\end{align*}
which is easily seen to be equivalent 
to~\eqref{eq: allowed frequency equation}.
This finishes the proof.
\end{proof}

\subsubsection*{Allowed frequencies}

In the continuum, vertical translation eigenfunctions
were associated to all quarter-integer Fourier modes. By contrast, 
in the discrete setup there are only finitely many possible 
frequencies, and these are only approximately quarter 
integers (in the appropriate units).
We consider the strip width~$\width \in \bN$ in lattice units
fixed, and to display the parallel with the continuum,
we use positive half-integers~$k$ to index the 
allowed positive frequencies. 
Finite-dimensionality now restricts the index set to
\begin{align}
\label{eq: discrete positive half integers} 
\dposhalfint 
:= \big[ 0 , \width \big] \cap \Big( \bZ + \half \Big)
= \set{\frac{1}{2} , \frac{3}{2} , \ldots,
    \width - \frac{1}{2}} .
\end{align}
The following lemma describes the positive frequencies
which satisfy~\eqref{eq: allowed frequency equation}.
\begin{lem}
\label{lem: allowed frequencies}
For any~$k \in \dposhalfint$,
the equation~\eqref{eq: allowed frequency equation},
\begin{align*}
\frac{\cos\big( (\width+\half) \, \omega \big)}
{\cos\big( (\width-\half) \, \omega \big)}
= 3-2\sqrt{2} ,
\end{align*}
has a unique solution~$\omega = \omega_k^{(\width)}$ on the
interval~$\big( (k-\half)\pi/\width , \, k\pi/\width \big)$.
\end{lem}
\begin{proof}
Using the trigonometric formula for~$\cos(\alpha+\beta)$
in both the numerator and the denominator,
we can rewrite the left hand side 
of~\eqref{eq: allowed frequency equation} as
\begin{align*}
\frac{\cos\big( (\width+\half) \, \omega \big)}
{\cos\big( (\width-\half) \, \omega \big)}
= \; & \frac{\cos(\omega/2)\,\cos(\width\omega)
    - \sin(\omega/2)\,\sin(\width\omega)}{
    \cos(\omega/2)\,\cos(\width\omega)
    + \sin(\omega/2)\,\sin(\width\omega)} \\
= \; & \frac{1 - \tan(\omega/2)\,\tan(\width\omega)}{
    1 + \tan(\omega/2)\,\tan(\width\omega)} .
\end{align*}
On the 
interval~$\omega \in \big[ (k-\half)\pi/\width , \, k\pi/\width \big)$,
the expression
$\tan(\omega/2) \tan(\width\omega)$ increases from~$0$ to~$+\infty$,
so there is a
unique~$\omega \in \big( (k-\half)\pi/\width , \, k\pi/\width \big)$
such that
$\tan(\omega/2) \tan(\width\omega) = \frac{1}{\sqrt{2}}$.
This is the desired unique solution.
\end{proof}

\subsubsection*{Explicit eigenfunctions and eigenvalues}

We can now describe the eigenvalues and eigenfunctions
of vertical translations explicitly.
We use positive and negative indices~$k$
for eigenfunctions that are growing in the downwards 
and upwards directions.
\begin{prop}
\label{prop: discrete vertical translation eigenfunctions}
For $k \in \dposhalfint$,
denote by
\begin{align*}
\omega_k^{(\width)} 
    \in \Big( (k-\half)\pi/\width , \; k\pi/\width \Big) 
\end{align*}
the unique solution 
to~\eqref{eq: allowed frequency equation}
on this interval.
Denote by
\begin{align*}
\eigval{k} := 
    2 - \cos (\omega_k^{(\width)} ) 
    + \sqrt{ \big(3-\cos ( \omega_k^{(\width)} ) \big)
        \big(1-\cos (\omega_k^{(\width)}) \big)} 
\end{align*}
the corresponding solution to~\eqref{eq: dispersion relation}
with~$\eigval{k} > 1$, and by
$\eigval{-k} := 1 / \eigval{k} < 1$ the other solution.
Then there exists non-zero functions
\begin{align*}
\eigF{k} , \eigF{-k} \colon \dE(\dstrip) \to \bC
\end{align*}
which are s-holomorphic and have Riemann boundary values
and satisfy
\begin{align*}
\eigF{\pm k} (z + \ii h) = (\eigval{\pm k})^h \; \eigF{\pm k}(z)
\qquad \text{ for all } z \in \dE(\dstrip) \text{ and } h \in \bZ ,
\end{align*}
and these are uniquely determined by the normalization conditions
that the argument on the left boundary is
$\eigF{\pm k}(\lft + \ii y') \in e^{-\ii \pi / 4} \, \bR_+$,
for $y' \in \bZ + \half$, and that their
restrictions
\begin{align*}
\eigf{\pm k} = \eigF{\pm k} \big|_{\crosssecdual} \; \in \dfunctionsp
\end{align*}
to the cross-section have unit norm~$\| \eigf{\pm k} \| = 1$. 

The following relations
hold for the normalized eigenfunctions with opposite indices:
\begin{align*}
\eigf{-k} (x') 
= \; & -\ii \; \overline{\eigf{k} (x')} , &
\eigF{-k} (x + \ii y) 
= \; & -\ii \; \overline{\eigF{k} (x - \ii y)} .
\end{align*}

The functions~$\eigf{k}$, $k \in \pm \dposhalfint$,
form an orthonormal basis of~$\dfunctionsp$.
\end{prop}
\begin{proof}
Lemma~\ref{lem: existence of eigenfunctions}
gives the existence of non-zero eigenfunctions~$\eigf{\pm k}$
of~$\propag$ with the desired eigenvalues~$\eigval{\pm k}$,
and it is clear that unit norm~$\|\eigf{\pm k}\| = 1$ 
fixes these up to a sign in the real vector space~$\dfunctionsp$,
and the argument on the left boundary fixes the remaining sign.

The relation between $\eigf{k}$ and $\eigf{-k}$
as well as between $\eigF{k}$ and $\eigF{-k}$ are
straightforward from 
Remark~\ref{rmk: reflection and inversion of eigenvalues},
since the reflected function
$\widetilde{F}(x+\ii y) = - \ii \, F(x-\ii y)$
has the same argument as~$F$ on the left boundary
(which we used for normalization purposes).

Among $\eigf{k}$, $k \in \pm \dposhalfint$, we have
$2 \, \# \dposhalfint$ normalized eigenfunctions 
of the symmetric operator~$\propag$ with distinct eigenvalues.
In view of $2 \, \# \dposhalfint 
= 2 \width = \dmn \, \dfunctionsp$, these form an orthonormal basis.
\end{proof}

In particular, all the earlier qualitative statements about 
the spectrum of~$\propag$ can of course be verified
from the above explicit diagonalization of it.

\subsubsection*{Decomposition of the function space}
Analogously to the continuous case, we 
split~$\dfunctionsp$ into orthogonally complementary subspaces
\begin{align*}
\dfunctionsp
= \; & \dfspTpole \oplus \dfspTzero ,
\end{align*}
where
\begin{align}\label{eq: discrete top poles and zeros def}
\dfspTpole 
	:= \; & {\spn_\bR \set{ \eigf{k} \; \big| \; k \in \dposhalfint}} &
\dfspTzero
	:= \; & {\spn_\bR \set{ \eigf{-k} \; \big| \; k \in \dposhalfint}} ,
\end{align}
with associated orthogonal projection 
operators
\begin{align*}
\dprTpole \colon \; & \dfunctionsp \to \dfspTpole , &
\dprTzero \colon \; & \dfunctionsp \to \dfspTzero .
\end{align*}
The subspace~$\dfspTpole$ consists of
functions whose s-holomorphic extensions with Riemann boundary values
in the lattice strip
grow exponentially fast in the upwards 
direction, and~$\dfspTzero$ of
functions whose extensions grows exponentially fast in 
the downwards direction.

\subsection{Functions in the lattice slit-strip}
\label{sub: functions in lattice strip}
\label{sub: distinguished s-holomorphic functions in the slit-strip}

We now consider functions in the lattice slit-strip~$\dslitstrip$
of Figure~\ref{sec: discrete complex analysis}.
\ref{sfig: lattice slit-strip}.
We use three subgraphs
\[ \dslitstripT , \; \dslitstripL , \; \dslitstripR
\; \subset \; \dslitstrip \]
of the lattice slit-strip.
The top part~$\dslitstripT$ is taken to consist of all vertices
and edges of~$\dslitstrip$ with non-negative imaginary part.
The left  
leg part~$\dslitstripL$ 
is taken to consist of vertices and edges with non-positive
imaginary part and non-positive real part,
except for those of the doubled edges along the slit
which are considered to form the left boundary of the right substrip.
The right
leg part~$\dslitstripR$ 
is defined similarly.
Note that these three subgraphs of~$\dslitstrip$
have otherwise disjoint edge sets except that each horizontal
edge in the cross-section~$\crosssecdual$ belongs to both
the top part and either the left or the right leg.
We correspondingly partition the 
cross-section~$\crosssecRdual = \dintervaldual{\lft}{\rgt}$
into the left and right halves,
$\crosssecLdual = \dintervaldual{\lft}{\mdpt}$ and
$\crosssecRdual = \dintervaldual{\mdpt}{\rgt}$,
and decompose the discrete function space
to functions with support on the left and right halves,
\begin{align*}
\dfunctionsp = \dfunctionspL \oplus \dfunctionspR ,
\end{align*}
where we define
$\dfunctionspL = \bC^{\crosssecLdual}$ and
$\dfunctionspR = \bC^{\crosssecRdual}$, and interpret both
as subspaces in~$\bC^{\crosssecdual} = \dfunctionsp$.

The strip~$\dstrip$ and the slit-strip~$\dslitstrip$ 
graphs coincide exactly 
in the top part~$\dslitstripT$, and in particular 
s-holomorphic functions
$F \colon \dE(\dslitstripT) \to \bC$
with Riemann boundary values 
in the top part are as in the strip:
the discrete analytic continuation 
upwards from the cross-section~$\crosssecdual$
is achieved by the same
operator~$\propag \colon \dfunctionsp \to \dfunctionsp$.

Downwards from the cross-section, on the other hand, the 
lattice slit-strip~$\dslitstrip$ has separate 
left and right halves~$\dslitstripL$ and $\dslitstripR$,
which coincide with lower halves of lattice strips
of smaller widths~$\widthL=-\lft$ and $\widthR = \rgt$.
Note that due to the double edges on the slit, the left and right
halves have their own sets of edges on which functions
are defined, and the Riemann boundary 
values~\eqref{eq: dRBV on the slit part} are exactly 
what one would require in the smaller width substrips.
Therefore the discrete analytic continuation downwards
from the cross-section~$\crosssecdual$ in the lattice
slit-strip is then simply the direct
sum~$(\propag^{(\widthL)})^{-1} \oplus (\propag^{(\widthR)})^{-1}$
of inverses of operators defined as 
in Section~\ref{sub: functions in lattice strip}
but in substrips of widths~$\widthL , \widthR$.

\subsubsection*{Decompositions of the function space}

The decomposition $\dfunctionsp = \dfunctionspL \oplus \dfunctionspR$
is clearly an orthogonal direct sum, and in each summand
we get an orthonormal basis in the same way as for
the lattice strip.
Instead of~\eqref{eq: discrete positive half integers},
the indexing sets for the (positive) modes are now
\begin{align*}
\dposhalfintL
:= & 
\set{\frac{1}{2} , \frac{3}{2} , \ldots,
    \widthL - \frac{1}{2}} , 
&
\dposhalfintR
:= & 
\set{\frac{1}{2} , \frac{3}{2} , \ldots,
    \widthR - \frac{1}{2}} .
\end{align*}
In the same way as in
Proposition~\ref{prop: discrete vertical translation eigenfunctions},
for each $k \in \pm \dposhalfintL$
we define the normalized eigenvector
$\eigfL{k} \in \dfunctionspL$ of
$\propag^{(\widthL)}$ with eigenvalue~$\eigvalW{\widthL}{k}$
and the extension
\begin{align*}
\eigFL{k} \colon \dE(\dslitstripL) \to \bC, 
\end{align*}
and for each $k \in \pm \dposhalfintR$
the normalized eigenvector
$\eigfR{k} \in \dfunctionspR$ of
$\propag^{(\widthR)}$ with eigenvalue~$\eigvalW{\widthR}{k}$
and the extension
\begin{align*}
\eigFR{k} \colon \dE(\dslitstripR) \to \bC , 
\end{align*}
Together,
$(\eigfL{k})_{k \in \pm \dposhalfintL}$
and~$(\eigfR{k})_{k \in \pm \dposhalfintR}$ form
an othonormal basis of~$\dfunctionsp$.

Given these bases, we may decompose
\begin{align*}
\dfunctionspL
= \; & \dfspLpole \oplus \dfspLzero , 
&
\dfunctionspR
= \; & \dfspRpole \oplus \dfspRzero , 
\end{align*}
where
\begin{align}
\dfspLpole 
	:= \; & {\spn_\bR \set{ \eigfL{k} \; \big| \; k < 0}} &
\dfspRpole 
	:= \; & {\spn_\bR \set{\eigfR{k} \; \big| \; k < 0}} 	
\\
\nonumber
\dfspLzero
	:= \; &{\spn_\bR \set{ \eigfL{k} \; \big| \; k > 0}}
&
\dfspRzero
	:= \; & {\spn_\bR \set{ \eigfR{k} \; \big| \; k > 0}},
\end{align}
with respective orthogonal projection 
operators~$\dprLpole, \dprRpole, \dprLzero, \dprRzero$.

We have thus introduced the decompositions of~$\dfunctionsp$
\begin{align*}
\dfunctionsp
= \; & \dfspTpole \oplus \dfspTzero ,
\end{align*}
and
\begin{align*}
\dfunctionsp
= \; & \dfspLpole \oplus \dfspLzero \oplus \dfspRpole \oplus \dfspRzero .
\end{align*}

\subsubsection*{Singular parts}
As in the continuous case,
for a function~$f \in \dfunctionsp$, we call
\begin{align}
\nonumber
\dprTpole(f) \in \; & \dfspTpole 
& \text{ its \term{singular part at the top},} & \\
\label{eq: definition of singular parts}
\dprLpole(f) \in \; & \cfspLpole 
& \text{ its \term{singular part in the left leg},} & \\
\nonumber
\dprRpole(f) \in \; & \dfspRpole 
& \text{ its \term{singular part in the right leg}.} &
\end{align}
If $\dprTpole(f) = 0$ (resp. $\dprLpole(f)=0$ or
$\dprRpole(f)=0$), we say that the function~$f$
admits a \term{regular extension} to the top 
(resp. regular extension to the left leg or regular
extension to the right leg).

The following result shows that a function is uniquely
characterized by its singular parts.
\begin{lem}\label{lem:discrete-uniqueness}
If a function~$f \in \dfunctionsp$ admits regular extensions
to the top, to the left leg, and to the right leg, then~$f \equiv 0$.
\end{lem}
We postpone the proof of this lemma to
Section~\ref{ssec: the iios trick}, where we have at our disposal
the necessary discrete complex analysis tools needed to carry out
the proof analogous to the continuum.

\subsubsection*{Functions with prescribed singular parts}
In the discrete setting,
the construction of the functions with prescribed singular parts
can now be achieved simply by finite-dimensional linear algebra.
\begin{lem}\label{lem:discrete-existence}
For any $g_\topsym \in \dfspTpole$,
$g_\lftsym \in \dfspLpole$, $g_\rgtsym \in \dfspRpole$,
there exists a unique function $f \in \dfunctionsp$ such that
\begin{align*}
\dprTpole(f) = \; & g_\topsym , &
\dprLpole(f) = \; & g_\lftsym , &
\dprRpole(f) = \; & g_\rgtsym .
\end{align*}
\end{lem}
\begin{proof}
Consider the linear map
\begin{align*}
f \mapsto 
    \big(\dprTpole(f) , \, \dprLpole(f) , \, \dprRpole(f) \big)
\end{align*}
on the function space~$\dfunctionsp$. It maps the
space $\dfunctionsp$ of dimension~$\dmn_\bR(\dfunctionsp) = 2 \width$
to the external direct sum
$\dfspTpole \oplus \dfspLpole \oplus \dfspRpole$, which is a
space of dimension
\begin{align*}
\dmn_\bR(\dfspTpole) + \dmn_\bR(\dfspTpole) + \dmn_\bR(\dfspTpole) 
= \; & \width + \widthL + \widthR = 2 \width .
\end{align*}
Its injectivity follows from Lemma~\ref{lem:discrete-uniqueness},
so bijectivity follows by the equality of the dimensions.
\end{proof}

By the above, in analogy with~\eqref{eq:cpole-asymptotic}
we define
\begin{align*}
\poleT{k} \in \; & \dfunctionsp \; \text{ for } k \in \dposhalfint, &
\poleL{k} \in \; & \dfunctionsp \; \text{ for } k \in \dposhalfintL, &
\poleR{k} \in \; & \dfunctionsp \; \text{ for } k \in \dposhalfintR, 
\end{align*}
as the functions whose singular parts
are
\begin{align}
\nonumber
\dprTpole(\poleT{k}) = \; & \eigf{+k} , &
\dprLpole(\poleT{k}) = \; & 0 , &
\dprRpole(\poleT{k}) = \; & 0 , \\ 
\label{eq:dpole-asymptotic}
\dprTpole(\poleL{k}) = \; & 0 , &
\dprLpole(\poleL{k}) = \; & \eigfL{-k} , &
\dprRpole(\poleL{k}) = \; & 0 , \\ 
\nonumber
\dprTpole(\poleR{k}) = \; & 0 , &
\dprLpole(\poleR{k}) = \; & 0 , &
\dprRpole(\poleR{k}) = \; & \eigfR{-k} .
\end{align}

These are functions which are singular under s-holomorphic 
propagation in one direction, while admitting regular 
extensions in the remaining two directions. Denote the
corresponding s-holomorphic functions with Riemann boundary
values in the lattice slit-strip by
\begin{align*}
\PoleT{k} \colon \dE(\dslitstrip) \to \bC , \qquad
\PoleL{k} \colon \dE(\dslitstrip) \to \bC , \qquad
\PoleR{k} \colon \dE(\dslitstrip) \to \bC 
\end{align*}
We call these the \term{discrete pole functions}.
Note that these are defined globally in the lattice slit-strip,
unlike for example~$\eigFL{k}$, $\eigFR{k}$, and $\eigF{k}$
(each of these is globally defined in a suitable lattice strip
which only coincides with the lattice slit-strip in
one of the three subgraphs).

These functions have asympotics analogous 
to~\eqref{eq:cpole-asymptotic}:
\begin{align}\nonumber
\PoleT{k}(x+\ii y) - \eigF{k}(x+\ii y) 
= \; & o(1) \qquad \text{ as $y \to +\infty$ and $x+\ii y \in \dslitstripT$} , 
\\ \label{eq:dpole-asymptotic}
\PoleL{k}(x+\ii y) - \eigFL{-k}(x+\ii y)
= \; & o(1) \qquad \text{ as $y \to -\infty$ and $x+\ii y \in \dslitstripL$} , 
    \\ \nonumber
\PoleR{k}(x+\ii y) - \eigFR{-k}(x+\ii y)
= \; & o(1) \qquad \text{ as $y \to -\infty$ and $x+\ii y \in \dslitstripR$} .
\end{align}
Together with the regular extension to the other two extremities in 
each case, the asymptotics~\eqref{eq:dpole-asymptotic} characterize
the discrete pole functions.

\section{Discrete complex analysis and scaling limit results}%
\label{sec: convergence of functions}

In Sections~\ref{sec: function spaces} 
and~\ref{sec: discrete complex analysis} we introduced
spaces of functions in continuum and discrete settings, respectively,
and distinguished functions adapted to the strip and the slit-strip
geometries in each case. In this section, we prove 
convergence of the discrete functions to the continuum ones,
as the lattice width increasese, $\width \to \infty$.
We must
require~${\lft/\width \to \mhalf}$
and~${\rgt/\width \to \phalf}$ as $\width \to \infty$,
and in order for the
functions~$\eigf{k}$ defined on the discrete
cross-section~$\crosssecdual = \dintervaldual{\lft}{\rgt}$
to approximate the functions~$\ccffun_k$ 
defined on $\ccrosssec = [\mhalf, \phalf]$, their arguments
must be rescaled by a factor~$\width^{-1}$. Because of the norms
induced by \eqref{eq: inner product of discrete functions}
and~\eqref{eq: inner product on L2} for discrete and continuous
functions, also the values of the discrete functions must be rescaled
by~$\width^{1/2}$
(the norm-squared of the constant function~$1$ in the discrete is~$\width$).
Similarly for functions on the discrete 
strip 
and slit-strip (both with vertex sets $\dinterval{\lft}{\rgt} \times \bZ$),
we rescale arguments by~$\width^{-1}$ and values by~$\width^{1/2}$.
In order to discuss convergence (typically uniformly over
compact subsets), we will interpret the
discrete functions being interpolated to the continuum
in any reasonable manner
\footnote{The details of the interpolation
are irrelevant except for the fact that the equicontinuity established
for the values on the lattice functions has to be inherited
by their interpolations to the continuum.
One possibility is to extend by local 
averages to a triangulation that refines the lattice on which the 
values are defined, and then to linearly interpolate on the 
triangles. Another possibility is to linearly interpolate along 
line segments on which adjacent values are defined, and then
to harmonically interpolate to the areas surrounded by the line 
segments.}
without explicit mention.

In Section~\ref{ssec: convergence of strip functions}
we first prove the convergence
in the scaling limit of the discrete vertical translation
eigenfunctions in the strip. The formulas we have in this
case are sufficiently explicit for the proof to be done
without analytical tools.
In Sections~\ref{ssec: the iios trick}~--
~\ref{ssec: precompactness and convergence} we introduce the
regularity theory for s-holomorphic functions as it is needed
for the remaining main results.
The key tool is the 
``imaginary part of the integral of the square'' of an
s-holomorphic function introduced by
Smirnov~\cite{Smirnov-towards_conformal_invariance}:
a function defined on both vertices and faces which
behaves almost like a harmonic function 
and has constant boundary values on any part of the boundary
on which the s-holomorphic function had Riemann boundary 
values. This will be introduced in Section~\ref{ssec: the iios trick}.
Notably, the almost harmonicity implies suitable versions of
maximum principles, Beurling-type estimates, and
equicontinuity results.
In Section~\ref{ssec: maximum principle and uniqueness},
the maximum principle will be used to prove that an s-holomorphic
function on the discrete slit-strip 
admitting regular extensions to all three directions is zero, 
and therefore any s-holomorphic function is uniquely characterized
by its singular parts.
In Section~\ref{ssec: precompactness and convergence},
the Beurling-type estimates and equicontinuity results will be
used to prove the convergence of the discrete pole functions
in the slit-strip to the continuum ones.

\subsection{Convergence of vertical translation eigenfunctions}
\label{ssec: convergence of strip functions}

We start from the distinguished functions in the strip geometry,
i.e., the vertical translation eigenfunctions of 
Sections~\ref{ssec: vertical translation eigenfunctions}
and~\ref{sub: functions in lattice strip}.
The convergence of these can be
proven directly from the explicit formulas we have obtained.

\subsubsection*{Auxiliary asymptotics}
Let us record auxiliary observations about the explicit 
formulas for the functions~$\eigf{\pm k} \in \dfunctionsp$
and the involved frequencies~$\omega^{(\width)}_k$
and eigenvalues~$\eigval{\pm k}$.
In the scaling limit setup, we consider the 
index~$k \in \poshalfint$ fixed, and consider the
limit~${\width \to \infty}$ of infinite width (in lattice units).

So let~$k \in \poshalfint$ be fixed.
For $\width \in \bN$, $\width > k$, let 
$\omega^{(\width)}_k \in \big( (k-\half)\pi/\width , \; k\pi/\width \big)$ 
be the unique
solution to~\eqref{eq: allowed frequency equation} as in 
Lemma~\ref{lem: allowed frequencies}, and let
$\eigval{k} := 
    2 - \cos (\omega_k^{(\width)} ) 
    + \sqrt{ \big(3-\cos ( \omega_k^{(\width)} ) \big)
        \big(1-\cos (\omega_k^{(\width)}) \big)} $
be the corresponding solution to~\eqref{eq: dispersion relation}
with~$\eigval{k} > 1$.
\begin{lem}
As $\width \to \infty$, we have
\begin{align}
\label{eq: frequency asymptotics} 
\omega^{(\width)}_k 
= \; & \frac{\pi}{\width} k + \OO(\width^{-2}) , \\
\label{eq: eigenvalue asymptotics} 
\eigval{k} = \; & 1 + \frac{\pi}{\width} k + \OO(\width^{-2}) .
\end{align}
\end{lem}
\begin{proof}
For the first formula, it is simple to use 
the method of proof of
Lemma~\ref{lem: allowed frequencies}.
Since $0 < \omega^{(\width)}_k < k \pi \width^{-1}$,
we have $0 < \tan \omega^{(\width)}_k / 2 < c \width^{-1}$
for some~$c>0$. Therefore the equation
$\tan \big( \omega^{(\width)}_k/2 \big) \,
\tan \big( \width \omega^{(\width)}_k \big) = \frac{1}{\sqrt{2}}$
implies $\tan \big( \width \omega^{(\width)}_k \big) > c' \, \width$
for some $c'>0$, and the first order pole of~$\tan$ at~$k\pi$ 
then requires
$k \pi \width^{-1} - c'' \, \width^{-2} 
    < \omega^{(\width)}_k < k \pi \width^{-1}$
for some~$c''>0$,
which gives~\eqref{eq: frequency asymptotics}.

For the second formula, let $\deigvalsym(\omega)
= 2 - \cos (\omega) 
    + \sqrt{ \big(3-\cos ( \omega ) \big)
        \big(1-\cos (\omega) \big)}$
for~$\omega \geq 0$. This has a 
power series representation on~$\omega \in (0,\pi)$ with 
initial terms~$\deigvalsym(\omega) = 1 + \omega + 
\OO(\omega^2)$. The second 
formula~\eqref{eq: eigenvalue asymptotics} thus follows from 
the first~\eqref{eq: frequency asymptotics}
in view of ${\eigval{k} = \deigvalsym(\omega^{(\width)}_k)}$.
\end{proof}

\begin{lem}
Let~$k \in \poshalfint$, and
let~$C_{\pm k}^{(\width);+},C_{\pm k}^{(\width);-} \in \bC$
denote the coefficients in
\begin{align*}
\eigf{\pm k} (x') 
= \; & C_{\pm k}^{(\width);+} \, \exp \big( +\ii \omega^{(\width)}_k x' \big)
    + C_{\pm k}^{(\width);-} \, \exp \big( -\ii \omega^{(\width)}_k x' \big) .
\end{align*}
Then as $\width \to \infty$, we have
\begin{align*}
\big| C_{+ k}^{(\width);+} \big| = \; & \OO(\width^{-3/2}) , & 
\big| C_{+ k}^{(\width);-} \big|
    = \; & \width^{-1/2} + \OO(\width^{-3/2}) , \\
\big| C_{- k}^{(\width);+} \big| 
    = \; & \width^{-1/2} + \OO(\width^{-3/2}) , & 
\big| C_{- k}^{(\width);-} \big| = \; & \OO(\width^{-3/2}) .
\end{align*}
\end{lem}
\begin{proof}
Consider the case of positive index~$k \in \poshalfint$.
Recall from Lemma~\ref{lem: dispersion relation} that we have
$C_{+k}^{(\width);-} 
= R(\omega^{(\width)}_k) \, \overline{C_{+k}^{(\width);+}}$,
where $R(\omega) = \frac
        {2 + \sqrt{2} \, \cos(\frac{3 \pi}{4} + \omega) - \deigvalsym(\omega)}
        {\sqrt{2} \big( 1 - \cos(\omega) \big)}$.
A calculation shows $R(\omega) = - \frac{2\sqrt{2}}{\omega} + \OO(\omega)$,
and since $\omega^{(\width)}_k = \OO(\width^{-1})$,
we see that $|R(\omega^{(\width)}_k)| > c \, \width$ for some~$c>0$,
i.e.,
\begin{align*}
\big| C_{+k}^{(\width);+} \big|
\leq \frac{1}{c \, \width} \, \big| C_{+k}^{(\width);-} \big| .
\end{align*}
Therefore for the values of the eigenfunction~$\eigf{k}$, we have
\begin{align*}
\eigf{k}(x') 
= \; & C_{+k}^{(\width);-} \Big( \exp \big( -\ii \omega^{(\width)}_k x' \big)
    + \OO(\width^{-1}) \Big) .
\end{align*}
The unit norm normalization condition~$\|\eigf{k}\| = 1$
gives
\begin{align*}
1 = \|\eigf{k}\|^2 = \sum_{x' \in \crosssecdual} |\eigf{k}(x')|^2
= \width \, \big| C_{+ k}^{(\width);-} \big|^2 \, 
    \big( 1 + \OO(\width^{-1}) \big) .
\end{align*}
We conclude that 
$\big| C_{+ k}^{(\width);-} \big| = \frac{1}{\sqrt{\width}} + 
\OO(\width^{-3/2})$
and 
$\big| C_{+ k}^{(\width);+} \big| = \OO(\width^{-3/2})$.
The case of negative indices can be done similarly, but it
also follows from the above using
Remark~\ref{rmk: reflection and inversion of eigenvalues}.
\end{proof}

\subsubsection*{Limit result for the strip functions}
We can now state and straightforwardly verify the scaling limit
result for vertical translation eigenfunctions.

\begin{thm}\label{thm: convergence of strip functions}
Choose sequences $(\lft_{n})_{n \in \bN}$,
$(\rgt_{n})_{n \in \bN}$ of integers
$\lft_{n}, \rgt_{n} \in \bZ$ such that
\begin{itemize}
\item $\lft_{n} < 0 < \rgt_{n}$ for all $n$;
\item $\width_{n} := \rgt_{n} - \lft_{n} \to +\infty$ as $n \to \infty$;
\item $\lft_{n} / \width_{n} \to -\half$
and $\rgt_{n} / \width_{n} \to +\half$ as $n \to \infty$.
\end{itemize}
Let~$\eigf{k}^{(\width_{n})}$ and~$\eigF{k}^{(\width_{n})}$ denote
the functions of 
Proposition~\ref{prop: discrete vertical translation eigenfunctions}
in the lattice strips with~$\lft = \lft_n$ and $\rgt = \rgt_n$.
Then for any~$k \in \pm \poshalfint$, as~$n \to \infty$ we have 
\begin{align*}
\sqrt{\width_n} \;  \eigf{k}^{(\width_{n})} \big( x \width_n \big) 
\to \; & \ccffun_{k}(x) & &
    \text{ uniformly on~$\ccrosssec \ni x$,} \\
\sqrt{\width_n} \; \eigF{k}^{(\width_{n})} \big( z \width_n \big) 
\to \; & \ccfFun_{k}(z) & &
    \text{ uniformly on compact subsets of~$\cstrip \ni z$}.
\end{align*}
\end{thm}
\begin{proof}
Consider $k \in \poshalfint$.
We will use the normalization constant~$C_k$
of the quarter-integer Fourier mode~$\ccffun_k$ given 
by~\eqref{eq: normalization of half integer Fourier modes},
and the normalization constants~$C_{+ k}^{(\width_n);\pm}$
as in the previous lemma but in lattice strip with $\lft = \lft_n$
and $\rgt = \rgt_n$.
Let us denote 
\[\xi_n := 
\frac{ C_{+ k}^{(\width_n);-} }{\big| C_{+ k}^{(\width_n);-} \big| \, C_k} . \]
Then $|\xi_n| = 1$, so $\xi_n$ is a phase factor, and we will first 
factor it out. In view of 
${\width_n \, \omega_k^{(\width_n)} \to k \pi}$
and~${\sqrt{\width_n} \, \big| C_{+ k}^{(\width_n);-} \big| \to 1}$
(two previous lemmas)
and the asymptotics for~$\eigf{k}$ from the previous lemma, we get
\begin{align*}
\frac{\sqrt{\width_n}}{\xi_n} \;
    \eigf{k}^{(\width_n)} \big( x \width_n \big) 
= \; & \frac{\sqrt{\width_n} \, C_{+k}^{(\width_n);-}}{\xi_n} \;
    \Big( \exp \big( -\ii \omega^{(\width_n)}_k x \width_n \big)
    + \OO(\width_n^{-1}) \Big) \\
\longrightarrow \; & C_k \;
    \exp \big( -\ii k \pi x \big) 
\; = \; \ccffun_k(x) 
\end{align*}
uniformly over~$x \in \ccrosssec$.
Since
\[ \big( \eigvalW{\width_n}{k} \big)^{y \, \width_n}
= \big( 1 + \pi k / \width_n + \OO(\width_n^{-2}) \big)^{y \, \width_n}
\; \longrightarrow \; e^{\pi k y} ,  \]
and~$\eigF{k}^{(\width_n)}(x + \ii y) 
= \big(\eigvalW{\width_n}{k}\big)^y \, \eigf{k} (x)$,
we also have
\begin{align*}
\frac{\sqrt{\width_n}}{\xi_n} \;
    \eigF{k}^{(\width_n)} \big( (x + \ii y) \width_n \big) 
\longrightarrow \; & e^{\pi k y} \, \ccffun_{k} (x)
\; = \; \ccfFun_k(x) 
\end{align*}
uniformly on compact subsets of~$\overline{\cstrip} = \ccrosssec \times \bR$.
To finish the proof of the convergence assertions, it only remains to 
show that the phase factor is asymptotically correct, $\xi_n \to 1$.
This is indeed a consequence of the chosen 
normalizations. We have defined~$C_{+ k}^{(\width_n);-}$ and $C_k$
so that ${\arg \big( \eigF{k} \big) = - \pi / 4}$
and ${\arg \big( \ccfFun_{k} \big) = - \pi / 4}$ on the left 
boundaries of the lattice strip and continuum strip, respectively,
so the uniform convergence on compacts that we established above
is only possible if also~$\xi_n \to 1$.

The case of negative indices can be done similarly, but it
also follows from the above using
Remark~\ref{rmk: reflection and inversion of eigenvalues}.
\end{proof}

\subsection{The imaginary part of the integral of the square}
\label{ssec: the iios trick}

In the remaining part,
Sections~\ref{ssec: the iios trick}~
--~\ref{ssec: precompactness and convergence}, 
we recall the regularity theory for s-holomorphic functions, 
and apply it to prove the main results.
Analogous to the continuous case, the lattice discretization of 
Cauchy-Riemann equations are equivalent to the existence of the
line integral of~$F$, i.e., the
closedness of the (discretization of the)
$1$-form~$F(z)\,\ud z$. S-holomorphicity is a 
strictly stronger notion which also implies closedness of 
(the discretization of)
another form,~$\im \left[ F(z)^2 \, \ud z\right]$
\cite{CS-universality_in_Ising}, so that
the ``imaginary part of the integral of the square''
becomes well-defined.
We remark that the literature contains a few different
conventions about s-holomorphicity.
\footnote{With some alternative
conventions the additional closed form
is~$\re \left[ F(z)^2 \, \ud z\right]$ instead.}
Our conventions coincide with those
of~\cite{CS-discrete_complex_analysis_on_isoradial_graphs, 
CS-universality_in_Ising} presented in the general context
of isoradial graphs, but they differ by a multiplicative factor 
and the orientation of the square grid from most of the 
literature specific to the square lattice such 
as~\cite{CHI-conformal_invariance_of_spin_correlations}.

\begin{figure}[tb]
    \includegraphics[width=.35\textwidth]{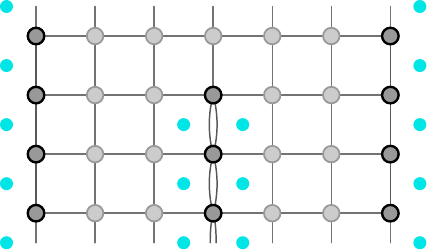}
\caption{The boundary faces with reduced weight are marked. The marked faces 
next to the slit are virtual faces accessed \emph{across} the slit from an 
interior face.}
\label{fig: corner values and boundary}
\end{figure}

\subsubsection*{Refinements to the lattice domains}
Since we use discrete complex analysis only in the 
lattice strip~$\dstrip$ and lattice slit-strip~$\dslitstrip$,
we will present the tools in the simplest form that applies to 
these cases.
Moreover, since our main objective is to show convergence 
results in the scaling limit framework, where the lattice variables 
are rescaled by a factor~$\width^{-1}$, we 
will present the key notions in the context of the
rescaled square lattice~$\delta\bZ^2$ for~$\delta>0$:
the choice~$\delta=1$ corresponds to the original lattice setup, 
and the choice~$\delta = \width^{-1}$ will be used for 
scaling limit results.

The following graph notations will be needed.
We denote the 
set of vertices 
by~$\dV = \delta \, \dstrip = \delta \, \dslitstrip$,
and the set of edges by~$\dE$; 
so $\dE = \delta \dE(\dstrip)$ or $\dE = \delta \dE(\dslitstrip)$.
We moreover use the 
notation~$\dV^*$ for the set of faces.
The ``imaginary part of the integral of the square''
will be defined on~$\dV \cup \dV^*$.
For the treatment of boundary values, we also introduce
\term{boundary faces}, which are imagined faces across 
the boundary edges as in
Figure~\ref{fig: corner values and boundary}.
The set of boundary faces is denoted by~$\bdry \dV^*$, and it is
by definition in bijective correspondence with the set 
of boundary edges.

Let us also define a \term{corner} of our graph to be
a pair~$c=(v,p)$ consisting of a vertex~$v \in \dV$ and 
a face~$p \in \dV^*$ which are adjacent to each other.

\subsubsection*{S-holomophicity and the imaginary part of 
the integral of the square}
The definition~\eqref{eq: s-holomorphicity} of 
s-holomorphicity of a function~$F \colon \dE \to \bC$
can be reinterpreted as follows.
For any corner~$c=(v,p)$,
the values~$F(z_1)$ and~$F(z_2)$ on 
both edges~$z_1 , z_2 \in \dE$ adjacent to~$v$ and~$p$ 
have the same projection to the
line $\sqrt{\ii/(v-p)} \, \bR$ in the complex plane, i.e.,
we may associate a well-defined value
to the function~$F$ at the corner~$c = (v,p)$ by
\begin{align*}
F(c) 
:= \; & \projtoline{\sqrt{\ii/(v-p)} \, \bR} \big( F(z_j) \big)
\, = \, \half \Big( F(z_j) 
        + \frac{\ii \, |v-p|}{v-p} \, \overline{F(z_j)} \Big) , 
\qquad \text{ for $j=1,2$.}
\end{align*}

For any s-holomophic function
\begin{align*}
F \colon \dE \to \bC
\end{align*}
there exists a function
\begin{align*}
H \colon \dV \cup \dV^* \to \bR ,
\end{align*}
defined uniquely up to an additive real constant by the condition that
for any vertex~$v$ and adjacent face~$p$, with $c=(v,p)$ the 
associated corner, we have
\begin{align}
\label{eq: definition of iis}
H(v) - H(p) := \; & \im \Big( 2 \, F(c)^2 \; (v-p) \Big)
    = \sqrt{2} \, \delta \, \left| F(c) \right|^2 \geq 0 ;
\end{align}
see~\cite[Proposition 3.6(i)]{CS-universality_in_Ising}.
The factor two in front of the squared value is included because of
the definition of the values on corners: if $u,u'$ are two 
adjacent vertices of faces and $z = \frac{u+u'}{2} \in \dE$ the 
edge between them, 
then a calculation~\cite[Proposition 3.6(ii)]{CS-universality_in_Ising}
from~\eqref{eq: definition of iis} shows
\begin{align}
\label{eq: iis increment}
H(u')-H(u) 
= \; & \im \Big( F(z)^2 \; (u'-u) \Big) .
\end{align}
We denote
\begin{align*}
H = \iis{F} ,
\end{align*}
and call~$H$ the ``imaginary part of the integral of the square'' of~$F$.
Note that from the definition~\eqref{eq: definition of iis} it
is clear that for a vertex~$v \in \dV$ and an adjacent
face~$p \in \dV^*$, we always have~$H(v) \geq H(p)$.

When~$F$ has Riemann boundary value~\eqref{eq:rbv_df}
on a boundary edge~$z$ between adjacent boundary vertices~$v,v'$, we 
see from~\eqref{eq: iis increment} that
$H(v)=H(v')$. Therefore Riemann boundary values for an s-holomorphic
function~$F$ imply that~$H = \iis{F}$ is constant on the boundary 
vertices of each boundary part. 
We then extend the definition to the boundary faces by the same constant,
and obtain a function
\begin{align*}
H \colon \dV \cup \dV^* \cup \bdry \dV^* \to \bR .
\end{align*}

\subsection{Sub- and superharmonicity and the maximum principle}
\label{ssec: maximum principle and uniqueness}

Unlike in the continuous case, 
the ``imaginary part of the integral of the square'' $H = \iis{F}$
of an s-holomorphic  function~$F$ is \emph{not}
(discrete) harmonic. The remarkable
observation is that it nevertheless
mimics the behavior of harmonic functions extremely well: 
its restriction~$H|_{\dV}$ to vertices is superharmonic,
and its restriction~$H|_{\dV^*}$ to faces is subharmonic,
and because of the boundary values, the values on vertices and faces
are suitably close.
The first incarnation of this almost harmonicity of~$H$ is the
following version of the maximum principle.

\begin{lem}[{\cite[Proposition~3.6(iii) and 
Lemma~3.14]{CS-universality_in_Ising}}]
\label{lem: maximum principle}
Let $F \colon \dE \to \bC$ be s-holomorphic with Riemann 
boundary values, and 
let~$H = \iis{F}$,
${H \colon \dV \cup \dV^* \cup \bdry \dV^* \to \bR}$ be defined
as above. Then we have:
\begin{itemize}
\item[(i)] At any interior vertex~$v$, the value of~$H$ is 
at most the simple arithmetic average of the values at the four 
neighboring vertices,
\begin{align*}
H(v) \leq \frac{1}{4} \Big( H(v+\delta) 
    + H(v+\ii\delta) + H(v-\delta) + H(v-\ii \delta) \Big) .
\end{align*}
In particular $H$ can not have a strict local maximum at 
an interior vertex.
\item[(ii)] At any interior face~$p$, the value of~$H$ is 
at least the weighted average of the values at the four 
neighboring faces,
\begin{align*}
H(p) \geq \frac{\sum_{s \in \set{\delta, \ii \delta, -\delta, -\ii \delta}} 
        \wgt (p+s) \, H(p+s)}
    {\sum_{s \in \set{\delta, \ii \delta, -\delta, -\ii \delta}} 
        \wgt (p+s)} ,
\end{align*}
where $\wgt (p') = 1$ for~$p' \in \dV^*$, and $\wgt (p') = 2(\sqrt{2}-1)$
for~$p' \in \bdry \dV^*$.
In particular $H$ can not have a strict local minimum at 
an interior face.
\end{itemize}
\end{lem}

With this version of the maximum principle
we can give the proof of
Lemma~\ref{lem:discrete-uniqueness}:
If a function~$f \in \dfunctionsp$ admits regular extensions
to the top, to the left leg, and to the right leg, then~$f \equiv 0$.
\begin{proof}[Proof of Lemma~\ref{lem:discrete-uniqueness}]
Suppose $f \in \dfunctionsp$ admits regular extensions to all
three extremities of the slit-strip.

Consider the s-holomorphic extension~$F \colon \dE(\dslitstrip) \to \bC$
of~$f$ to the lattice slit-strip, with Riemann boundary values.
The regular extensions assumption 
implies that $F$ decays exponentially in all three extremities:
the norm of its restrictions to horizontal cross-sections
decreases by at least a factor
$\eigval{-1/2}<1$ (resp. $\max\{\eigvalW{\widthL}{-1/2},
\eigvalW{\widthR}{-1/2} \}<1$) on each vertical step in the 
upwards direction (resp. downwards direction).
It follows that on horizontal lines in the top part,
the differences of all values of~$H = \iis{F}$ to the boundary values 
tend to zero:
\begin{align*}
\begin{cases} 
\max_{x \in \crosssecdual} \big|H(x + \ii y) - H(\lft + \ii y) \big| \to 0 \\
\max_{x \in \crosssecdual} \big|H(x + \ii y) - H(\rgt + \ii y) \big| \to 0
\end{cases}
\qquad \text{ as $y \to +\infty$.}
\end{align*}
Recalling that the boundary values~$H(\lft + \ii y), H(\rgt + \ii y)$
are constant (independent of~$y$), it follows first of all that
the values on the two boundaries
are equal, ${H(\lft + \ii y) = H(\rgt + \ii y) =: M}$, and furthermore 
that~$H(x+\ii y)$ approaches these boundary values~$M$ as~$y \to +\infty$.
Similarly, the values in the horizontal cross-sections of the left and 
right legs are tending to the same constant~$M$, and in particular
the boundary values on the slit part are also equal to~$M$.

It then follows from the maximum principle on vertices,
Lemma~\ref{lem: maximum principle}(i), 
that~$H |_{\dV} \leq M$. Similarly by the minimum
principle on faces,
Lemma~\ref{lem: maximum principle}(ii), we get~$H|_{\dV^*} \geq M$.
But as the values on vertices are at least the values on adjacent faces,
we get~$H \equiv M$. We thus
conclude that~$F \equiv 0$ and in particular~$f \equiv 0$.
\end{proof}

\subsection{Convergence of the distinguished functions in the slit-strip}
\label{ssec: precompactness and convergence}

Besides the maximum principle, we need more quantitative tools
of the regularity theory of s-holomorphic functions to prove the
scaling limit result for the distinguished functions in 
the slit-strip.

\subsubsection*{Beurling-type estimates}
We will need the following weak
Beurling-type estimates
(meaning that the exponent~$\beta$ does not have to be optimal)
for discrete harmonic measures
on vertices and faces
from~\cite{CS-discrete_complex_analysis_on_isoradial_graphs}.
\begin{prop}[{%
\cite[Proposition~2.11]{CS-discrete_complex_analysis_on_isoradial_graphs}}]
\label{prop: Beurling} 
There exists absolute constants ${\beta, \const > 0}$
such that the following holds.
Let $\bdry \dV$ denote the set of boundary vertices, 
and $\bdry \dV_0 \subset \bdry \dV$ a subset. Then the discrete 
harmonic function~$\omega_\circ \colon \dV \to [0,1]$ on vertices,
with boundary values~$0$ on~$\bdry \dV_0$
and~$1$ on~$\bdry \dV \setminus \bdry \dV_0$, satisfies
\begin{align*}
\omega_\circ (v) 
\leq \const \cdot 
\Big( \frac{\dist(v, \; \bdry \dV)}%
        {\dist(v, \; \bdry \dV \setminus \bdry \dV_0)} \Big)^\beta .
\end{align*}
Similarly, let $\bdry \dV^*$ denote the set of boundary faces, 
and $\bdry \dV^*_0 \subset \bdry \dV^*$ a subset. Then the discrete 
harmonic function (w.r.t. modified boundary face
weights as in Lemma~\ref{lem: maximum principle})
$\omega_\bullet \colon \dV^* \to [0,1]$ on faces
with boundary values~$0$ on~$\bdry \dV^*_0$
and~$1$ on~$\bdry \dV^* \setminus \bdry \dV^*_0$, satisfies
\begin{align*}
\omega_\bullet (p) 
\leq \const \cdot 
\Big( \frac{\dist(p, \; \bdry \dV^*)}%
        {\dist(p, \; \bdry \dV^* \setminus \bdry \dV^*_0)} \Big)^\beta .
\end{align*}
\end{prop}

\subsubsection*{Precompactness estimates}
The following result 
from~\cite{CS-universality_in_Ising}
yields uniform boundedness 
and equicontinuity (i.e., precompactness in
the Arzel\`a-Ascoli sense)
for both~$F$ and~$H$, given control on~$|H|$.
\begin{thm}[{\cite[Theorem~3.12]{CS-universality_in_Ising}}]
\label{thm: precompactness}
There exists absolute constants (independent of lattice mesh~$\delta$
and lattice domains) such that the following estimates hold.
Let $F \colon \dE \to \bC$ be s-holomorphic with Riemann 
boundary values, and let~$H = \iis{F}$,
$H \colon \dV \cup \dV^* \to \bR$ be defined
as above. Suppose that the ball~$B_r(x_0)$, with
$r \geq \const \cdot \delta$,
is contained in the lattice domain and does not intersect its
boundary. Then for $z,z' \in \dE$ adjacent edges
contained in the smaller ball~$B_{r/2}(x_0)$, we have
\begin{align}\label{eq:precompactness}
\left| F(z)\right| 
\leq \; & \const \cdot 
        \sqrt{\frac{\max_{B_r(x_0)} \left| H \right| }{r}}, \\ 
\nonumber
\frac{\left| F(z')-F(z) \right|}{\delta}
\leq \; & \const \cdot 
        \sqrt{\frac{\max_{B_r(x_0)} \left| H\right| }{r^3}} .
\end{align}
\end{thm}

\subsubsection*{Limit result for the slit-strip functions}
With the above tools, we can prove the scaling limit result
for the pole functions.

\begin{thm}\label{thm: convergence of slit strip functions}
Choose sequences $(\lft_{n})_{n \in \bN}$,
$(\rgt_{n})_{n \in \bN}$ of integers
$\lft_{n}, \rgt_{n} \in \bZ$ such that
\begin{itemize}
\item $\lft_{n} < 0 < \rgt_{n}$ for all $n$;
\item $\width_{n} := \rgt_{n} - \lft_{n} \to +\infty$ as $n \to \infty$;
\item $\lft_{n} / \width_{n} \to -\half$
and $\rgt_{n} / \width_{n} \to +\half$ as $n \to \infty$.
\end{itemize}
For~$k \in \poshalfint$, 
let~$\PoleT{k}^{(\width_{n})}, \PoleL{k}^{(\width_{n})}, 
\PoleR{k}^{(\width_{n})}$ 
denote the functions of 
Section~\ref{sub: functions in lattice strip}
in the lattice strips with~$\lft = \lft_n$ and $\rgt = \rgt_n$,
and let~$\cPoleT{k}, \cPoleL{k}, \cPoleR{k} \colon \cslitstrip \to \bC$
denote the pure pole functions of
Proposition~\ref{prop: pure pole functions}.
Then, as~$n \to \infty$, we have 
\begin{align*}
\sqrt{\width_n} \; \PoleT{k}^{(\width_{n})} \big( z \width_n \big) 
\to \; & \cPoleT{k}(z) & &
    \text{ uniformly on compact subsets of~$\cslitstrip \ni z$} \\ 
\sqrt{\width_n} \; \PoleL{k}^{(\width_{n})} \big( z \width_n \big) 
\to \; & \cPoleL{k}(z) & &
    \text{ uniformly on compact subsets of~$\cslitstrip \ni z$} \\ 
\sqrt{\width_n} \; \PoleR{k}^{(\width_{n})} \big( z \width_n \big) 
\to \; & \cPoleR{k}(z) & &
    \text{ uniformly on compact subsets of~$\cslitstrip \ni z$} .
\end{align*}
\end{thm}
\begin{proof}
Let us consider the convergence of the left leg pole functions~---
the other cases are similar. 
Including the rescalings in the definition, let us
define functions~$F^{(n)}$ 
on~$\frac{1}{\width_n} \dslitstripW{\width_n}$ 
by 
the formula~$F^{(n)}(z) = \sqrt{\width_n} \; 
\PoleL{k}^{(\width_{n})} \big( z \width_n \big)$.
Define also~$H^{(n)} = \iis{F^{(n)}}$, with the additive constant 
chosen so that this function vanishes at the tip of the
slit, $H^{(n)}(0) = 0$. Then $H^{(n)}$ vanishes on the entire slit
(it is constant on boundary components), and 
since~$F^{(n)}$ decays exponentially in the top and the right leg
extremities, the same argument as in the proof of 
Lemma~\ref{lem:discrete-uniqueness} shows that the boundary values
of~$H^{(n)}$ also on the left and right boundaries are zero,
and that~$H^{(n)}$ tends to zero in the top and right extremities.

For $K>0$, consider a horizontal cross-cut~$L_K^{(n)}$ of the
left leg of~$\frac{1}{\width_n} \dslitstripW{\width_n}$ 
at imaginary part~$-K$ (more precisely, the horizontal line
of the lattice~$\frac{1}{\width_n} \dslitstripW{\width_n}$
with largest imaginary part below~$-K$), and consider
the truncated 
slit-strip~$S_K^{(n)} \subset \frac{1}{\width_n} \dslitstripW{\width_n}$
defined as the component of the complement of this cross-section
which contains the top and right extremities.
Define $M^{(n)}_K := \max_{L_K^{(n)}} \big| H^{(n)} \big|$,
where~$L_K^{(n)}$ is interpreted to consist of
both the vertices on the horizontal line and the
faces just below that horizontal line.
By the maximum principle, Lemma~\ref{lem: maximum principle}, we have
$M^{(n)}_K = \max_{S_K^{(n)}} \big| H^{(n)} \big|$,
so in view of $S_K^{(n)} \subset S_{K'}^{(n)}$ for $K < K'$
we have that $M^{(n)}_K$ is increasing in~$K$.

We will later prove that for any~$K>0$,
the sequence~$(M^{(n)}_K)_{n \in \bN}$ is bounded, i.e., we have
$M_K := \sup_{n \in \bN} M^{(n)}_K < \infty$. 
We now first prove the convergence of~$F^{(n)}$ to~$\cPoleL{k}$
assuming this.
Denote by~$L_K$ the horizontal cross-cut of the left leg of the continuum 
slit-strip~$\cslitstrip$ at imaginary part~$-K$, and 
by~$S_K \subset \cslitstrip$ the
component of~$\cslitstrip \setminus L_K$ which contains 
the top and right extremities.
Then by the boundedness of~$(M^{(n)}_K)_{n \in \bN}$,
the functions~$H^{(n)}$ restricted to the 
part~$S_K$
are uniformly bounded, and
therefore as a consequence 
of Theorem~\ref{thm: precompactness},
both $H^{(n)}$ and~$F^{(n)}$ are 
equicontinuous and uniformly bounded on compact
subsets of~$S_K$.
By the Arzel\`a-Ascoli theorem, along a subsequence we have uniform
convergence of~$H^{(n)}$ and~$F^{(n)}$ on compact subsets of~$S_K$,
and since this holds for any~$K>0$, by diagonal extraction
there exists a subsequence along which $H^{(n)}$ and~$F^{(n)}$ converge
uniformly on compact subsets of the whole slit-strip~$\cslitstrip$.
We must show that in any such subsequential limit~$(H,F)$ we
in fact have~$F=\cPoleL{k}$.

Note that in such a subsequential limit we 
have~$H=\im \Big( \int F(z)^2 \, \ud z \Big)$, and
as a locally uniform limit of both subharmonic and
superharmonic functions ($H^{(n)}$ on vertices and faces),
$H$ is harmonic. It follows that~$F^2 = 2 \ii \, \partial_z H$
is holomorphic, and thus~$F$ is also holomorphic.
By Lemma~\ref{lem: maximum principle}, $H^{(n)}$ is bounded above
by $M_K \, \omega_\circ^{(n)}$, where $\omega_\circ^{(n)}$ is the 
discrete harmonic measure on the vertices of~$S_K^{(n)}$ of the 
cross-cut~$L_K^{(n)}$. Similarly $H^{(n)}$ is bounded below
by $-M_K \, \omega_\bullet^{(n)}$, where $\omega_\bullet^{(n)}$ is the 
discrete harmonic measure on the faces of~$S_K^{(n)}$ just below
the cross-cut~$L_K^{(n)}$ (with the modified boundary weights).
By Beurling estimates, Proposition~\ref{prop: Beurling},
these harmonic measures decay at the top and right 
extremities uniformly in~$n$, so the subsequential limit~$H$ 
of the~$H^{(n)}$ decays at the at the top and right 
extremities. On $S_{K-\eps}$, for any~$\eps>0$,
these harmonic measures also 
decay uniformly upon approaching the boundaries of the slit-strip,
again by virtue of the Beurling estimates.
We conclude that $H$ also tends to zero on the boundary.
Then also~$F$ decays at top and right extremities, 
by Theorem~\ref{thm: precompactness},
and $F$ has Riemann boundary values
by \cite[Remark~6.3]{CS-universality_in_Ising}. In order to conclude
that~$F=\cPoleL{k}$, it remains to show that
$F - \ccfFunL_{-k}$ decays in the left leg extremity.

By definition of the discrete pole function~$\PoleL{k}$,
in the left leg we can write
\begin{align*}
F^{(n)}(z) = F_0^{(n)}(z) + \sqrt{\width_n} \, \eigFL{-k}(z \width_n) ,
\end{align*}
where~$F_0^{(n)}$ decays in the left leg extremity.
From Theorem~\ref{thm: convergence of strip functions}, we already
know that the second term on
the right hand side above converges to~$\ccfFunL_{-k}$,
uniformly on compact subsets. 
To control~$F_0^{(n)}$, consider $H_0^{(n)} = \iis{F_0^{(n)}}$.
By similar arguments as above for~$H^{(n)}$, one shows that
one can extract subsequences 
from~$(H_0^{(n)},F_0^{(n)})$ which converge uniformly on compacts
in the left leg, and that for any subsequential limit~$(H_0,F_0)$
we have that $F_0$ decays in the left leg extremity.
But such an~$F_0(z)$ is, as the limit of
$F^{(n)}(z) - \sqrt{\width_n} \, \eigFL{-k}(z \width_n)$,
equal to $F(z) - \ccfFunL_{-k}(z)$. We have thus seen that~$F$
is holomorphic with Riemann boundary values in 
the slit-strip, 
$F$ decays at the top and right leg extremities, 
and $F- \ccfFunL_{-k}$ decays at the left leg extremity.
We conclude that~$F = \cPoleL{k}$.

To finish the proof, we must still show the boundedness
of~$(M^{(n)}_K)_{n \in \bN}$. Suppose that instead
$M^{(n)}_K \to \infty$ along some subsequence, for
some~$K>0$ and therefore by monotonicity for all large enough~$K$.
Now $(M^{(n)}_K)^{-1} \, \big| H^{(n)} \big|$
is bounded by~$1$ on~$L_K^{(n)}$,
and by selecting a large enough~$K$ we get by
Beurling estimates as before that
$(M^{(n)}_K)^{-1} \, \big| H^{(n)} \big| \leq \frac{1}{5}$ 
on~$L_0^{(n)}$. Again decompose
$F^{(n)}(z) 
= F_0^{(n)}(z) + \sqrt{\width_n} \, \eigFL{-k}(z \width_n)$
in the left leg, and denote by~$H_0^{(n)}$ the imaginary 
part of the integral of the square of~$F_0^{(n)}$. 
Noting that as~$n \to \infty$ along the subsequence, we have
$(M^{(n)}_K)^{-1/2} \, \sqrt{\width_n} \,\big| \eigFL{-k}(z \width_n) \big|
\to 0$ uniformly on~$S_{K+1} \ni z$, we see that
$(M^{(n)}_K)^{-1} \, \big| H_0^{(n)} \big| \leq \frac{2}{5}$ 
on~$L_0^{(n)}$ for large enough~$n$. By the decay in the left leg
and the maximum principle, $H_0^{(n)}$ is bounded by its values 
on~$L_0^{(n)}$, so
$(M^{(n)}_K)^{-1} \, \big| H_0^{(n)} \big| \leq \frac{2}{5}$
everywhere. But similarly by the smallness of
$(M^{(n)}_K)^{-1/2} \, \sqrt{\width_n} \,\big| \eigFL{-k}(z \width_n) \big|$ 
on~$S_{K+1} \ni z$, we see that the difference~$H^{(n)}-H_0^{(n)}$ 
is small, to that in particular
$(M^{(n)}_K)^{-1} \, \big| H^{(n)} \big| \leq \frac{3}{5}$
on~$L_K^{(n)}$, for large~$n$ in the subsequence.
This is a contradiction with the definition
of~$M^{(n)}_K$, so indeed~$(M^{(n)}_K)_{n \in \bN}$ 
had to be bounded and the proof is complete.
\end{proof}

\subsubsection*{Convergence of inner products}

For applications to the convergence of the Ising model fusion
coefficients, it is not enough for us to have the uniform
convergence on compacts of the distinguished discrete functions
to the distinguished continuum ones. 
We need the convergence of the inner products of their 
restrictions to the cross-section as well.

\begin{cor}\label{cor:ip-convergence}
Choose sequences $(\lft_{n})_{n \in \bN}$,
$(\rgt_{n})_{n \in \bN}$ of integers
$\lft_{n}, \rgt_{n} \in \bZ$ such that
\begin{itemize}
\item $\lft_{n} < 0 < \rgt_{n}$ for all $n$;
\item $\width_{n} := \rgt_{n} - \lft_{n} \to +\infty$ as $n \to \infty$;
\item $\lft_{n} / \width_{n} \to -\half$
and $\rgt_{n} / \width_{n} \to +\half$ as $n \to \infty$.
\end{itemize}
For~$k \in \poshalfint$, 
let
\begin{align*}
\poleT{k}^{(\width_{n})} , \; 
\poleL{k}^{(\width_{n})} , \;
\poleR{k}^{(\width_{n})} , \;
\eigf{\pm k}^{(\width_{n})} , \; 
\eigfL{\pm k}^{(\width_{n})} , \;
\eigfR{\pm k}^{(\width_{n})} 
    \; \in \; \dfunctionspW{\width_{n}}
\end{align*}
denote the functions defined before
in the lattice strips with~$\lft = \lft_n$ and $\rgt = \rgt_n$.
Correspondingly, let
\begin{align*}
\cpoleT{k} , \; 
\cpoleL{k} , \;
\cpoleR{k} , \;
\ccffun_{\pm k} , \; 
\ccffunL_{\pm k} , \;
\ccffunR_{\pm k}
    \; \in \; \cfunctionsp
\end{align*}
be the continuum functions defined before.

Then as $n \to \infty$, we have the convergence of all inner products
in~$\dfunctionspW{\width_{n}}$ to the corresponding ones in~$\cfunctionsp$:
\begin{align*}
\innprod{\eigfX{k}^{(\width_{n})}}{\eigfXbis{k'}^{(\width_{n})}} 
\to \; & \innprod{\ccffunX_{k}}{\ccffunXbis_{k'}}
    & & \text{ for 
        $\wildsym , \wildsymbis \in \set{\topsym, \lftsym, \rgtsym}$
        and $k , k' \in \pm\poshalfint$,} \\
\innprod{\poleX{k}^{(\width_{n})}}{\eigfXbis{k'}^{(\width_{n})}} 
\to \; & \innprod{\cpoleX{k}}{\ccffunXbis_{k'}}
    & & \text{ for 
        $\wildsym , \wildsymbis \in \set{\topsym, \lftsym, \rgtsym}$
        and $k \in \poshalfint$, $ k' \in \pm\poshalfint$,} 
\\
\innprod{\poleX{k}^{(\width_{n})}}{\poleXbis{k'}^{(\width_{n})}} 
\to \; & \innprod{\cpoleX{k}}{\cpoleXbis{k'}}
    & & \text{ for 
        $\wildsym , \wildsymbis \in \set{\topsym, \lftsym, \rgtsym}$
        and $k , k' \in \poshalfint$}
\end{align*}
(where the hitherto 
undefined notations are 
interpreted so that~$\eigfT{k} = \eigf{k}$ 
and~$\ccffunT_{k}=\ccffun_{k}$).
\end{cor}
\begin{proof}
The proofs of all cases are similar, so we will only consider
in detail a typical one,
\begin{align*}
\innprod{\poleL{k}}{\eigf{k'}} \to \innprod{\cpoleL{k}}{\ccffun_{k'}} .
\end{align*}
We again work in the rescaled 
slit-strip~$\frac{1}{\width_n}\dslitstrip$, and now 
use the functions
$x \mapsto \sqrt{\width_n} \, \poleL{k}^{(\width_n)}(x \width_n)$, 
${x \mapsto {\sqrt{\width_n} \, \eigf{k'}^{(\width_n)}(x \width_n)}}$
with piecewise constant interpolation for convenience
(by the equicontinuity estimates, this does not change the 
convergence statements).
The discrete inner product can be written as the integral
of the piecewise constant interpolation
\begin{align*}
\innprod{\poleL{k}^{(\width_n)}}
{\eigf{k'}^{(\width_n)}}
= \; & \re \left( \sum_{x' \in \dintervaldual{\lft_n}{\rgt_n}}
    \poleL{k}^{(\width_n)}(x') \; 
    \overline{\eigf{k'}^{(\width_n)}(x')} \right) \\
= \; & \re \left( \int_{\lft_n/\width_n}^{\rgt_n/\width_n}
    \sqrt{\width_n} \, \poleL{k}^{(\width_n)}(x \width_n) \;
    \overline{\sqrt{\width_n} \, \eigf{k'}^{(\width_n)}(x \width_n)}
    \; \ud x \right) .
\end{align*}

For small~$\epsilon>0$, let
$\ccrosssec^\epsilon := 
[-\half + \epsilon, -\epsilon] \cup [\epsilon, \half-\epsilon]$.
By Theorem~\ref{thm: convergence of strip functions},
we have
\begin{align*}
\sqrt{\width_n} \, \eigf{k'}^{(\width_n)}(x \width_n) 
\to \; & \ccffun_{k'}(x) & 
\text{ uniformly on~$\ccrosssec = \Big[-\half,\half\Big] \ni x$,}
\end{align*}
and by Theorem~\ref{thm: convergence of slit strip functions}, 
we have
\begin{align*}
\sqrt{\width_n} \, \poleL{k}^{(\width_n)}(x \width_n)
\to \; & \cpoleL{k}(x) &
\text{ uniformly on~$\ccrosssec^\epsilon \ni x$.}
\end{align*}
Comparing the discrete inner product with the continuum inner product
\begin{align*}
\innprod{\cpoleL{k}}{\ccffun_{k'}}
= \; & \re \left( \int_{-1/2}^{+1/2}
    \cpoleL{k}(x) \;
    \overline{\ccffun_{k'}(x)}
    \; \ud x \right)
\end{align*}
the uniform convergence shows that the contributions to the integrals
from~$\ccrosssec^\eps$ converge to the desired ones for any~$\epsilon>0$,
and it remains to show that
the contributions from within distance~$\epsilon$ to
the points~$-\half, 0 , \half$ 
are negligible in the limit~$\epsilon \to 0$.

Note that~$|\ccffun_{k'}(x)| \leq 1$ for all~$x \in \ccrosssec$,
and its discrete  
counterpart~$\sqrt{\width_n} \, \eigf{k'}^{(\width_n)}(x \width_n)$
is bounded by an absolute constant, too. It therefore remains to
control the discrete pole functions
$
\sqrt{\width_n} \; 
    \poleL{k}^{(\width_{n})} \big( x \width_n \big)$
and their continuous counterparts. 
But in the proof of 
Theorem~\ref{thm: convergence of slit strip functions}
we saw by Beurling estimates that
the imaginary part of the integral of the square
$H^{(n)} = \iis{F^{(n)}}$ of the pole function
$F^{(n)}(z) = \sqrt{\width_n} \; 
    \PoleL{k}^{(\width_{n})} \big( z \width_n \big)$
tends to zero near the boundaries. Combined with~\eqref{eq:precompactness},
we get
$\big| \sqrt{\width_n} \; 
         \poleL{k}^{(\width_{n})} \big( x \width_n \big) \big|
= \OO(|x - c|^{(\beta-1)/2})$ for $c \in \set{-\half, 0, \half}$.
With these estimates of the two types of functions to
be integrated, we see that the contribution to the
integrals from within distance~$\epsilon$ to
the points~$-\half, 0 , \half$ is in this case~$\OO(\epsilon^{(1+\beta)/2})$.
This proves the desired convergence
$\innprod{\poleL{k}}{\eigf{k'}} \to \innprod{\cpoleL{k}}{\ccffun_{k'}}$.

Note that among the many cases in the statement, the above type of 
reasoning results in the worst bounds for two pole type functions
(no a priori bounds besides the Beurling estimates are available for 
either factors). But even in that case the product in the integrand is
$\OO(|x - c|^{\beta-1})$ for $c \in \set{-\half, 0, \half}$, and the
contributions to the 
integrals from within distance~$\epsilon$ to
the points~$-\half, 0 , \half$ are~$\OO(\epsilon^\beta)$
as~$\epsilon \to 0$,
which is sufficient for the convergence of the inner products.
\end{proof}

\bigskip

\appendix

{\bf Acknowledgments:}
S.P. is supported by a KIAS Individual Grant (MG077201) at Korea Institute 
for Advanced Study. T.A. was affiliated with the Department of Mathematics and Statistics at the American University of Sharjah, and recognizes their financial support.

\end{document}